\documentclass{article}
\usepackage{amsmath,amssymb,amsthm,graphics,setspace,xcolor,colortbl,curves}
\usepackage[margin=1in]{geometry}
\usepackage[longnamesfirst]{natbib}
\allowdisplaybreaks
\newcommand{\sh}{\cellcolor{gray!33}}
\newtheorem{lemma}{Lemma}
\newtheorem{theorem}{Theorem}
\newtheorem{assumption}{Assumption}
\title{Slow Movers in Panel Data\thanks{
We have benefited from very useful comments by Cheng Hsiao, M. Hashem Pesaran, Valentin Verdier, and participants at the University of North Carolina at Chapel Hill and the University of Southern California. We are responsible for all the remaining errors.}
} 
\author{Yuya Sasaki\thanks{Department of Economics, Vanderbilt University, 2301 Vanderbilt Place, Nashville, TN 37235; Email: yuya.sasaki@vanderbilt.edu} \and Takuya Ura\thanks{Department of Economics, University of California, Davis, One Shields Avenue, Davis, CA 95616; Email: takura@ucdavis.edu}}
\begin{document}
\date{}
\maketitle
\begin{abstract}
Panel data often contain stayers (units with no within-variations) and \textit{slow movers} (units with little within-variations). In the presence of many slow movers, conventional econometric methods can fail to work. We propose a novel method of robust inference for the average partial effects in correlated random coefficient models robustly across various distributions of within-variations, including the cases with many stayers and/or many slow movers in a unified manner. In addition to this robustness property, our proposed method entails smaller biases and hence improves accuracy in inference compared to existing alternatives. Simulation studies demonstrate our theoretical claims about these properties:  the conventional 95\% confidence interval covers the true parameter value with 37-93\% frequencies, whereas our proposed one achieves 93-96\% coverage frequencies. 
\begin{description} 
\item[] Keywords: bias reduction, correlated random coefficient, panel data, robustness, slow movers.
\item[] JEL Classification Codes:  C14, C23, C33
\end{description}
\end{abstract}

\section{Introduction}

Panel datasets used in economics often contain many individuals/families/firms with no or little within-variations. 
We refer to individuals with no within-variation as stayers, and individuals with little within-variation as \textit{slow movers}.\footnote{We thank Hashem Pesaran for providing us with the terminology of slow movers. We define the terms in Section \ref{sec:overview}.} 
Figure \ref{fig:PSID_17_19} illustrates an example for the log total family income in the U.S. Panel Survey of Income Dynamics (PSID).
In this figure, about 3.19\% of the sample are stayers, and the large spike of the kernel density plot also indicates that there are many slow movers.\footnote{This kernel density plot does not use the information about the mass of $D=0$.}
Many slow movers appear in various fields of economics, e.g., labor economics, international trade, and political economy. 

\begin{figure}
	\centering
	\scalebox{0.20}{
		\includegraphics{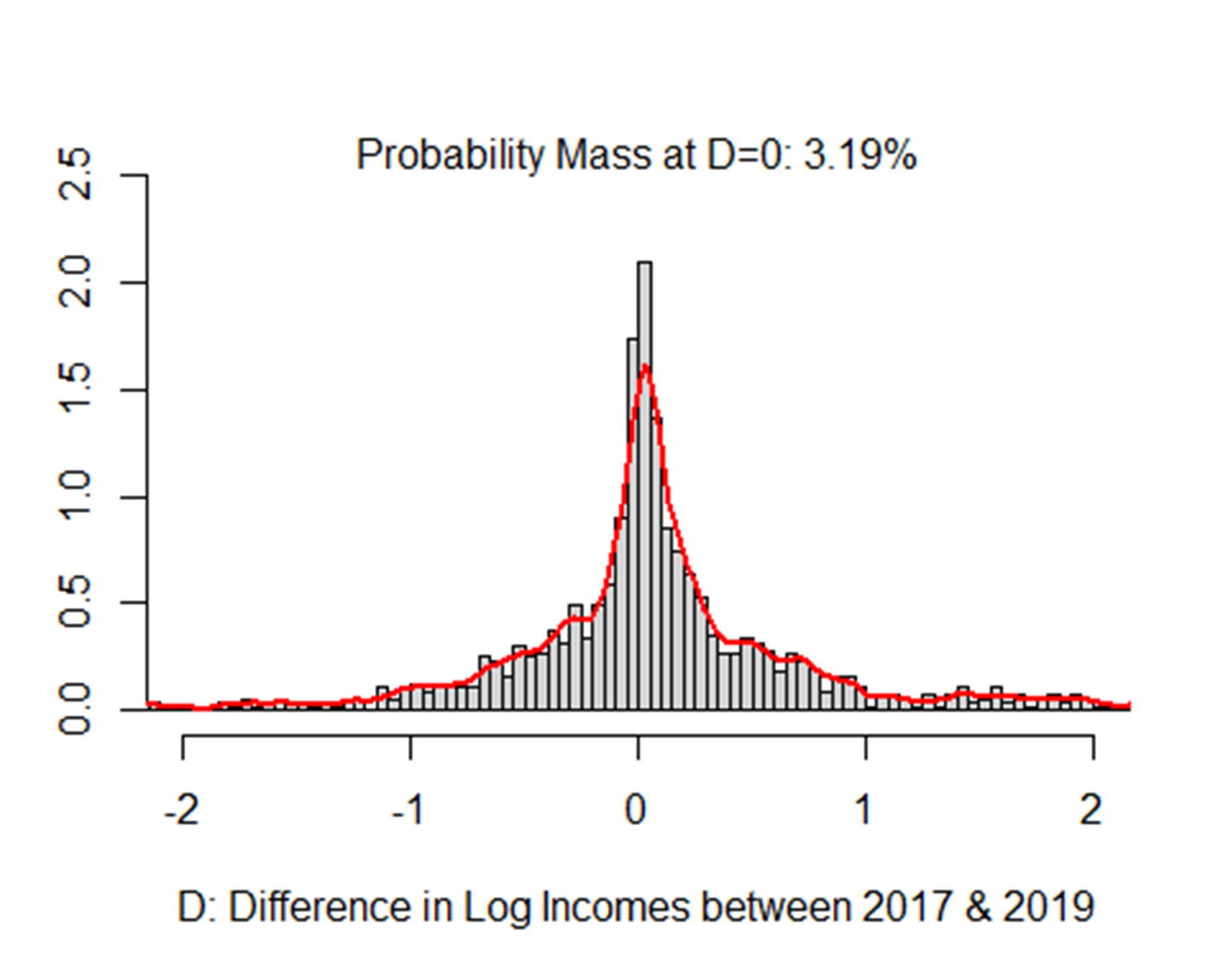}
	}
	\caption{Empirical distribution of the difference $D$ in log total family incomes between 2017 and 2019 in the U.S. Panel Survey of Income Dynamics (PSID). The chart displays a histogram and a kernel density plot. About 3.19\% of the sample has the exactly zero difference ($D=0$). The displayed kernel density plot does not use the information about the mass with $D=0$, and it therefore represents the continuous part of the distribution of $D$. The sample consists of 1,127 households with positive total family income less than or equal to 25K U.S. dollars in both 2017 and 2019.}
	\label{fig:PSID_17_19}
\end{figure}

In this paper, we highlight issues concerning slow movers in the context of correlated random coefficient models in short panels \citep{chamberlain:1992}. 
\cite{chamberlain:1992} establishes a $\sqrt{N}$-consistent estimation of the average partial effects in regular cases. 
\citet*{graham/powell:2012} further investigate the identification and estimation for the average partial effects under more general settings in the presence of slow movers. 
In their analysis, slow movers play two crucial roles. 
The existence of slow movers can allow us to identify the time trends in the average partial effects. 
On the other hand, their existence can at the same time also induce irregular identification of important parameters such as the average partial effects. 

Equipped with technologies from the literature on heavy-tailed distributions, we revisit the issues regarding slow movers. 
Notably, we provide a novel method of inference that is robustly valid across various distributional patterns concerning stayers and slow movers. 
Our proposed inference is valid for a large class of distributions with a possible probability mass of stayers (i.e., many stayers) and a possibly unbounded density of within-variations (i.e., many  slow movers). 
The large class includes the case in Figure \ref{fig:PSID_17_19}, as well as all the eight cases illustrated in Figure \ref{fig:eight_cases} (a)--(h), among others. 
In the existing literature, the method of \citet[][Sec. 2]{graham/powell:2012} focuses on the case of Figure \ref{fig:eight_cases} (a) -- see Section \ref{sec:case_determinant_zero} ahead for details -- and the method of \citet[][Sec. 3.2]{graham/powell:2012} focuses on the case of Figure \ref{fig:eight_cases} (b) -- see Section \ref{sec:case_determinant_nonzero} ahead for details. 
The existing literature, including \cite{graham/powell:2012}, does not consider the other six cases, (c)--(h), in Figure \ref{fig:eight_cases}. 
The method proposed in this paper, on the other hand, covers all of the eight cases in Figure \ref{fig:eight_cases} (a)--(h) in a unified manner.
We use Karamata's theorem \citep[Theorem 2.1]{resnick2007heavy} from the literature on heavy-tailed distributions as a main technical device to analyze such a possible unboundedness, local asymmetry, or non-smoothness of the density.   
Extensive simulation studies demonstrate that our proposed approach significantly improves both estimation bias and coverage frequencies compared to existing methods, as formally predicted by our theory.
In our simulation studies, our proposed 95\% confidence interval covers the true parameter value with 93-96\% frequencies, whereas the conventional method provides 37-93\% coverage frequencies. 

\begin{figure}[htp]
\centering
\setlength{\unitlength}{0.6cm}
\begin{tabular}{cc}
\begin{picture}(10,7)
	\put(0,6.5){(a)}
	\multiput(1,1)(1,0){1}{\vector(0,1){5}}
	\multiput(1,1)(1,0){1}{\vector(1,0){8}}
	\put(9,0.5){$D$}
	
	\linethickness{0.25mm}
	\qbezier(1.5,2)(3,3)(5,3)
	\qbezier(5,3)(7,3)(8.5,2)
	\put(8.6,2){$f_D$}

	\linethickness{0.1mm}
	\curvedashes[1.0mm]{1,1}
	\qbezier(5,1)(5,2.5)(5,6)
	\put(4.9,0.5){$0$}
\end{picture}
&
\begin{picture}(10,7)
	\put(0,6.5){(b)}
	\multiput(1,1)(1,0){1}{\vector(0,1){5}}
	\multiput(1,1)(1,0){1}{\vector(1,0){8}}
	\put(9,0.5){$D$}
	
	\linethickness{0.25mm}
	\qbezier(1.5,2)(3,3)(5,3)
	\qbezier(5,3)(7,3)(8.5,2)
	\put(8.6,2){$f_D$}
	
	\put(5,5.2){\circle*{0.25}}
	\put(5.25,5){$P(D=0)>0$}

	\linethickness{0.1mm}
	\curvedashes[1.0mm]{1,1}
	\qbezier(5,1)(5,2.5)(5,6)
	\put(4.9,0.5){$0$}
\end{picture}
\\
\begin{picture}(10,7)
	\put(0,6.5){(c)}
	\multiput(1,1)(1,0){1}{\vector(0,1){5}}
	\multiput(1,1)(1,0){1}{\vector(1,0){8}}
	\put(9,0.5){$D$}
	
	\linethickness{0.25mm}
	\qbezier(1.5,1.5)(4,1.5)(5,4)
	\qbezier(5,4)(6,1.5)(8.5,1.5)
	\put(8.6,1.5){$f_D$}

	\linethickness{0.1mm}
	\curvedashes[1.0mm]{1,1}
	\qbezier(5,1)(5,2.5)(5,6)
	\put(4.9,0.5){$0$}
\end{picture}
&
\begin{picture}(10,7)
	\put(0,6.5){(d)}
	\multiput(1,1)(1,0){1}{\vector(0,1){5}}
	\multiput(1,1)(1,0){1}{\vector(1,0){8}}
	\put(9,0.5){$D$}
	
	\linethickness{0.25mm}
	\qbezier(1.5,1.5)(4,1.5)(5,4)
	\qbezier(5,4)(6,1.5)(8.5,1.5)
	\put(8.6,1.5){$f_D$}
	
	\put(5,5.2){\circle*{0.25}}
	\put(5.25,5){$P(D=0)>0$}

	\linethickness{0.1mm}
	\curvedashes[1.0mm]{1,1}
	\qbezier(5,1)(5,2.5)(5,6)
	\put(4.9,0.5){$0$}
\end{picture}
\\
\begin{picture}(10,7)
	\put(0,6.5){(e)}
	\multiput(1,1)(1,0){1}{\vector(0,1){5}}
	\multiput(1,1)(1,0){1}{\vector(1,0){8}}
	\put(9,0.5){$D$}
	
	\linethickness{0.25mm}
	\qbezier(1.5,1.5)(4,1.5)(5,3)
	\qbezier(5,4)(6,1.5)(8.5,1.5)
	\put(8.6,1.5){$f_D$}

	\linethickness{0.1mm}
	\curvedashes[1.0mm]{1,1}
	\qbezier(5,1)(5,2.5)(5,6)
	\put(4.9,0.5){$0$}
\end{picture}
&
\begin{picture}(10,7)
	\put(0,6.5){(f)}
	\multiput(1,1)(1,0){1}{\vector(0,1){5}}
	\multiput(1,1)(1,0){1}{\vector(1,0){8}}
	\put(9,0.5){$D$}
	
	\linethickness{0.25mm}
	\qbezier(1.5,1.5)(4,1.5)(5,3)
	\qbezier(5,4)(6,1.5)(8.5,1.5)
	\put(8.6,1.5){$f_D$}
	
	\put(5,5.2){\circle*{0.25}}
	\put(5.25,5){$P(D=0)>0$}

	\linethickness{0.1mm}
	\curvedashes[1.0mm]{1,1}
	\qbezier(5,1)(5,2.5)(5,6)
	\put(4.9,0.5){$0$}
\end{picture}
\\
\begin{picture}(10,7)
	\put(0,6.5){(g)}
	\multiput(1,1)(1,0){1}{\vector(0,1){5}}
	\multiput(1,1)(1,0){1}{\vector(1,0){8}}
	\put(9,0.5){$D$}
	
	\linethickness{0.25mm}
	\qbezier(1.5,1.1)(4.5,1.1)(4.8,6)
	\qbezier(5.2,6)(5.5,1.1)(8.5,1.1)
	\put(8.6,1.2){$f_D$}

	\linethickness{0.1mm}
	\curvedashes[1.0mm]{1,1}
	\qbezier(5,1)(5,2.5)(5,6)
	\put(4.9,0.5){$0$}
\end{picture}
&
\begin{picture}(10,7)
	\put(0,6.5){(h)}
	\multiput(1,1)(1,0){1}{\vector(0,1){5}}
	\multiput(1,1)(1,0){1}{\vector(1,0){8}}
	\put(9,0.5){$D$}
	
	\linethickness{0.25mm}
	\qbezier(1.5,1.1)(4.5,1.1)(4.8,6)
	\qbezier(5.2,6)(5.5,1.1)(8.5,1.1)
	\put(8.6,1.2){$f_D$}
	
	\put(5,5.2){\circle*{0.25}}
	\put(5.75,5.5){$P(D=0)>0$}

	\linethickness{0.1mm}
	\qbezier(5,5.2)(5.35,5.45)(5.7,5.7)
	\curvedashes[1.0mm]{1,1}
	\qbezier(5,1)(5,2.5)(5,6)
	\put(4.9,0.5){$0$}
\end{picture}
\end{tabular}
\caption{Illustration of eight distribution models to which the method of robust inference proposed in this paper robustly applies in a unified manner: 
(a) a smooth, bounded, and locally symmetric distribution of $D$ with no mass of $D=0$; 
(b) a smooth, bounded, and locally symmetric distribution of $D$ with a mass of $D=0$;
(c) a non-smooth, bounded, and locally symmetric distribution of $D$ with no mass of $D=0$; 
(d) a non-smooth, bounded, and locally symmetric distribution of $D$ with a mass of $D=0$;
(e) a non-smooth, bounded, and locally asymmetric distribution of $D$ with no mass of $D=0$; 
(f) a non-smooth, bounded, and locally asymmetric distribution of $D$ with a mass of $D=0$;
(g) an unbounded distribution of $D$ with no mass of $D=0$; 
and
(h) an unbounded distribution of $D$ with a mass of $D=0$.
Densities of the absolutely continuous part of the distribution is represented by solid curves, while the probability mass of the singular part is represented by a black circle.
}
\label{fig:eight_cases}
\end{figure}
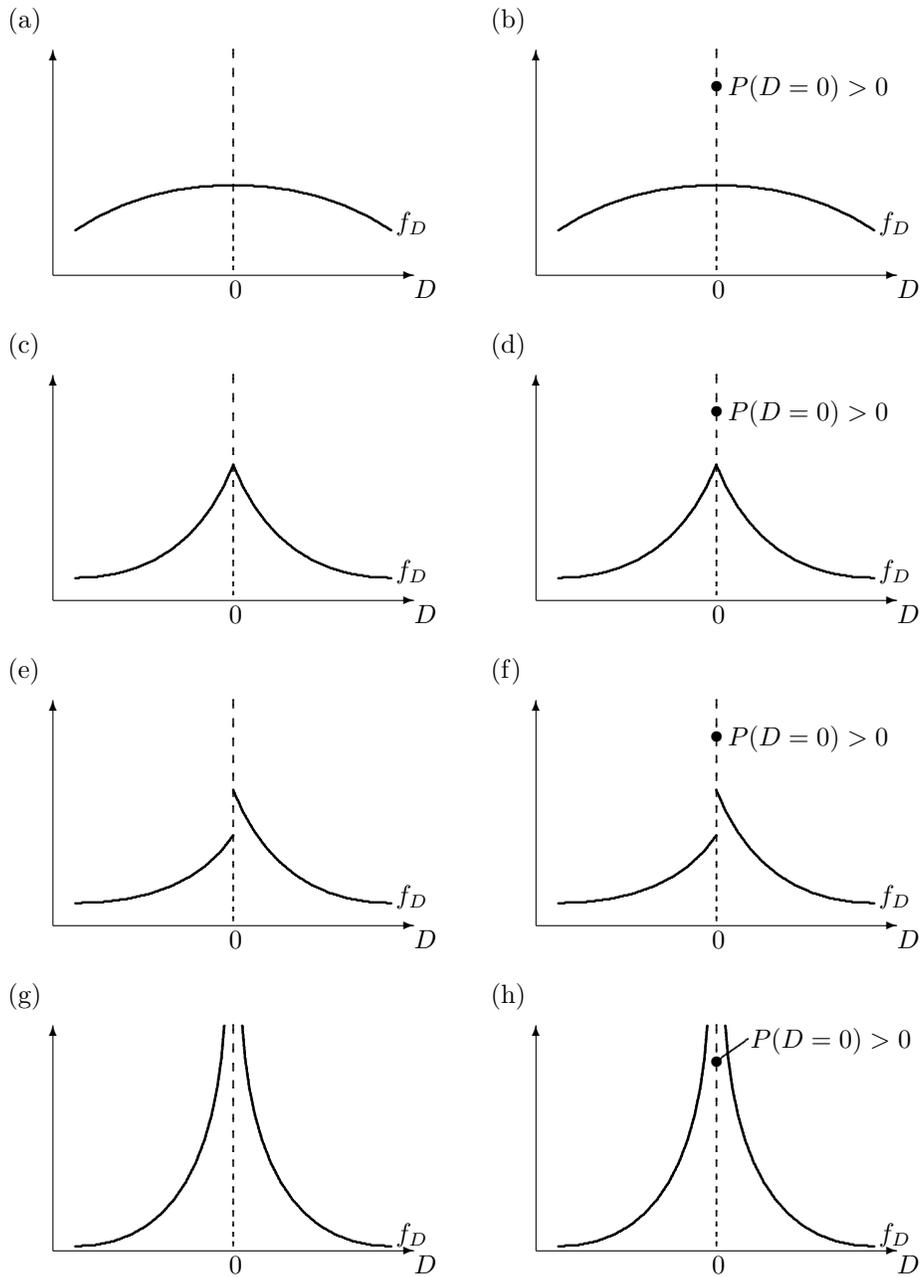


A natural question is whether there is a cost of this robustness property by our unified approach.
In particular, one may wonder whether we lose anything as a result of using our estimator when the underlying distribution in fact does not contain many slow movers. 
Fortunately, we can rather enjoy a gain than incur a cost even in such a case, and this feature is a positive by-product of the proposed estimator.
Specifically, in this case, the proposed estimator can even reduce a finite-sample bias, compared to the estimator proposed by \cite{graham/powell:2012}.   
Therefore, it allows for more accurate statistical inference than the existing estimator, leading to the aforementioned superior coverage frequencies by our method (93-96\%) over the conventional one (37-93\%). 

This paper is related to several branches of the econometrics literature.
Random coefficient panel models have been extensively studied at least since \citet*{swamy1970efficient} -- see \citet[][Sec. 6.2]{hsiao2014analysis} and \cite{hsiao/pesaran:2008} for a review.
Among others, we particularly consider \textit{correlated} random coefficient panel models studied by \citet*{chamberlain:1982,chamberlain:1992}, \citet*{wooldridge2005fixed}, \citet*{murtazashvili2008fixed}, \citet*{arellano2012identifying}, \citet*{graham/powell:2012}, \citet*{fernandez2013panel}, \citet*{bryan/hahn/poirier/powell:2017}, and \citet*{verdier2020average}.
Also related is the literature on nonseparable panel models with correlated random effects, such as \citet*{altonji2005cross}, \citet*{hoderlein/white:2012}, \citet*{chernozhukov2013average}, and \citet*{chernozhukov2015nonparametric}.
Finally, the debiasing feature of our estimator is also related to the bias correction methods \citep*[e.g.,][]{calonico2014robust} in kernel estimation -- also see bias bound approaches \citep*[e.g.,][]{armstrong2020simple,schennach2020bias} to nonparametric inference with kernel methods.

The paper proceeds as follows. 
Section \ref{sec:review} introduces the model, and briefly reviews the identification and estimation strategies in the existing literature. 
Section \ref{sec:overview} proposes the method of unified estimation and robust inference. 
Section \ref{sec:main} provides theoretical results for the proposed method. 
We present Monte Carlo simulations in Section \ref{sec:simulation}. 
Sections \ref{sec:review}--\ref{sec:simulation} focus on the case in which the number ($T$) of time periods is equal to the number ($p$) of regressors.
Section \ref{section:T>p} presents an extension to the case where $T$ is larger than $p$.
Section \ref{sec:conclusion} concludes. 
The appendix collects proofs of the main theorems, auxiliary lemmas and their proofs.

\section{Correlated Random Coefficient Panel Data Models}\label{sec:review}
Following \cite{chamberlain:1992}, we consider the correlated random coefficient panel data model
\begin{equation}\label{eq:panel_model}
Y_t=\mathbf{X}_t'b_t(A,U_t)
\end{equation}
for $t=1,\ldots,T$, where $Y_t$ denotes a scalar outcome, $\mathbf{X}_t$ denotes a $p$-dimensional column  vector of regressors, $A$ denotes a fixed effect, and $U_t$ denotes an idiosyncratic error. 
The parameter of interest is the average partial effect,  $E[b_t(A,U_t)]$, i.e., the effect of an exogenous change in $\mathbf{X}_t$ on $Y_t$.
With these notations, the panel model \eqref{eq:panel_model} is equipped with the following assumption.

\begin{assumption}\label{assn:GP1.1}
(i) $b_t(A,U_t)=b^\ast(A,U_t)+d_t(U_{2t})$ for $t=1,\ldots,T$ and $U_t=(U_{1t},U_{2t})$.
(ii) Given $(\mathbf{X}_1,\ldots,\mathbf{X}_T,A)$, the random variables $U_t$ and $U_s$ have the same conditional distribution. 
(iii) $U_{2t}$ and $(\mathbf{X}_1,\ldots,\mathbf{X}_T,A)$ are independent. 
(iv) $E[b_t(A,U_t)\mid \mathbf{X}_1,\ldots,\mathbf{X}_T]$ exists for $t=1,\ldots,T$.
(v) $E[d_1(U_{21})]=0$.
\end{assumption}

This assumption is standard in the literature \citep[e.g.,][]{chamberlain:1992,graham/powell:2012}. 
Under this assumption, the parameter of interest, $E[b_t(A,U_t)]$, can be decomposed as $E[b_t(A,U_t)]=E[b_1(A,U_1)]+E[d_t(U_{2t})]$. 
The first part $E[b_1(A,U_1)]$ is the average partial effect at $t=1$, and the second part $E[d_t(U_{2t})]$ is the mean time-shift from the first period to the $t$-th period in the average partial effect.
Therefore, the identification for $E[b_t(A,U_t)]$ boils down to the identification for $(\boldsymbol{\beta}',\boldsymbol{\delta}')'$, where $\boldsymbol{\beta}=E[b_1(A,U_1)]$ and $\boldsymbol{\delta}= (E[d_2(U_{22})]',\ldots,E[d_T(U_{2T})]')'$. 

\cite{graham/powell:2012} provide the identification of $(\boldsymbol{\beta}',\boldsymbol{\delta}')'$. 
We focus on the case of $T=p$ in which $E[b_t(A,U_t)]$ cannot be regularly identified. (We discuss an extension to the case of $T>p$ in Section \ref{section:T>p}.) 
Assumption \ref{assn:GP1.1} and the panel model \eqref{eq:panel_model} imply  
\begin{equation}\label{eq:base_ID}
E[\mathbf{X}^\ast\mathbf{Y}\mid\mathbf{X}]=\mathbf{X}^\ast\mathbf{W}\boldsymbol{\delta}+D\boldsymbol{\beta}(\mathbf{X})\mbox{ with }\boldsymbol{\beta}=E[\boldsymbol{\beta}(\mathbf{X})],
\end{equation}
where $\mathbf{Y}=(Y_1,\ldots,Y_T)'$, $\mathbf{X}=(\mathbf{X}_1,\ldots,\mathbf{X}_T)'$, $D=\mathrm{det}(\mathbf{X})$, $\boldsymbol{\beta}(\mathbf{X})= E[b^\ast(A,U_t)|\mathbf{X}]$, $\mathbf{X}^\ast$ is the adjoint of $\mathbf{X}$, and 
$$
\mathbf{W}=
\left(
\begin{array}{ccc}
{0}_p'&&{0}_p'\\
\mathbf{X}_2'&&{0}_p'\\
&\ddots&\\
{0}_p'&&\mathbf{X}_T'
\end{array}
\right).
$$
Based on the conditional moment restriction \eqref{eq:base_ID}, $\boldsymbol{\delta}$ can be point-identified via 
\begin{equation}\label{eq:delta_identification}
\boldsymbol{\delta}=E[(\mathbf{X}^\ast\mathbf{W})'\mathbf{X}^\ast\mathbf{W}\mid D=0]^{-1}E[(\mathbf{X}^\ast\mathbf{W})'\mathbf{X}^\ast\mathbf{Y}\mid D=0].
\end{equation}
The identification analysis for $\boldsymbol{\beta}$ in \cite{graham/powell:2012} branches into two separate cases depending on the probability of the event $D=0$.
The first case, $P(D=0)=0$, represented by Figure \ref{fig:eight_cases} (a), is presented in Section \ref{sec:case_determinant_zero}.
The second case, $P(D=0) \ne 0$, represented by Figure \ref{fig:eight_cases} (b), is presented in Section \ref{sec:case_determinant_nonzero}.

\subsection{Identification and Estimation under $P(D=0)=0$}\label{sec:case_determinant_zero}
When $P(D=0)=0$, the average partial effect $\boldsymbol{\beta}$ is equal to $E[\boldsymbol{\beta}(\mathbf{X})\mid D\ne 0]$. Under Assumption \ref{assn:GP1.1}, \citet[][Proposition 1.2]{graham/powell:2012} propose to estimate $\boldsymbol{\beta}$ by
\begin{equation}\label{eq:beta_hat_m}
\boldsymbol{\hat\beta}^M
=
\frac{E_N[\boldsymbol{1}\{|D|>h_N\}D^{-1}\mathbf{X}^\ast(\mathbf{Y}-\mathbf{W}\boldsymbol{\hat \delta})]}{E_N[\boldsymbol{1}\{|D|>h_N\}]},
\end{equation}
where 
$\boldsymbol{\hat\delta}$ is an estimator of $\boldsymbol{\delta}$, $E_N$ is the sample mean operator and $h_N>0$ is a tuning parameter that converges to $0$ as $N\rightarrow\infty$.\footnote{Although the framework is different, $\boldsymbol{\hat\beta}^M$ can be regarded as the mean group estimator (cf. \citealp{pesaran1995estimating,hsiao_pesaran_tahmiscioglu_1999}). We can interpret $D^{-1}\mathbf{X}^\ast(\mathbf{Y}-\mathbf{W}\boldsymbol{\hat \delta})$ as the individual-specific estimator for $\boldsymbol{\beta}$.} 
Here, the superscript `$M$'  in $\boldsymbol{\hat\beta}^M$ is used to emphasize that it estimates the average partial effect for the `movers.'
It is shown that the bias of $\boldsymbol{\hat\beta}^M$ is $O(h_N)$, i.e., $E[\boldsymbol{\hat\beta}^M]-\boldsymbol{\beta}=O(h_N)$, when $P(D=0)=0$ and $\phi$ is bounded.
This is why \citet[][Assumption 2.5]{graham/powell:2012} use a bandwidth with rate $N^{-1/3}$ in their empirical application, which is smaller than the standard one-dimensional MSE optimal rate $N^{-1/5}$.
This approach is effective in models with smooth density of $D$ without a mass of stayers represented by Figure \ref{fig:eight_cases} (a).

\subsection{Identification and Estimation under $P(D=0)>0$}\label{sec:case_determinant_nonzero}
\citet[][Section 3.2]{graham/powell:2012} separately consider the case with a point mass of stayers, i.e., $P(D=0)>0$. 
They derive the following constructive identifying formula: 
\begin{equation}
\boldsymbol{\beta}
=
P(D\ne 0)E[\boldsymbol{\beta}(\mathbf{X})\mid D\ne 0]+P(D=0)\mathbf{m}^{(1)}(0),
\label{eq:mixture}
\end{equation}
where $\mathbf{m}^{(k)}(u)$ is the $k$-th derivative of the function $\mathbf{m}(u)$ defined by 
\begin{align}\label{eq:def_m}
\mathbf{m}(u)= E[\mathbf{X}^\ast(\mathbf{Y}-\mathbf{W}\boldsymbol{\delta})\mid D=u].
\end{align}
Based on \eqref{eq:mixture}, they estimate the average partial effect $\boldsymbol{\beta}$  by the mixture estimator
\begin{equation}\label{eq:beta_hat}
\boldsymbol{\hat\beta}_1
=
E_N[\boldsymbol{1}\{|D|>h_N\}]\boldsymbol{\hat\beta}^M+E_N[\boldsymbol{1}\{|D|\leq h_N\}]\mathbf{\hat{m}}^{(1)}(0)
\end{equation}
of the average marginal effect $\boldsymbol{\beta}$, where 
$\mathbf{\hat{m}}^{(1)}(0)$ is an estimator for $\mathbf{m}^{(1)}(0)$.
This approach is effective in models with smooth density of $D$ with a mass of stayers represented by Figure \ref{fig:eight_cases} (b).

\section{The Robust Method of Estimation and Inference}\label{sec:overview}

As discussed in the previous section, the existing literature proposes two separate estimation procedures for $\boldsymbol{\beta}$ depending on whether the probability of $D=0$ is zero or not.
For the case of $P(D=0)=0$ (Section \ref{sec:case_determinant_zero}) as represented by Figure \ref{fig:eight_cases} (a), the estimator $\boldsymbol{\hat\beta}^M$ given in \eqref{eq:beta_hat_m} is consistent for $\boldsymbol{\beta}$.
In contrast, it is necessary to estimate the mixture \eqref{eq:mixture} of the marginal effects of movers and stayers for the case in which there is a point mass of stayers, i.e., $P(D=0)>0$ (Section \ref{sec:case_determinant_nonzero}) as represented by Figure \ref{fig:eight_cases} (b).
The existing literature, including \cite{graham/powell:2012}, does not cover the cases illustrated in Figure \ref{fig:eight_cases} (c)--(h).

Our goal is to provide a method that covers various patterns regarding the distribution of $D$.
There can be a positive point mass of $D=0$ (e.g., Figure \ref{fig:eight_cases} (b), (d), (f), \& (h)), and the density of $D$ given $D\ne 0$ can be non-smooth (e.g., Figure \ref{fig:eight_cases} (c)--(h)), unbounded (e.g., Figure \ref{fig:eight_cases} (g)--(h)), or locally asymmetric (e.g., Figure \ref{fig:eight_cases} (e)--(f)). 
Our proposed estimator introduced below accounts for different characteristics of three groups:  the stayers (the individual with $D=0$), the slow movers (the individual with $0<|D|\leq h_N$), and the movers (the individual with $|D|>h_N$).

\subsection{The Estimator for $\boldsymbol{\beta}$}\label{sec:proposed_estimator}

In this section, we introduce a novel estimator for $\boldsymbol{\beta}$ which is robustly valid for a large class of distributions in a unified (rather than separate) manner. 
The large class of distributions includes not only the two cases in  \cite{graham/powell:2012}  but also all the cases illustrated in Figure \ref{fig:eight_cases} (a)--(h). 
Notably, we allow for $P(D=0)=0$ and $P(D=0)>0$, and the density $\phi$ can be unbounded or non-smooth around zero. 
For an integer $L\geq 1$, our proposed estimator $\boldsymbol{\hat\beta}_L$ of $\boldsymbol{\beta}$ is
\begin{equation}\label{eq:proposed_estimation}
\boldsymbol{\hat\beta}_L
=
E_N[\boldsymbol{1}\{|D|>h_N\}D^{-1}\mathbf{X}^\ast(\mathbf{Y}-\mathbf{W}\boldsymbol{\hat \delta})]+\boldsymbol{\hat\gamma}\mathbf{\hat{h}},\mbox{ where }
\end{equation}
$\mathbf{\hat{h}}=E_N[\boldsymbol{1}\{|D|\leq h_N\}(\begin{array}{cccc}1&\cdots&D^{L-1}\end{array})']$, 
$\mathbf{e}_1=(\begin{array}{cccc}1&0&\cdots&0\end{array})'$, 
$\mathbf{D}_{0:L}=\boldsymbol{1}\{|D|\leq h_N\}(\begin{array}{cccc}1&D&\cdots&D^L\end{array})'$, 
$\mathbf{D}_{1:L}=\boldsymbol{1}\{|D|\leq h_N\}(\begin{array}{cccc}D&\cdots&D^L\end{array})'$, 
$$
\boldsymbol{\hat\delta}
=
E_N[(\mathbf{D}_{0:L}'E_N[\mathbf{D}_{0:L}\mathbf{D}_{0:L}']^{-1}\mathbf{e}_1)(\mathbf{X}^\ast\mathbf{W})'\mathbf{X}^\ast\mathbf{W}]^{-1}E_N[(\mathbf{D}_{0:L}'E_N[\mathbf{D}_{0:L}\mathbf{D}_{0:L}']^{-1}\mathbf{e}_1)(\mathbf{X}^\ast\mathbf{W})'\mathbf{X}^\ast\mathbf{Y}],$$ 
and 
$\boldsymbol{\hat\gamma}=E_N\left[\mathbf{X}^\ast(\mathbf{Y}-\mathbf{W}\boldsymbol{\hat\delta})\mathbf{D}_{1:L}'\right]E_N\left[\mathbf{D}_{1:L}\mathbf{D}_{1:L}'\right]^{-1}$.
This proposed estimator $\boldsymbol{\hat\beta}_L$ is shown to be consistent and asymptotically normal for the large class of distributions in a unified manner.

\subsubsection{Intuitions behind the Proposed Estimator}\label{sec:proposed_estimator2}
The estimator $\boldsymbol{\hat\delta}$ for $\boldsymbol{\delta}$ is based on local polynomial regressions on $D$. 
To estimate $\boldsymbol{\delta}$ based on \eqref{eq:delta_identification}, we need to estimate the conditional expectation of $Z$ given $D=0$ for every entry $Z$ of the matrix $(\mathbf{X}^\ast\mathbf{W})'\mathbf{X}^\ast[\mathbf{Y},\mathbf{W}]$. 
When we use the $L$-th order local polynomial regression of $Z$ on $D$ at $D=0$, the estimator for $E[Z\mid D=0]$ is given by $\mathbf{e}_1'\boldsymbol{\hat\kappa}$, where $\boldsymbol{\hat\kappa}$ is the $(L+1)$-dimensional column vector defined by  $$
\arg\min_{\boldsymbol{\kappa}}
E_N\left[\boldsymbol{1}\{|D|\leq h_N\}\left(Z-\left(\begin{array}{cccc}1&D&\cdots&D^L\end{array}\right)'\boldsymbol{\kappa}\right)^2\right].
$$ 
We can derive $\mathbf{e}_1'\boldsymbol{\hat\kappa}=\mathbf{e}_1'E_N[\mathbf{D}_{0:L}\mathbf{D}_{0:L}']^{-1}E_N[\mathbf{D}_{0:L}Z]=E_N[(\mathbf{D}_{0:L}'E_N[\mathbf{D}_{0:L}\mathbf{D}_{0:L}']^{-1}\mathbf{e}_1)Z]$ as an estimator for $E[Z\mid D=0]$ by the first-order condition.
If we construct such an estimator for $E[Z\mid D=0]$ for every entry $Z$ of the matrix $(\mathbf{X}^\ast\mathbf{W})'\mathbf{X}^\ast[\mathbf{Y},\mathbf{W}]$, then the resulting estimator for $\boldsymbol{\delta}$ is $\boldsymbol{\hat\delta}$.
Similarly, the estimator $\boldsymbol{\hat\gamma}$ is the $L$-th order local polynomial regression of $\mathbf{X}^\ast(\mathbf{Y}-\mathbf{W}\boldsymbol{\hat\delta})$ on $D$ without an intercept.\footnote{We do not need to include a constant term here because $E[\mathbf{X}^\ast(\mathbf{Y}-\mathbf{W}\boldsymbol{\delta})\mid D=0]=0$ always holds by \eqref{eq:base_ID}.}

The proposed estimator $\boldsymbol{\hat\beta}_L$ for $\boldsymbol{\beta}$ can be rewritten as 
\begin{equation}\label{eq:proposed_estimation2}
\boldsymbol{\hat\beta}_L
=
E_N[\boldsymbol{1}\{|D|>h_N\}]\boldsymbol{\hat\beta}^M+\sum_{l=1}^{L}\frac{\mathbf{\hat{m}}^{(l)}(0)}{l!}E_N[D^{l-1}\boldsymbol{1}\{|D|\leq h_N\}],
\end{equation}
where $\boldsymbol{\hat\beta}^M$ is defined in \eqref{eq:beta_hat_m} and $(\begin{array}{cccc}\mathbf{\hat{m}}^{(1)}(0)/1!&\cdots&\mathbf{\hat{m}}^{(L)}(0)/L!\end{array}) := \boldsymbol{\hat\gamma}$. 
We remark that this estimator $\boldsymbol{\hat\beta}_L$ becomes the mixture estimator $\boldsymbol{\hat\beta}_1$ in \eqref{eq:beta_hat} in the special case of $L=1$.
This observation implies that our proposed estimator is effective even when there is a point mass of stayers (i.e., $P(D=0)>0$), and the same turns out to hold true for $L\geq 1$ as well.

\subsubsection{Bias Reduction}\label{sec:proposed_estimator3}

In addition to the robustness and unifying feature, there is a bias advantage of considering the proposed estimator $\boldsymbol{\hat\beta}_L$ over $\boldsymbol{\hat\beta}^M$. 
The proposed estimator uses the information about slow movers. 
Even when $\boldsymbol{\hat\beta}^M$ is consistent, the information about slow movers is still useful for the purpose of correcting the bias of $\boldsymbol{\hat\beta}^M$, and thus for the objective of improving the accuracy of estimation and inference over the conventional estimator $\boldsymbol{\hat\beta}^M$. 
More precisely, the proposed estimator $\boldsymbol{\hat\beta}_L$ has the bias of smaller order than that of $\boldsymbol{\hat\beta}^M$. 
In fact, the conventional estimator $\boldsymbol{\hat\beta}^M$ in \eqref{eq:beta_hat_m} captures only  $E[\boldsymbol{\beta}(\mathbf{X})\mid |D|>h_N]$ in the approximation
$$
\boldsymbol{\beta}
=
E[\boldsymbol{\beta}(\mathbf{X})\boldsymbol{1}\{|D|>h_N\}]+O\left(P(|D|\leq h_N)\right)
=
P(|D|>h_N)E[\boldsymbol{\beta}(\mathbf{X})\mid |D|>h_N]+O\left(P(|D|\leq h_N)\right).
$$
Consequently, $\boldsymbol{\hat\beta}^M$ generates a bias of order $o(1)$ as an estimator for $\boldsymbol{\beta}$ when the fraction of stayers is zero, i.e., $P(D=0)=0$. 
In contrast, based on \eqref{eq:proposed_estimation} and \eqref{eq:proposed_estimation2}, the estimator $\boldsymbol{\hat\beta}_L$ proposed in this paper estimates its population counterpart $\mathbf{b}_L$ defined by  
$$
\mathbf{b}_L=E[\boldsymbol{\beta}(\mathbf{X})\boldsymbol{1}\{|D|>h_N\}]+\boldsymbol{\gamma}\mathbf{h}=E[\boldsymbol{\beta}(\mathbf{X})\boldsymbol{1}\{|D|>h_N\}]+\sum_{l=1}^{L}\frac{\mathbf{{m}}^{(l)}(0)}{l!}E[D^{l-1}\boldsymbol{1}\{|D|\leq h_N\}],
$$
where $\mathbf{h}=E[\boldsymbol{1}\{|D|\leq h_N\}(\begin{array}{cccc}1&\cdots&D^{L-1}\end{array})']$ and $\boldsymbol{\gamma}=(\begin{array}{cccc}\mathbf{m}^{(1)}(0)/1!&\cdots&\mathbf{m}^{(L)}(0)/L!\end{array})$. 
This estimand $\mathbf{b}_L$ approximates $\boldsymbol{\beta}$ better than just the leading term $E[\boldsymbol{\beta}(\mathbf{X})\boldsymbol{1}\{|D|>h_N\}]$ because the additional term $\boldsymbol{\gamma}\mathbf{h}$ is based on 
$$
\boldsymbol{\beta}
=
E[\boldsymbol{\beta}(\mathbf{X})\boldsymbol{1}\{|D|>h_N\}]+\sum_{l=1}^L\frac{\mathbf{m}^{(l)}(0)}{l!}E[D^{l-1}\boldsymbol{1}\{|D|\leq h_N\}]
+O\left(E[D^L\boldsymbol{1}\{|D|\leq h_N\}]\right),
$$
which we formally show in Theorem \ref{bias_comp_higher} in Section \ref{sec:main}. 
Consequently, $\boldsymbol{\hat\beta}_L$ generates a smaller bias of order $o(h_N^L)$, as opposed to $o(1)$, regardless of whether $P(D=0)=0$ or not. 
When $L\geq 1$, this bias is smaller than the bias for the conventional estimator $\boldsymbol{\hat\beta}^M$.\footnote{\citet[p.2125]{graham/powell:2012} discuss a possibility of bias correction. Indeed, our proposed estimator in the special case of $L=1$ reduces to  their bias-corrected estimator. However, \cite{graham/powell:2012} leave an asymptotic analysis for future research. We fill the gap by providing a formal asymptotic analysis for the proposed estimator with $L=1$ as well as with general $L\geq 1$.} 

Thanks to this bias reduction feature of our proposed estimator $\boldsymbol{\hat\beta}_L$, our requirement regarding the convergence rate of $h_N$ is weaker than that for $\boldsymbol{\hat\beta}^M$. When we use $\boldsymbol{\hat\beta}^M$ as an estimator for $\boldsymbol{\beta}$, we would need to use a small bandwidth $h_N$ satisfying $Nh_N^3\rightarrow 0$ to undersmooth the asymptotic bias as in the empirical application of \cite{graham/powell:2012}. 
In contrast, we can choose a large bandwidth $h_N$, e.g., a value which is proportional to $N^{-1/(2L+1)}$. 
This weaker requirement on $h_N$ allows us to achieve a faster convergence rate in estimation by taking $h_N$ larger. 
At the same time, the validity for the statistical inference is insensitive to $h_N$ because our inference is valid for a larger class of $h_N$ than that based on $\boldsymbol{\hat\beta}^M$.

%
%

\subsection{Estimation and Inference on the Average Partial Effect $E[b_t(A,U_t)]$}

Based on the estimators proposed in Section \ref{sec:proposed_estimator}, we can conduct robust inference about the parameters. 
In particular, we consider the average partial effect $E[b_t(A,U_t)]=E[b_1(A,U_1)]+E[d_t(U_{2t})]$, which can be succinctly written as $\boldsymbol{\theta}=\boldsymbol{\beta}+\mathbf{R}\boldsymbol{\delta}$ using a $p \times (p\cdot(T-1))$ matrix $\mathbf{R}$ extracting time effects for a desired period $t$.
For instance, when $p=T=2$, $\mathbf{R}$ is the $2 \times 2$ matrix of zeros for $t=1$ and $\mathbf{R}$ is the $2 \times 2$ identity matrix for $t=2$.

The estimator for the average partial effect $\boldsymbol{\theta}$ is $\boldsymbol{\hat\theta}=\boldsymbol{\hat\beta}_L+\mathbf{R}\boldsymbol{\hat\delta}$. 
We estimate the influence function for this estimator by 
\begin{align*}
\boldsymbol{\hat\zeta}
=&
\boldsymbol{1}\{|D|>h_N\}D^{-1}\mathbf{X}^\ast(\mathbf{Y}-\mathbf{W}\boldsymbol{\hat\delta})-E_N[\boldsymbol{1}\{|D|>h_N\}D^{-1}\mathbf{X}^\ast(\mathbf{Y}-\mathbf{W}\boldsymbol{\hat \delta})]
\\&+
(\mathbf{X}^\ast(\mathbf{Y}-\mathbf{W}\boldsymbol{\hat\delta})-\boldsymbol{\hat\gamma}\mathbf{D}_{1:L})\mathbf{D}_{1:L}'E_N\left[\mathbf{D}_{1:L}\mathbf{D}_{1:L}'\right]^{-1}\mathbf{\hat{h}}
\\&+
\mathbf{\hat{Q}}\mathbf{\hat{V}}^{-1}\left(\mathbf{D}_{0:L}'E_N\left[\mathbf{D}_{0:L}\mathbf{D}_{0:L}'\right]^{-1}\mathbf{e}_1\right)(\mathbf{X}^\ast\mathbf{W})'\mathbf{X}^\ast(\mathbf{Y}-\mathbf{W}\boldsymbol{\hat\delta}),
\end{align*}
where 
\begin{align*}
\mathbf{\hat{V}}
=&
E_N\left[\left(\mathbf{D}_{0:L}'E_N\left[\mathbf{D}_{0:L}\mathbf{D}_{0:L}'\right]^{-1}\mathbf{e}_1\right)(\mathbf{X}^\ast\mathbf{W})'\mathbf{X}^\ast\mathbf{W}\right],\mbox{ and }
\\
\mathbf{\hat{Q}}
=&
\mathbf{R}-E_N\left[\left(\boldsymbol{1}\{|D|>h_N\}D^{-1}+\mathbf{D}_{1:L}'E_N\left[\mathbf{D}_{1:L}\mathbf{D}_{1:L}'\right]^{-1}\mathbf{\hat{h}}\right)\mathbf{X}^\ast\mathbf{W}\right].
\end{align*}
In the next section, we establish the asymptotic normality 
$$
\sqrt{N}E_N\left[\boldsymbol{\hat\zeta}\boldsymbol{\hat\zeta}'\right]^{-1/2}(\boldsymbol{\hat\theta}-\boldsymbol{\theta})
\rightarrow_d
\mathcal{N}\left(0,\mathbf{I}_{p}\right)
$$
as $h_N \rightarrow 0$ and $N \rightarrow \infty$.
This result allows us to conduct a statistically valid inference on the average partial effect $\boldsymbol{\theta} = E[b_t(A,U_t)]$.

\section{Asymptotic Properties of the Estimator $\boldsymbol{\hat\theta}$}\label{sec:main}

In this section, we present asymptotic properties about $\boldsymbol{\hat\theta}$ to provide a theoretical guarantee that the method introduced in Section \ref{sec:overview} works. 
Let us start with stating assumptions on the data generating process. 
We use the information near $D=0$ to estimate $\boldsymbol{\beta}$, and we therefore impose the following assumption, which generalizes the assumptions in \cite{graham/powell:2012}. 

\begin{assumption}\label{assn:GP_around0}
(i) $P(D\ne 0)>0$. (ii) The conditional distribution of $D$ given $D\ne 0$ has a density $\phi$ with respect to the Lebesgue measure in a neighborhood of zero. 
(iii) $\phi$ is bounded away from zero in a neighborhood of zero. 
(iv) There are positive numbers $\alpha_1,\alpha_2$ with $\max\{\alpha_1,\alpha_2\}\leq 1$ such that the conditional distribution of $|D|^{-1}$ given $D>0$ has a regularly varying tail with index $\alpha_1$ and the conditional distribution of $|D|^{-1}$ given $D<0$ has a regularly varying tail with index $\alpha_2$.\footnote{In other words, $\lim_{t \rightarrow \infty} \frac{P(|D|^{-1}\geq tu\mid D>0)}{P(|D|^{-1}\geq t\mid D>0)}=u^{\alpha_1}$ and $\lim_{t \rightarrow \infty} \frac{P(|D|^{-1}\geq tu\mid D<0)}{P(|D|^{-1}\geq t\mid D<0)}=u^{\alpha_2}$ for every positive number $u$.}
\end{assumption}

In this assumption, we allow for $P(D=0)>0$ as well as $\lim_{u\rightarrow 0}\phi(u)=\infty$, whereas it requires the absolute continuity of the conditional distribution of $D$ ``given $D\ne 0$.''
It can be empirically relevant many stayers and many slow movers may exist as discussed in the introductory section with Figures \ref{fig:PSID_17_19} and \ref{fig:eight_cases}. 
This feature is different from \cite{graham/powell:2012}, where they treat the case with no mass and the case with a mass separately and they assume that $\phi$ is smooth and bounded around zero, as in Figure \ref{fig:eight_cases} (a)--(b). 
Assumption \ref{assn:GP_around0} encompasses all of the eight cases illustrated in Figure \ref{fig:eight_cases}.

To control for a potential unboundedness or non-smoothness of the density $\phi$ near $0$ as in Figure \ref{fig:eight_cases} (c)--(h), we impose the regularly varying tail condition.\footnote{The regularly varying tail condition holds for a large class of distributions for $D$ and has been used in different contexts. For example, \cite{chaudhuri/hill:2016} uses this condition for the inverse probability weighting estimation of the average treatment effect.}
The tails are regularly varying with $\alpha_1=\alpha_2=1$ in the special case where $\phi$ is continuous around zero as in \citet{graham/powell:2012}, and it is also the case  when $\phi(t)$ is proportional to $|t|^{\alpha-1}$ near zero for some $\alpha>0$.
We require that the density of $D$ is bounded away from zero in a neighborhood of zero, so $\max\{\alpha_1,\alpha_2\}\leq 1$. For example, when $\phi(t)=\alpha |t|^{\alpha-1}/2$ for $t\in[-1,1]$, the density function $\phi$ is bounded away from zero near zero if and only if $\alpha\leq 1$.

Since our proposed estimator uses local polynomial regressions, we impose the following smoothness conditions on a few functions. 

\begin{assumption}\label{assn:more_primitiv_higher}
(i) The mappings, $u\mapsto E[\mathbf{X}^\ast(\mathbf{Y}-\mathbf{W}\boldsymbol{\delta})\mid D=u]$ and $u\mapsto E\left[(\mathbf{X}^\ast\mathbf{W})'\mathbf{X}^\ast(\mathbf{Y}-\mathbf{W}\boldsymbol{\delta})\mid D=u\right]$, are $(L+1)$-times differentiable at $0$ whose $(L+1)$-th derivative is  bounded in a neighborhood of zero. 
(ii) The mappings, $u\mapsto E\left[(\mathbf{X}^\ast\mathbf{W})'\mathbf{X}^\ast\mathbf{W}\mid D=u\right]$ and $u\mapsto Var(\mathbf{X}^\ast(\mathbf{Y}-\mathbf{W}\boldsymbol{\delta})\mid D=u)$, are differentiable at $0$ whose derivative is  bounded in a neighborhood of zero. 
\end{assumption}

Moreover, we assume i.i.d. sampling across individuals, bounded moments, and invertibility for matrices.

\begin{assumption}\label{assn:iid_plus}
(i) $\{(\mathbf{Y}_i,\mathbf{X}_i)\}_{i=1}^N$ are independently and identically distributed. 
(ii) $E[Z]<\infty$ and $E[Z\mid D=u]$ is bounded in a neighborhood of zero for 
$Z=\|\mathbf{X}^\ast\mathbf{W}\|^4$, 
$\|\mathbf{X}^\ast(\mathbf{Y}-\mathbf{W}\boldsymbol{\delta})\|^4$, 
$\|(\mathbf{X}^\ast\mathbf{W})'\mathbf{X}^\ast(\mathbf{Y}-\mathbf{W}\boldsymbol{\delta})\|^4$.
(iii) $E[(\mathbf{X}^\ast\mathbf{W})'\mathbf{X}^\ast\mathbf{W}\mid D=0]$ and $Var(\mathbf{X}^\ast(\mathbf{Y}-\mathbf{W}\boldsymbol{\delta})\mid D=0)$ are invertible.
\end{assumption}
 
Now, we present the two main theoretical results. 
First, we characterize the bias of $\boldsymbol{\hat\beta}_L$ as an estimator for $\boldsymbol{\beta}$. 
The term $\mathbf{b}_L-\boldsymbol{\beta}$ captures the bias, and we show that the bias can be smaller as $L$ becomes larger. 

\begin{theorem}\label{bias_comp_higher}
Under $T=p$ and Assumptions \ref{assn:GP1.1}, \ref{assn:GP_around0}, and \ref{assn:more_primitiv_higher}, $\mathbf{b}_L=\boldsymbol{\beta}+O\left(E[D^L\boldsymbol{1}\{|D|\leq h_N\}]\right)$. 
\end{theorem}

\noindent
Together with this bias reduction result, we can characterize the asymptotic distribution for $\boldsymbol{\hat\theta}$ as an estimator for $\boldsymbol{\theta}$. 

\begin{theorem}\label{theorem_asynormal}
Suppose that there is a positive constant $c$ such that 
\begin{equation}\label{assn:bandwidth}
Nh_N^{2+c}\rightarrow\infty\mbox{ and }Nh_N^{2(L+1)}\rightarrow 0
\end{equation}
as $h_N \rightarrow 0$ and $N \rightarrow \infty$.
Under $T=p$ and Assumptions \ref{assn:GP1.1}, \ref{assn:GP_around0}, \ref{assn:more_primitiv_higher}, and \ref{assn:iid_plus},  
$$
\sqrt{N}E_N\left[\boldsymbol{\hat\zeta}\boldsymbol{\hat\zeta}'\right]^{-1/2}(\boldsymbol{\hat\theta}-\boldsymbol{\theta})\rightarrow_d\mathcal{N}\left(0,\mathbf{I}_{p}\right)
$$ 
as $h_N \rightarrow 0$ and $N \rightarrow \infty$.
\end{theorem}

\noindent
The condition \eqref{assn:bandwidth} on the tuning parameter $h_N$ can be satisfied for example if we take $h_N$ to be proportional to $N^{-1/(2L+1)}$.

\section{Simulation Studies}\label{sec:simulation}

In this section, we use extensive simulations to demonstrate the robust performance of the unified estimator $\boldsymbol{\hat{\beta}}_L$ proposed in this paper.
In the first set of simulation exercise, we consider various data generating processes with no mass of stayers, i.e., $P(D=0)=0$. 
The objective of this exercise is to compare the performance of the unified estimator $\boldsymbol{\hat\beta}_L$ with that of the conventional estimator $\boldsymbol{\hat\beta}^M$ when both of them are consistent. 
In the second set of simulation exercise, we consider various data generating processes with a nontrivial mass of stayers, i.e., $P(D=0)>0$.
The object of this exercise is to demonstrate the robust consistency of the unified estimator $\boldsymbol{\hat\beta}_L$ even in the presence of stayers.

\subsection{Data Generating Process}

Consider the correlated random coefficient panel model \citep{graham/powell:2012}:
$$
Y_t = \mathbf{X}_t' b_t(A,U_t), \qquad t = 1, 2,
$$
where $\mathbf{X}_t = (1,X_t)'$ is a bivariate explanatory variable.
The random coefficients are decomposed as
$
b_t(A,U_t) = b^\ast(A,U_t) + d_t(U_{2t}),
$
where the endogenous effects $b^\ast(A,U_t)$ and the time effects $d_t(U_{2t})$ are given by
\begin{align*}
b^\ast(A,U_{11t},U_{12t}) &= (A+U_{11t}, A+ U_{12t})'
\qquad\text{and}\\
d_t(U_{21t},U_{22t}) &= (U_{21t} + 0.5 \cdot \boldsymbol{1}\{t=2\}, U_{22t} + 0.5 \cdot \boldsymbol{1}\{t=2\})',
\end{align*}
respectively.
The correlated part $A$ is generated by 
$$
A\mid \varepsilon\sim N(\rho\sigma_A(1+\varepsilon),\sigma_A^2),
$$
where $\varepsilon$ is a non-negative random variable with $P(\varepsilon\leq t)=\pi_0+(1-\pi_0)t^{\alpha}$ for every $t\in[0,1]$.
The idiosyncratic parts $(U_{11t}, U_{12t}, U_{21t}, U_{22t})'$ are generated independently as 
$$
U_{11t}, U_{12t}, U_{21t}, U_{22t} \sim N(0,\sigma_U^2).
$$
The explanatory variables $(X_1,X_2)$ are generated according to 
$$
X_1\sim N(0,1)
\mbox{ and } 
X_2=X_1+\lambda\varepsilon
$$
where $\lambda$ is a Rademacher random variable.

Roles played by the main parameters in this set up are as follows.
The correlation parameter $\rho$ determines the extent of endogeneity.
The mixture parameter $\pi_0 = P(D=0)$ determines the probability mass of exact stayers. 
The power parameter $\alpha$ determines the concentration of movers near $D=0$.
Specifically, the smaller $\alpha$ is, the more slow movers will exist.
The signal-to-noise ratio $\sigma_A / \sigma_U$ determines the relative strength of the fixed effects to the idiosyncratic errors.
In this setup, the baseline population average endogenous effect is
$$
\boldsymbol{\beta} = E[b^\ast(A,U_{11t},U_{12t})] = 
\left(\pi_0+(1-\pi_0)\frac{2\alpha+1}{\alpha+1}\right)\rho\sigma_A
\left(1,1\right)',
$$
and the population average time effects are
$$
\delta_{0t} = E[d_1(U_{21t},U_{22t})] = 
\begin{cases}
(0.00, 0.00)' & \text{if } t = 1
\\
(0.05, 0.05)' & \text{if } t = 2
\end{cases}.
$$
These numbers are consistent with the normalization $\delta_{01} = (0.0,0.0)'$ imposed by \citet{graham/powell:2012}.
In our simulation analysis, we compare results across various values of the correlation parameter $\rho \in \{0.5,1.0\}$, the stayer probability parameter $\pi_0 \in \{0.0,0.1\}$, and the concentration parameter $1/\alpha \in \{1,2,3,4\}$.
We set $\sigma_A=\sigma_U=0.1$ throughout.

\subsection{Simulation Procedure}

For each realization of data $\{(Y_{i1},X_{i1}',Y_{i2},X_{i2}')'\}_{i=1}^n$ generated by the process described above, we first compute the tuning parameter $h_N$ via the plug-in rule $h_N = \frac{1}{2}\min\{\hat\sigma_D, \hat\tau_D/1.34\} N^{-1/(2L+1)}$, where $\hat\sigma_D$ and $\hat\tau_D$ denote the sample standard deviation and the sample inter-quartile range of $D$, respectively.
This choice rule coincides with that of \citet[][page 2138]{graham/powell:2012} when $L=1$.
After $h_N$ is selected, the population average effect $\boldsymbol{\beta}$ is estimated by the conventional estimator $\boldsymbol{\hat\beta}^M$ and the proposed unified estimator $\boldsymbol{\hat\beta}_L$.
Finally, the standard errors of $\boldsymbol{\hat\beta}^M$ and $\boldsymbol{\hat\beta}_L$ are estimated.
The estimated standard errors together with the parameter estimates are in turn used to construct estimated 95\% confidence intervals for the parameter $\boldsymbol{\beta}$, with the boundaries defined as the 2.5-th and 97.5-th percentiles of the estimated asymptotic normal distribution.

Simulations are run across the alternative sample sizes: $N \in \{500,1000\}$.
Each set of simulations runs $2,500$ iterations of data generation, $h_N$-choice, parameter estimation, and standard error estimation.
The average, bias, standard deviation, and root mean squared error of estimates, $\boldsymbol{\hat\beta}_L$ and $\boldsymbol{\hat\delta}$, are computed over $2,500$ iterations for each simulation set.
In addition, the coverage frequencies of the true parameter value $\boldsymbol{\beta}$ by the estimated 90\% and 95\% confidence intervals are computed over $2,500$ iterations for each simulation set.

\subsection{Simulation Results}\label{sec:simulation_results}

\subsubsection{With No Exact Stayers}\label{sec:simulation_results_no_stayers}

We first focus on the subset of simulation results based on the data generating processes with no exact stayers, i.e., $\pi_0 = 0$. 
Table \ref{tab:simulations_benchmark} reports a summary of results for various values of $1/\alpha \in \{1,2\}$ and $\rho \in \{0.5,1.0\}$.
The upper (respectively, lower) half of the table shows results for the first (respectively, second) coordinate, $\boldsymbol{\beta}_{00}$ (respectively, $\boldsymbol{\beta}_{01}$), of the two-dimensional parameter vector $\boldsymbol{\beta}$.
Recall that the first coordinate represents the additive heterogeneity while the second represents the heterogeneous marginal effect of $X_t$.
For each combination of data generating parameter values, $\pi_0$, $1/\alpha \in \{1,2\}$, and $\rho \in \{0.5,1.0\}$, simulation statistics are reported for each of the two estimators, namely the conventional estimator $\boldsymbol{\hat\beta}^M$ and our proposed unified estimator $\boldsymbol{\hat\beta}_L$ with the degree $L=2$ of local polynomial.
Displayed statistics are the mean, bias, standard deviation (SD), root mean square error (RMSE), 90\% coverage, and 95\% coverage.

\begin{table}
\centering
\begin{tabular}{cccccccccccccc}
\hline\hline
& $\pi_0$ & $1/\alpha$ & $\rho$ & Estimator & $N$ & True & Mean & Bias & SD & RMSE & 90\% & 95\%\\
\hline
$\boldsymbol{\beta}_{00}$
& 0.0 & 1 & 0.5 & $\boldsymbol{\hat\beta}^M$ &  500 & 0.075 & 0.077 & 0.002 & 0.009 & 0.009 & 0.853 & 0.920\\
&     &   &     &                            & 1000 & 0.075 & 0.077 & 0.002 & 0.006 & 0.006 & 0.861 & 0.926\\
\cline{5-13}
&     &   &     & $\boldsymbol{\hat\beta}_L$ & \sh  500 & \sh 0.075 & \sh 0.075 & \sh 0.000 & \sh 0.010 & \sh 0.010 & \sh 0.900 & \sh 0.949\\
&     &   &     &                            & \sh 1000 & \sh 0.075 & \sh 0.075 & \sh 0.000 & \sh 0.006 & \sh 0.006 & \sh 0.908 & \sh 0.958\\
\cline{2-13}
& 0.0 & 2 & 0.5 & $\boldsymbol{\hat\beta}^M$ &  500 & 0.067 & 0.072 & 0.005 & 0.009 & 0.010 & 0.733 & 0.815\\
&     &   &     &                            & 1000 & 0.067 & 0.071 & 0.004 & 0.006 & 0.008 & 0.708 & 0.794\\
\cline{5-13}
&     &   &     & $\boldsymbol{\hat\beta}_L$ & \sh  500 & \sh 0.067 & \sh 0.067 & \sh 0.000 & \sh 0.009 & \sh 0.009 & \sh 0.898 & \sh 0.947\\
&     &   &     &                            & \sh 1000 & \sh 0.067 & \sh 0.067 & \sh 0.000 & \sh 0.006 & \sh 0.006 & \sh 0.917 & \sh 0.958\\
\cline{2-13}
& 0.0 & 1 & 1.0 & $\boldsymbol{\hat\beta}^M$ &  500 & 0.150 & 0.154 & 0.004 & 0.009 & 0.010 & 0.826 & 0.902\\
&     &   &     &                            & 1000 & 0.150 & 0.154 & 0.004 & 0.006 & 0.007 & 0.817 & 0.884\\
\cline{5-13}
&     &   &     & $\boldsymbol{\hat\beta}_L$ & \sh  500 & \sh 0.150 & \sh 0.151 & \sh 0.001 & \sh 0.020 & \sh 0.020 & \sh 0.909 & \sh 0.959\\
&     &   &     &                            & \sh 1000 & \sh 0.150 & \sh 0.150 & \sh 0.000 & \sh 0.007 & \sh 0.007 & \sh 0.905 & \sh 0.956\\
\cline{2-13}
& 0.0 & 2 & 1.0 & $\boldsymbol{\hat\beta}^M$ &  500 & 0.133 & 0.143 & 0.010 & 0.009 & 0.013 & 0.595 & 0.702\\
&     &   &     &                            & 1000 & 0.133 & 0.142 & 0.009 & 0.006 & 0.011 & 0.488 & 0.592\\
\cline{5-13}
&     &   &     & $\boldsymbol{\hat\beta}_L$ & \sh  500 & \sh 0.133 & \sh 0.134 & \sh 0.001 & \sh 0.009 & \sh 0.009 & \sh 0.910 & \sh 0.958\\
&     &   &     &                            & \sh 1000 & \sh 0.133 & \sh 0.134 & \sh 0.000 & \sh 0.007 & \sh 0.007 & \sh 0.909 & \sh 0.953\\
\hline
& $\pi_0$ & $1/\alpha$ & $\rho$ & Estimator & $N$ & True & Mean & Bias & SD & RMSE & 90\% & 95\%\\
\hline
$\boldsymbol{\beta}_{01}$
& 0.0 & 1 & 0.5 & $\boldsymbol{\hat\beta}^M$ &  500 & 0.075 & 0.077 & 0.002 & 0.016 & 0.016 & 0.852 & 0.904\\
&     &   &     &                            & 1000 & 0.075 & 0.077 & 0.002 & 0.017 & 0.017 & 0.854 & 0.917\\
\cline{5-13}
&     &   &     & $\boldsymbol{\hat\beta}_L$ & \sh  500 & \sh 0.075 & \sh 0.075 & \sh 0.000 & \sh 0.018 & \sh 0.018 & \sh 0.905 & \sh 0.953\\
&     &   &     &                            & \sh 1000 & \sh 0.075 & \sh 0.075 & \sh 0.000 & \sh 0.014 & \sh 0.014 & \sh 0.912 & \sh 0.956\\
\cline{2-13}
& 0.0 & 2 & 0.5 & $\boldsymbol{\hat\beta}^M$ &  500 & 0.067 & 0.071 & 0.005 & 0.009 & 0.010 & 0.739 & 0.822\\
&     &   &     &                            & 1000 & 0.067 & 0.071 & 0.004 & 0.006 & 0.007 & 0.719 & 0.804\\
\cline{5-13}
&     &   &     & $\boldsymbol{\hat\beta}_L$ & \sh  500 & \sh 0.067 & \sh 0.067 & \sh 0.000 & \sh 0.009 & \sh 0.009 & \sh 0.894 & \sh 0.943\\
&     &   &     &                            & \sh 1000 & \sh 0.067 & \sh 0.067 & \sh 0.000 & \sh 0.006 & \sh 0.006 & \sh 0.899 & \sh 0.948\\
\cline{2-13}
& 0.0 & 1 & 1.0 & $\boldsymbol{\hat\beta}^M$ &  500 & 0.150 & 0.154 & 0.004 & 0.020 & 0.020 & 0.804 & 0.887\\
&     &   &     &                            & 1000 & 0.150 & 0.153 & 0.003 & 0.010 & 0.011 & 0.808 & 0.880\\
\cline{5-13}
&     &   &     & $\boldsymbol{\hat\beta}_L$ & \sh  500 & \sh 0.150 & \sh 0.150 & \sh 0.000 & \sh 0.024 & \sh 0.024 & \sh 0.910 & \sh 0.956\\
&     &   &     &                            & \sh 1000 & \sh 0.150 & \sh 0.150 & \sh 0.000 & \sh 0.010 & \sh 0.010 & \sh 0.911 & \sh 0.957\\
\cline{2-13}
& 0.0 & 2 & 1.0 & $\boldsymbol{\hat\beta}^M$ &  500 & 0.133 & 0.143 & 0.009 & 0.009 & 0.013 & 0.600 & 0.704\\
&     &   &     &                            & 1000 & 0.133 & 0.142 & 0.009 & 0.006 & 0.011 & 0.484 & 0.589\\
\cline{5-13}
&     &   &     & $\boldsymbol{\hat\beta}_L$ & \sh  500 & \sh 0.133 & \sh 0.134 & \sh 0.000 & \sh 0.009 & \sh 0.009 & \sh 0.912 & \sh 0.957\\
&     &   &     &                            & \sh 1000 & \sh 0.133 & \sh 0.134 & \sh 0.000 & \sh 0.007 & \sh 0.007 & \sh 0.910 & \sh 0.953\\
\hline\hline
\end{tabular}
\caption{Simulation results for the conventional estimator $\boldsymbol{\hat\beta}^M$ and the unified estimator $\boldsymbol{\hat\beta}$ based on 2,500 Monte Carlo iterations. Results are displayed for each of the two coordinates, $\boldsymbol{\beta}_{00}$ and $\boldsymbol{\beta}_{01}$, of $\boldsymbol{\beta}_{0}$, and for each of the two sample sizes, $N \in \{500,1000\}$. While exact stayers are absent (i.e., $\pi_0=0.0$), the values of $1/\alpha \in \{1,2\}$ and $\rho \in \{0.5,1.0\}$ are varied across sets of simulations. The order $L=2$ of local polynomials are used for estimation of $\boldsymbol{\delta}$ and $\boldsymbol{\beta}_L$. Displayed statistics are the mean, bias, standard deviation (SD), root mean square error (RMSE), 90\% coverage, and 95\% coverage.}
\label{tab:simulations_benchmark}
\end{table}

Choose any combination of the data generating parameters, for in stance, $(\pi_0,1/\alpha,\rho)=(0.0,1,0.5)$ as in the first row of Table \ref{tab:simulations_benchmark}.
Comparing the bias in this first row group between the conventional estimator $\boldsymbol{\hat\beta}^M$ (marked by white cells) and the unified estimator $\boldsymbol{\hat\beta}_L$ (marked by shaded cells), we find that the former estimator is biased while the latter is unbiased.
This feature is invariant regardless of the sample size $N$.
Due to the bias, the confidence intervals associated with the conventional estimator $\boldsymbol{\hat\beta}^M$ are not correctly centered, and hence their coverage frequencies fall short of the nominal probabilities of 90\% and 95\%.
On the other hand, the confidence intervals associated with the unified estimator $\boldsymbol{\hat\beta}_L$ are more correctly centered, and hence their coverage frequencies are closer to the nominal probabilities.
These observations are true for any other combination of the values of the parameters $(\pi_0,1/\alpha,\rho)$ displayed (and not displayed) in Table \ref{tab:simulations_benchmark}.
In fact, we can see that the extents of bias and undercoverage for the conventional estimator even exacerbate under the other parameter combinations.

While we found a significant contrast in bias and coverage frequencies between the conventional estimator $\boldsymbol{\hat\beta}^M$ and the unified estimator $\boldsymbol{\hat\beta}_L$, there are not large differences in the standard deviations.
In other words, the bias correction by the unified estimator $\boldsymbol{\hat\beta}_L$ works without significantly sacrificing the variance.
As a consequence, the root mean square error of the unified estimator $\boldsymbol{\beta}_L$ is no worse than and that of the conventional estimator $\boldsymbol{\beta}^M$ in general.
Finally, notice that the results for the first coordinate $\boldsymbol{\beta}_{00}$ displayed in the upper half of Table \ref{tab:simulations_benchmark} and those for the second coordinate $\boldsymbol{\beta}_{01}$ displayed in the lower half are similar to each other under our data generating design.
Following these observations, we hereafter focus on just one of these two coordinates, say, the second coordinate $\boldsymbol{\beta}_{01}$, and examine only the bias, root mean square error, and 95\% coverage frequencies in order to highlight the key observations more concisely.
Extracting these selected features, Table \ref{tab:simulations_bias_rmse_95_pi0_L2} summarizes simulation results with the combinations, $\pi_0 = 0$, $1/\alpha \in \{1,2,3,4\}$, and $\rho \in \{0.5,1.0\}$, of the data generating parameters.
In this table, we still use the degree $L=2$ of local polynomials.
 
\begin{table}
\centering
\begin{tabular}{|cc|cccc|c|cc|cccc|}
\multicolumn{13}{c}{Conventional Estimator versus Unified Estimator with Large Bandwidths and $L=2$}\\
\multicolumn{13}{c}{}\\
\multicolumn{6}{c}{(I) Bias: $\boldsymbol{\hat\beta}^M$} & \multicolumn{1}{c}{} & \multicolumn{6}{c}{(II) Bias: $\boldsymbol{\hat\beta}_L$}\\
\cline{1-6}\cline{8-13}
\multicolumn{2}{|c|}{       } & \multicolumn{4}{|c|}{$1/\alpha$} &&\multicolumn{2}{|c|}{} & \multicolumn{4}{|c|}{$1/\alpha$}\\
\multicolumn{2}{|c|}{$N=1000$} & 1 & 2 & 3 & 4 && \multicolumn{2}{|c|}{$N=1000$} & 1 & 2 & 3 & 4\\
\cline{1-6}\cline{8-13}
$\rho$ & 0.5 & 0.002 & 0.004 & 0.005 & 0.005 && $\rho$ & 0.5 & \sh 0.000 & \sh 0.000 & \sh 0.000 & \sh 0.000\\
       & 1.0 & 0.003 & 0.009 & 0.010 & 0.010 &&        & 1.0 & \sh 0.000 & \sh 0.000 & \sh 0.000 & \sh 0.000\\ 
\cline{1-6}\cline{8-13}
\multicolumn{13}{c}{}\\
\multicolumn{6}{c}{(III) RMSE: $\boldsymbol{\hat\beta}^M$} & \multicolumn{1}{c}{} & \multicolumn{6}{c}{(IV) RMSE: $\boldsymbol{\hat\beta}_L$}\\
\cline{1-6}\cline{8-13}
\multicolumn{2}{|c|}{       } & \multicolumn{4}{|c|}{$1/\alpha$} &&\multicolumn{2}{|c|}{} & \multicolumn{4}{|c|}{$1/\alpha$}\\
\multicolumn{2}{|c|}{$N=1000$} & 1 & 2 & 3 & 4 && \multicolumn{2}{|c|}{$N=1000$} & 1 & 2 & 3 & 4\\
\cline{1-6}\cline{8-13}
$\rho$ & 0.5 & 0.017 & 0.007 & 0.008 & 0.008 && $\rho$ & 0.5 & \sh 0.014 & \sh 0.006 & \sh 0.007 & \sh 0.008\\
       & 1.0 & 0.011 & 0.011 & 0.012 & 0.012 &&        & 1.0 & \sh 0.010 & \sh 0.007 & \sh 0.007 & \sh 0.008\\ 
\cline{1-6}\cline{8-13}
\multicolumn{13}{c}{}\\
\multicolumn{6}{c}{(V) 95\% Cover: $\boldsymbol{\hat\beta}^M$} & \multicolumn{1}{c}{} & \multicolumn{6}{c}{(VI) 95\% Cover: $\boldsymbol{\hat\beta}_L$}\\
\cline{1-6}\cline{8-13}
\multicolumn{2}{|c|}{       } & \multicolumn{4}{|c|}{$1/\alpha$} &&\multicolumn{2}{|c|}{} & \multicolumn{4}{|c|}{$1/\alpha$}\\
\multicolumn{2}{|c|}{$N=1000$} & 1 & 2 & 3 & 4 && \multicolumn{2}{|c|}{$N=1000$} & 1 & 2 & 3 & 4\\
\cline{1-6}\cline{8-13}
$\rho$ & 0.5 & 0.917 & 0.804 & 0.756 & 0.732 && $\rho$ & 0.5 & \sh 0.956 & \sh 0.948 & \sh 0.954 & \sh 0.946\\
       & 1.0 & 0.880 & 0.589 & 0.508 & 0.493 &&        & 1.0 & \sh 0.957 & \sh 0.953 & \sh 0.956 & \sh 0.947\\ 
\cline{1-6}\cline{8-13}
\end{tabular}
\caption{Simulation results for the conventional estimator $\boldsymbol{\hat\beta}^M$ and the unified estimator $\boldsymbol{\hat\beta}_L$ based on 2,500 Monte Carlo iterations. Results are displayed for the second coordinate $\boldsymbol{\beta}_{01}$ of $\boldsymbol{\beta}_{0}$ and for the sample size $N=1000$. While exact stayers are absent (i.e., $\pi_0=0.0$), the values of $1/\alpha \in \{1,2,3,4\}$ and $\rho \in \{0.5,1.0\}$ are varied across sets of simulations. The order $L=2$ of local polynomials are used for estimation of $\boldsymbol{\delta}$ and $\boldsymbol{\beta}_L$. Displayed statistics are the bias, root mean square error (RMSE), and 95\% coverage.}
\label{tab:simulations_bias_rmse_95_pi0_L2}
\end{table}

The contrast between panels (I) and (II) of Table \ref{tab:simulations_bias_rmse_95_pi0_L2} summarizes the aforementioned observations regarding the bias of the conventional estimator $\boldsymbol{\hat\beta}^M$ and the unbiasedness of the unified estimator $\boldsymbol{\hat\beta}_L$, but now across a larger set of data generating parameters.
Likewise, a comparison between panels (III) and (IV) of Table \ref{tab:simulations_bias_rmse_95_pi0_L2} implies that the root mean square of the unified estimator $\boldsymbol{\hat\beta}_L$ is no worse than that of the conventional estimator $\boldsymbol{\hat\beta}^M$ in general, and also that the former is strictly better as the sample size becomes larger.
Finally, a comparison between panels (V) and (VI) of Table \ref{tab:simulations_bias_rmse_95_pi0_L2} evidences that the undercoverage by the confidence intervals associated with the conventional estimator $\boldsymbol{\hat\beta}^M$ tends to become worse as $1/\alpha$ increases and/or $\rho$ increases, while the confidence intervals associated with the unified estimator $\boldsymbol{\hat\beta}_L$ robustly achieves the desired coverage probability regardless of the values of the data generating parameters, $1/\alpha \in \{1,2,3,4\}$ and $\rho \in \{0.5,1.0\}$.
Table \ref{tab:simulations_bias_rmse_95_pi0_L3} shows a counterpart of Table \ref{tab:simulations_bias_rmse_95_pi0_L2}, but with the degree $L=3$ of local polynomials instead of $L=2$.
The key patterns of the results displayed in Table \ref{tab:simulations_bias_rmse_95_pi0_L3} remain the same as those displayed in Table \ref{tab:simulations_bias_rmse_95_pi0_L2}.

\begin{table}
\centering
\begin{tabular}{|cc|cccc|c|cc|cccc|}
\multicolumn{13}{c}{Conventional Estimator versus Unified Estimator with Large Bandwidths and $L=3$}\\
\multicolumn{13}{c}{}\\
\multicolumn{6}{c}{(I) Bias: $\boldsymbol{\hat\beta}^M$} & \multicolumn{1}{c}{} & \multicolumn{6}{c}{(II) Bias: $\boldsymbol{\hat\beta}_L$}\\
\cline{1-6}\cline{8-13}
\multicolumn{2}{|c|}{       } & \multicolumn{4}{|c|}{$1/\alpha$} &&\multicolumn{2}{|c|}{} & \multicolumn{4}{|c|}{$1/\alpha$}\\
\multicolumn{2}{|c|}{$N=1000$} & 1 & 2 & 3 & 4 && \multicolumn{2}{|c|}{$N=1000$} & 1 & 2 & 3 & 4\\
\cline{1-6}\cline{8-13}
$\rho$ & 0.5 & 0.002 & 0.005 & 0.006 & 0.005 && $\rho$ & 0.5 & \sh 0.000 & \sh 0.000 & \sh 0.000 & \sh 0.000\\
       & 1.0 & 0.005 & 0.011 & 0.012 & 0.011 &&        & 1.0 & \sh 0.000 & \sh 0.000 & \sh 0.000 & \sh 0.001\\ 
\cline{1-6}\cline{8-13}
\multicolumn{13}{c}{}\\
\multicolumn{6}{c}{(III) RMSE: $\boldsymbol{\hat\beta}^M$} & \multicolumn{1}{c}{} & \multicolumn{6}{c}{(IV) RMSE: $\boldsymbol{\hat\beta}_L$}\\
\cline{1-6}\cline{8-13}
\multicolumn{2}{|c|}{       } & \multicolumn{4}{|c|}{$1/\alpha$} &&\multicolumn{2}{|c|}{} & \multicolumn{4}{|c|}{$1/\alpha$}\\
\multicolumn{2}{|c|}{$N=1000$} & 1 & 2 & 3 & 4 && \multicolumn{2}{|c|}{$N=1000$} & 1 & 2 & 3 & 4\\
\cline{1-6}\cline{8-13}
$\rho$ & 0.5 & 0.006 & 0.008 & 0.009 & 0.009 && $\rho$ & 0.5 & \sh 0.006 & \sh 0.007 & \sh 0.008 & \sh 0.009\\
       & 1.0 & 0.010 & 0.013 & 0.014 & 0.013 &&        & 1.0 & \sh 0.008 & \sh 0.007 & \sh 0.009 & \sh 0.009\\ 
\cline{1-6}\cline{8-13}
\multicolumn{13}{c}{}\\
\multicolumn{6}{c}{(V) 95\% Cover: $\boldsymbol{\hat\beta}^M$} & \multicolumn{1}{c}{} & \multicolumn{6}{c}{(VI) 95\% Cover: $\boldsymbol{\hat\beta}_L$}\\
\cline{1-6}\cline{8-13}
\multicolumn{2}{|c|}{       } & \multicolumn{4}{|c|}{$1/\alpha$} &&\multicolumn{2}{|c|}{} & \multicolumn{4}{|c|}{$1/\alpha$}\\
\multicolumn{2}{|c|}{$N=1000$} & 1 & 2 & 3 & 4 && \multicolumn{2}{|c|}{$N=1000$} & 1 & 2 & 3 & 4\\
\cline{1-6}\cline{8-13}
$\rho$ & 0.5 & 0.889 & 0.727 & 0.678 & 0.674 && $\rho$ & 0.5 & \sh 0.951 & \sh 0.947 & \sh 0.947 & \sh 0.937\\
       & 1.0 & 0.804 & 0.430 & 0.374 & 0.386 &&        & 1.0 & \sh 0.954 & \sh 0.951 & \sh 0.945 & \sh 0.951\\  
\cline{1-6}\cline{8-13}
\end{tabular}
\caption{Simulation results for the conventional estimator $\boldsymbol{\hat\beta}^M$ and the unified estimator $\boldsymbol{\hat\beta}_L$ based on 2,500 Monte Carlo iterations. Results are displayed for the second coordinate $\boldsymbol{\beta}_{01}$ of $\boldsymbol{\beta}_{0}$ and for the sample size $N=1000$. While exact stayers are absent (i.e., $\pi_0=0.0$), the values of $1/\alpha \in \{1,2,3,4\}$ and $\rho \in \{0.5,1.0\}$ are varied across sets of simulations. The order $L=3$ of local polynomials are used for estimation of $\boldsymbol{\delta}$ and $\boldsymbol{\beta}_L$. Displayed statistics are the bias, root mean square error (RMSE), and 95\% coverage.}
\label{tab:simulations_bias_rmse_95_pi0_L3}
\end{table}

Summarizing the simulation results so far based on the data generating designs with no exact stayers ($\pi_0=0$), we make the following points.
First, the conventional estimator $\boldsymbol{\hat\beta}^M$ is biased even in the absence of exact stayers, while the unified estimator $\boldsymbol{\hat\beta}_L$ is almost unbiased.
Second, the bias correction by the unified estimator $\boldsymbol{\hat\beta}_L$ is achieved without much sacrificing the variance in comparison with the conventional estimator $\boldsymbol{\hat\beta}^M$.
As a consequence, the root mean square error of the unified estimator $\boldsymbol{\hat\beta}_L$ is no worse than that of the conventional estimator $\boldsymbol{\hat\beta}^M$ in general especially as the sample size become larger.
Third, confidence intervals associated with the conventional estimator $\boldsymbol{\hat\beta}^M$ suffer from severe undercoverage even in the absence of exact stayers, while confidence intervals associated with the unified estimator $\boldsymbol{\hat\beta}_L$ almost perfectly achieves the coverage probability at the desired nominal probability.
We ran a lot more simulations whose results are not presented here for the purpose of conserving space, but the observed patterns summarized above remain to hold under all of the alternative simulation settings.

\subsubsection{With Exact Stayers}\label{sec:simulation_results_stayers}

We next turn to the set of simulation results based on the data generating processes with a mass of exact stayers, i.e., $\pi_0 > 0$.
In this setting, the conventional estimator $\boldsymbol{\hat\beta}^M$ that uses the information of only the movers is biased, and hence the consistency of $\boldsymbol{\hat\beta}^M$ is out of the question.
Our unified estimator $\boldsymbol{\hat\beta}_L$ which generalizes the mixture estimator of \citet[][Section 3.2]{graham/powell:2012} continues to be asymptotically unbiased and consistent.
In this light, our discussions hereafter disregard comparisons with $\boldsymbol{\hat\beta}^M$, and instead focus on the performance of the unified estimator $\boldsymbol{\hat\beta}_L$.

Table \ref{tab:simulations_bias_rmse_95_stayer} summarizes simulation results under data generating designs with $\pi_0 = 0.1$ (left column) and $\pi_0 = 0.2$ (right column), allowing for 10\% and 20\%, respectively, of exact stayers.
Displayed statistics are the bias (panels (I)--(II)), root mean square error (panels (III)--(IV)), and 95\% coverage (panels (V)--(VI)) based on the unified estimator $\boldsymbol{\hat\beta}_L$ with the degree $L=2$ of local polynomials.
The results are displayed across combinations, $\pi_0 \in \{0.1,0.2\}$, $\alpha \in \{1,2,3,4\}$, and $\rho \in \{0.5,1.0\}$, of the data generating parameters and sample size $N=1000$.

\begin{table}
\centering
\begin{tabular}{|cc|cccc|c|cc|cccc|}
\multicolumn{13}{c}{Robust Performance of the Unified Estimator in the Presence of Exact Stayers}\\
\multicolumn{13}{c}{}\\
\multicolumn{6}{c}{(I) Bias: $\pi_0=0.1$} & \multicolumn{1}{c}{} & \multicolumn{6}{c}{(II) Bias: $\pi_0=0.2$}\\
\cline{1-6}\cline{8-13}
\multicolumn{2}{|c|}{       } & \multicolumn{4}{|c|}{$1/\alpha$} &&\multicolumn{2}{|c|}{} & \multicolumn{4}{|c|}{$1/\alpha$}\\
\multicolumn{2}{|c|}{$N=1000$} & 1 & 2 & 3 & 4 && \multicolumn{2}{|c|}{$N=1000$} & 1 & 2 & 3 & 4\\
\cline{1-6}\cline{8-13}
$\rho$ & 0.5 & \sh 0.000 & \sh 0.000 & \sh 0.000 & \sh 0.000 && $\rho$ & 0.5 & \sh 0.000 & \sh 0.000 & \sh 0.000 & \sh 0.000\\
       & 1.0 & \sh 0.001 & \sh 0.001 & \sh 0.001 & \sh 0.000 &&        & 1.0 & \sh 0.001 & \sh 0.001 & \sh 0.000 & \sh 0.000\\  
\cline{1-6}\cline{8-13}
\multicolumn{13}{c}{}\\
%
\multicolumn{13}{c}{}\\
\multicolumn{6}{c}{(III) RMSE: $\pi_0=0.1$} & \multicolumn{1}{c}{} & \multicolumn{6}{c}{(IV) RMSE: $\pi_0=0.2$}\\
\cline{1-6}\cline{8-13}
\multicolumn{2}{|c|}{       } & \multicolumn{4}{|c|}{$1/\alpha$} &&\multicolumn{2}{|c|}{} & \multicolumn{4}{|c|}{$1/\alpha$}\\
\multicolumn{2}{|c|}{$N=1000$} & 1 & 2 & 3 & 4 && \multicolumn{2}{|c|}{$N=1000$} & 1 & 2 & 3 & 4\\
\cline{1-6}\cline{8-13}
$\rho$ & 0.5 & \sh 0.007 & \sh 0.007 & \sh 0.009 & \sh 0.010 && $\rho$ & 0.5 & \sh 0.009 & \sh 0.010 & \sh 0.011 & \sh 0.012\\
       & 1.0 & \sh 0.007 & \sh 0.008 & \sh 0.009 & \sh 0.010 &&        & 1.0 & \sh 0.010 & \sh 0.010 & \sh 0.011 & \sh 0.011\\  
\cline{1-6}\cline{8-13}
\multicolumn{13}{c}{}\\
\multicolumn{6}{c}{(V) 95\% Cover: $\pi_0=0.1$} & \multicolumn{1}{c}{} & \multicolumn{6}{c}{(IV) 95\% Cover: $\pi_0=0.2$}\\
\cline{1-6}\cline{8-13}
\multicolumn{2}{|c|}{       } & \multicolumn{4}{|c|}{$1/\alpha$} &&\multicolumn{2}{|c|}{} & \multicolumn{4}{|c|}{$1/\alpha$}\\
\multicolumn{2}{|c|}{$N=1000$} & 1 & 2 & 3 & 4 && \multicolumn{2}{|c|}{$N=1000$} & 1 & 2 & 3 & 4\\
\cline{1-6}\cline{8-13}
$\rho$ & 0.5 & \sh 0.945 & \sh 0.950 & \sh 0.946 & \sh 0.944 && $\rho$ & 0.5 & \sh 0.931 & \sh 0.936 & \sh 0.944 & \sh 0.942\\
       & 1.0 & \sh 0.948 & \sh 0.952 & \sh 0.948 & \sh 0.940 &&        & 1.0 & \sh 0.937 & \sh 0.942 & \sh 0.943 & \sh 0.941\\   
\cline{1-6}\cline{8-13}
\end{tabular}
\caption{Simulation results for the unified estimator $\boldsymbol{\hat\beta}_L$ in the presence of exact stayers based on 2,500 Monte Carlo iterations. Results are displayed for the second coordinate $\boldsymbol{\beta}_{01}$ of $\boldsymbol{\beta}_{0}$ and for the sample size $N=1000$. The values of $\pi \in \{0.1,0.2\}$, $1/\alpha \in \{1,2,3,4\}$, and $\rho \in \{0.5,1.0\}$ are varied across sets of simulations. The order $L=2$ of local polynomials are used for estimation of $\boldsymbol{\delta}$ and $\boldsymbol{\beta}_L$. Displayed statistics are the bias, root mean square error (RMSE), and 95\% coverage.}
\label{tab:simulations_bias_rmse_95_stayer}
\end{table}

We make the following observations in Table \ref{tab:simulations_bias_rmse_95_stayer}.
First, the estimator remains unbiased even in the non-atomic presence of exact stayers (panels (I)--(II)).
The magnitude of the bias is invariant from the sample size.
Second, as a consequence of the bias reduction, the confidence intervals still achieve coverage frequencies that are close to the nominal probability of 95\% (panels (V)--(VI)).
The coverage accuracy improves as the sample size increases.
Third, the root mean square error becomes larger as the proportion $\pi_0$ of exact stayers increases (panels (III)--(IV)).
Specifically, for instance, observe that the numbers gradually increase as we move from Table \ref{tab:simulations_bias_rmse_95_pi0_L2} (IV) in which $\pi_0=0.0$ to Table \ref{tab:simulations_bias_rmse_95_stayer} (III) in which $\pi_0=0.1$, and then to Table \ref{tab:simulations_bias_rmse_95_stayer} (IV) in which $\pi_0=0.2$.
This result concerning the root mean square error is natural, because the mass of stayers fails to provide data variations and a larger proportion $\pi_0$ means more mass of such stayers exhibiting no data variations.
Finally, we remark that we ran more simulations whose results are not presented here for the purpose of conserving space, but the observed patterns summarized above remain to hold under all of the alternative simulation settings.

\section{The Robust Method of Estimation and Inference with $T>p$}\label{section:T>p}

Thus far, we have focused on the case with $T=p$.
In this section, we modify the methodology in the previous sections to the case of $T>p$. 
In this framework, we redefine $D$, $\mathbf{D}_{1:L}$,  $\mathbf{m}$, $\boldsymbol{\hat\theta}$, $\boldsymbol{\hat\beta}_L$, $ \boldsymbol{\hat\delta}$, $\boldsymbol{\hat\gamma}$, $\mathbf{\hat{h}}$, $\boldsymbol{\hat\zeta}$ differently from the ones in the previous sections. 
In the current section, we only introduce an estimator and its estimated influence function. 
We present formal theoretical results to support this procedure in Appendix \ref{supp:sec:main}. 

The correlated random coefficient panel data model is the same as in the previous sections: $Y_t=\mathbf{X}_t'b_t(A,U_t)$ for $t=1,\ldots,T$. The short-hand notations such as $\mathbf{Y}$, $\mathbf{X}$, $\mathbf{W}$, $\boldsymbol{\beta}(\mathbf{X})$, $\boldsymbol{\beta}$, and $\boldsymbol{\delta}$, also remain the same as in the previous sections. 
Using a matrix $\mathbf{R}$ that extracts time effects as before, suppose that the parameter of interest is $\boldsymbol{\theta}=\boldsymbol{\beta}+\mathbf{R}\boldsymbol{\delta}$. 
The parameter $(\boldsymbol{\beta}',\boldsymbol{\delta}')'$ can be identified based on the following implications from Assumption \ref{assn:GP1.1}: 
\begin{align}
&E[\boldsymbol{1}\{D>h_N\}\mathbf{W}\mathbf{M}_{\mathbf{X}}\mathbf{Y}]=E[\boldsymbol{1}\{D>h_N\}\mathbf{W}'\mathbf{M}_{\mathbf{X}}\mathbf{W}]\boldsymbol{\delta}
\qquad\text{and}
\label{supp:eq:base_ID1}
\\
&E[(\mathbf{X}'\mathbf{X})^\ast\mathbf{X}'(\mathbf{Y}-\mathbf{W}\boldsymbol{\delta})\mid\mathbf{X}]=D\boldsymbol{\beta}(\mathbf{X})\mbox{ with }\boldsymbol{\beta}=E[\boldsymbol{\beta}(\mathbf{X})],
\label{supp:eq:base_ID}
\end{align}
where $D=\mathrm{det}(\mathbf{X}'\mathbf{X})$, $\mathbf{M}_{\mathbf{X}}=\mathbf{I}-\mathbf{X}(\mathbf{X}'\mathbf{X})^{-1}\mathbf{X}'$ and $(\mathbf{X}'\mathbf{X})^\ast$ denotes the adjoint of $\mathbf{X}'\mathbf{X}$. 

We can estimate the parameter $\boldsymbol{\theta}=\boldsymbol{\beta}+\mathbf{R}\boldsymbol{\delta}$  by $\boldsymbol{\hat\theta}=\boldsymbol{\hat\beta}_L+\mathbf{R}\boldsymbol{\hat\delta}$ with 
\begin{align*}
\boldsymbol{\hat\delta}
=&
E_N[\boldsymbol{1}\{D>h_N\}\mathbf{W}'\mathbf{M}_{\mathbf{X}}\mathbf{W}]^{-1}E_N[\boldsymbol{1}\{D>h_N\}\mathbf{W}\mathbf{M}_{\mathbf{X}}\mathbf{Y}]
\qquad\text{and }
\\
\boldsymbol{\hat\beta}_L
=&
E_N[\boldsymbol{1}\{D>h_N\}D^{-1}(\mathbf{X}'\mathbf{X})^\ast\mathbf{X}'(\mathbf{Y}-\mathbf{W}\boldsymbol{\hat\delta})]+\boldsymbol{\hat\gamma}\mathbf{\hat{h}},
\end{align*}
where  $\mathbf{\hat{h}}=E_N[\boldsymbol{1}\{D\leq h_N\}(\begin{array}{cccc}1&\cdots&D^{L-1}\end{array})']$, $\mathbf{D}_{1:L}=\boldsymbol{1}\{D\leq h_N\}(\begin{array}{cccc}D&\cdots&D^L\end{array})'$,
and 
$$
\boldsymbol{\hat\gamma}
=
E_N\left[(\mathbf{X}'\mathbf{X})^\ast\mathbf{X}'(\mathbf{Y}-\mathbf{W}\boldsymbol{\hat\delta})\mathbf{D}_{1:L}'\right]E_N\left[\mathbf{D}_{1:L}\mathbf{D}_{1:L}'\right]^{-1}.
$$
The influence function for $\boldsymbol{\hat\theta}$ is estimated by 
\begin{align*}
\boldsymbol{\hat\zeta}
=&
\boldsymbol{1}\{D>h_N\}D^{-1}(\mathbf{X}'\mathbf{X})^\ast\mathbf{X}'(\mathbf{Y}-\mathbf{W}\boldsymbol{\hat\delta})-E_N[\boldsymbol{1}\{D>h_N\}D^{-1}(\mathbf{X}'\mathbf{X})^\ast\mathbf{X}'(\mathbf{Y}-\mathbf{W}\boldsymbol{\hat\delta})]
\\&
+
((\mathbf{X}'\mathbf{X})^\ast\mathbf{X}'(\mathbf{Y}-\mathbf{W}\boldsymbol{\hat\delta})-\boldsymbol{\hat\gamma}\mathbf{D}_{1:L})\mathbf{D}_{1:L}'E_N\left[\mathbf{D}_{1:L}\mathbf{D}_{1:L}'\right]^{-1}\mathbf{\hat{h}}
\\&
+
\mathbf{\hat{Q}}\mathbf{\hat{V}}^{-1}\boldsymbol{1}\{D>h_N\}\mathbf{W}\mathbf{M}_{\mathbf{X}}(\mathbf{Y}-\mathbf{W}\boldsymbol{\hat\delta}),
\end{align*}
where 
\begin{align*}
\mathbf{\hat{V}}
=&
E_N[\boldsymbol{1}\{D>h_N\}\mathbf{W}'\mathbf{M}_{\mathbf{X}}\mathbf{W}],\mbox{ and }
\\
\mathbf{\hat{Q}}
=&
\mathbf{R}-E_N\left[\left(\boldsymbol{1}\{D>h_N\}D^{-1}+\mathbf{D}_{1:L}'E_N\left[\mathbf{D}_{1:L}\mathbf{D}_{1:L}'\right]^{-1}\mathbf{\hat{h}}\right)(\mathbf{X}'\mathbf{X})^\ast\mathbf{X}'\mathbf{W}\right].
\end{align*}

With these notations, we establish the asymptotic normality
$$
\sqrt{N}E_N\left[\boldsymbol{\hat\zeta}\boldsymbol{\hat\zeta}'\right]^{-1/2}(\boldsymbol{\hat\theta}-\boldsymbol{\theta})\rightarrow_d\mathcal{N}\left(0,\mathbf{I}_{p}\right).
$$ 
See Theorem \ref{supp:theorem_asynormal} in Appendix \ref{supp:sec:main} for a formal statement of this result.

\section{Conclusion}\label{sec:conclusion}

Panel data often contain units with no within-variations (stayers) and units with little within-variations (slow movers). Figure \ref{fig:PSID_17_19} demonstrates a case in point in the log of total family income in the U.S. Panel Survey of Income Dynamics (PSID) between 2017 and 2019, where there are many stayers and many slow movers. In the presence of many slow movers, as in Figure \ref{fig:PSID_17_19}, conventional econometric methods can fail to work. In this paper, we propose a novel method of robust inference for the average partial effects in correlated random coefficient models robustly across various distributions of within-variations, including the cases with many stayers and/or many slow movers. The extent of this robustness encompasses all of the eight cases illustrated in Figure \ref{fig:eight_cases}, and a single method of estimation and inference proposed in this paper can handle all these eight cases in a unified manner. In contrast, the two existing methods focus separately on the case illustrated in Figure \ref{fig:eight_cases} (a) and the case illustrated in Figure \ref{fig:eight_cases} (b). In addition to the robustness property, our proposed method enjoys smaller biases and hence improved accuracy in inference compared to existing methods. Simulation studies demonstrate theoretical claims about these preferred properties of the method proposed in this paper. Specifically, the proposed method enjoys smaller biases and more accurate coverage frequencies than the conventional method robustly across various data generating processes. 
In our simulation studies, our proposed 95\% confidence interval achieves 93-96\% coverage frequencies, as opposed to 37-93\% coverage frequencies for the conventional method. 
While most parts of the paper focus on the baseline case in which $T=p$, we also provide an extension to the case with $T>p$.

\bibliography{mybib}

\appendix
\section{Proofs} 

This section collects proofs of the main results together with auxiliary lemmas and their proofs.

\subsection{Proof of Theorem \ref{bias_comp_higher}} 

\begin{proof}
Note that, by Assumption \ref{assn:more_primitiv_higher}, there is a constant $C < \infty$ such that  
$$
\left\|\mathbf{m}(u)-\sum_{l=0}^{L}\mathbf{m}^{(l)}(0)\frac{u^{l}}{l!}\right\|\leq C|u|^{L+1}
$$
for every $u$ in a neighborhood of zero. 
Using $\mathbf{m}(u)=uE[\boldsymbol{\beta}(\mathbf{X})\mid D=u]$, we can decompose $\mathbf{b}_L$ and $\boldsymbol{\beta}$ as
\begin{align*}
\mathbf{b}_L
=&
E[\mathbf{X}^{-1}(\mathbf{Y}-\mathbf{W}\boldsymbol{\delta})\boldsymbol{1}\{|D|>h_N\}]
+
E[\boldsymbol{\beta}(\mathbf{X})\boldsymbol{1}\{D=0\}]
+
E\left[\sum_{l=1}^L\mathbf{m}^{(l)}(0)\frac{D^{l-1}}{l!}\boldsymbol{1}\{0<|D|\leq h_N\}\right],
\quad\text{and}\\
\boldsymbol{\beta}
=& 
E[\mathbf{X}^{-1}(\mathbf{Y}-\mathbf{W}\boldsymbol{\delta})\boldsymbol{1}\{|D|>h_N\}]
+
E[\boldsymbol{\beta}(\mathbf{X})\boldsymbol{1}\{D=0\}]+
E\left[\frac{\mathbf{m}(D)}{D}\boldsymbol{1}\{0<|D|\leq h_N\}\right].
\end{align*}
Taking the difference and noting that $\phi$ is the conditional density of the absolutely continuous component of the distribution of $D$ by Assumption \ref{assn:GP_around0}, we have
$$
\left\|\mathbf{b}_L-\boldsymbol{\beta}\right\|
\leq
\int_{-h_N}^{h_N}\left\|\frac{\mathbf{m}(u)-\mathbf{m}(0)-\sum_{l=1}^L\mathbf{m}^{(l)}(0)\frac{u^{l}}{l!}}{u}\right\|\phi(u)du
\leq
C\int_{-h_N}^{h_N}|u|^{L}\phi(u)du,
$$
where the first inequality uses $\mathbf{m}(0)=0$. 
\end{proof}

\subsection{Proof of Theorem \ref{theorem_asynormal}} 

Before proceeding with a proof of Theorem \ref{theorem_asynormal}, we first introduce short-hand notations to simplify the proof. 
When we derive the influence function representation, we control the remainder term as $o_p\left(\sqrt{\frac{\pi_h^2}{Nh_N^2\pi_{h\setminus 0}}}\right)$ with $\pi_h=P(|D|\leq h_N)$ and $\pi_{h\setminus 0}=P(0<|D|\leq h_N)$. 
For $\left(\boldsymbol{\hat\delta},\boldsymbol{\hat\gamma}\right)$, the design matrices, $E_N\left[\mathbf{D}_{0:L}\mathbf{D}_{0:L}'\right]$ and $E_N\left[\mathbf{D}_{1:L}\mathbf{D}_{1:L}'\right]$, are used. As in the standard local polynomial regression, the inverses of these design matrices are not bounded, so we normalize them to ensure that the inverses are bounded. The resulting matrices are 
\begin{align*}
\boldsymbol{\hat\Omega}=&\frac{1}{\pi_h}\boldsymbol{\Psi}_{0:L}^{-1}E_N\left[\mathbf{D}_{0:L}\mathbf{D}_{0:L}'\right]\boldsymbol{\Psi}_{0:L}^{-1}
\qquad\mbox{ and }\\
\boldsymbol{\hat\Pi}=&\frac{1}{\pi_{h\setminus 0}}\boldsymbol{\Psi}_{1:L}^{-1}E_N\left[\mathbf{D}_{1:L}\mathbf{D}_{1:L}'\right]\boldsymbol{\Psi}_{1:L}^{-1},
\end{align*}
where $\boldsymbol{\Psi}_{0:L}=\mathrm{diag}\left(\begin{array}{cccc}1&h_N&\cdots&h_N^L\end{array}\right)$ and $\boldsymbol{\Psi}_{1:L}=\mathrm{diag}\left(\begin{array}{cccc}h_N&\cdots&h_N^L\end{array}\right)$.
Their population counterparts are 
\begin{align*}
\boldsymbol{\Omega}=&\frac{1}{\pi_h}\boldsymbol{\Psi}_{0:L}^{-1}E\left[\mathbf{D}_{0:L}\mathbf{D}_{0:L}'\right]\boldsymbol{\Psi}_{0:L}^{-1}
\qquad\text{and}\\
\boldsymbol{\Pi}=&\frac{1}{\pi_{h\setminus 0}}\boldsymbol{\Psi}_{1:L}^{-1}E\left[\mathbf{D}_{1:L}\mathbf{D}_{1:L}'\right]\boldsymbol{\Psi}_{1:L}^{-1}.
\end{align*}
In the proof, we characterize the influence function for $\boldsymbol{\hat\beta}_L+\mathbf{R}\boldsymbol{\hat\delta}$ by 
$$
\boldsymbol{\zeta}=\boldsymbol{\xi}_1+\frac{1}{\pi_{h\setminus 0}}\boldsymbol{\xi}_2\boldsymbol{\Pi}^{-1}\boldsymbol{\Psi}_{1:L}^{-1}\mathbf{h}+\frac{1}{\pi_h}\mathbf{Q}\mathbf{V}^{-1}\boldsymbol{\xi}_3\boldsymbol{\Omega}^{-1}\mathbf{e}_1,
$$ 
where 
\begin{eqnarray*}
\mathbf{V}
&=&
E\left[\left(\mathbf{D}_{0:L}'E\left[\mathbf{D}_{0:L}\mathbf{D}_{0:L}'\right]^{-1}\mathbf{e}_1\right)(\mathbf{X}^\ast\mathbf{W})'\mathbf{X}^\ast\mathbf{W}\right],
\\
\mathbf{Q}
&=&
\mathbf{R}-E\left[\left(\boldsymbol{1}\{|D|>h_N\}D^{-1}+\mathbf{D}_{1:L}'E\left[\mathbf{D}_{1:L}\mathbf{D}_{1:L}'\right]^{-1}\mathbf{h}\right)\mathbf{X}^\ast\mathbf{W}\right],
\\
\boldsymbol{\xi}_1
&=&
\boldsymbol{1}\{|D|>h_N\}D^{-1}\mathbf{X}^\ast(\mathbf{Y}-\mathbf{W}\boldsymbol{\delta})-E[\boldsymbol{1}\{|D|>h_N\}D^{-1}\mathbf{X}^\ast(\mathbf{Y}-\mathbf{W}\boldsymbol{\delta})],
\\
\boldsymbol{\xi}_2
&=&
(\mathbf{X}^\ast(\mathbf{Y}-\mathbf{W}\boldsymbol{\delta})-\boldsymbol{\gamma}\mathbf{D}_{1:L})\mathbf{D}_{1:L}'\boldsymbol{\Psi}_{1:L}^{-1},
\qquad\text{and}
\\
\boldsymbol{\xi}_3
&=&
(\mathbf{X}^\ast\mathbf{W})'\mathbf{X}^\ast(\mathbf{Y}-\mathbf{W}\boldsymbol{\delta})\mathbf{D}_{0:L}'\boldsymbol{\Psi}_{0:L}^{-1}.
\end{eqnarray*}
We also use the sample counterparts of $\boldsymbol{\xi}_1$, $\boldsymbol{\xi}_2$ and $\boldsymbol{\xi}_3$:
\begin{eqnarray*}
\boldsymbol{\hat\xi}_1
&=&
\boldsymbol{1}\{|D|>h_N\}D^{-1}\mathbf{X}^\ast(\mathbf{Y}-\mathbf{W}\boldsymbol{\hat\delta})-E_N[\boldsymbol{1}\{|D|>h_N\}D^{-1}\mathbf{X}^\ast(\mathbf{Y}-\mathbf{W}\boldsymbol{\hat\delta})],
\\
\boldsymbol{\hat\xi}_2
&=&
(\mathbf{X}^\ast(\mathbf{Y}-\mathbf{W}\boldsymbol{\hat\delta})-\boldsymbol{\hat\gamma}\mathbf{D}_{1:L})\mathbf{D}_{1:L}'\boldsymbol{\Psi}_{1:L}^{-1},
\qquad\text{and}
\\
\boldsymbol{\hat\xi}_3
&=&
(\mathbf{X}^\ast\mathbf{W})'\mathbf{X}^\ast(\mathbf{Y}-\mathbf{W}\boldsymbol{\hat\delta})\mathbf{D}_{0:L}'\boldsymbol{\Psi}_{0:L}^{-1}.
\end{eqnarray*}
In the proof, we use the operator norm (based on the Euclidean norms) as the norm for matrices.
All the asymptotic results in the proof hold as $h_N\rightarrow 0$ and $N\rightarrow\infty$.   

It is worthwhile discussing $(\pi_{h\setminus 0},\pi_h)$. 
For any sample size, we have $\pi_{h\setminus 0}\leq \pi_h$. 
Also, since $\phi$ is bounded away from zero, the term $\pi_{h\setminus 0}/h_N$ is also bounded away from zero. 

With these notations introduced, we are now going to show Theorem \ref{theorem_asynormal}. 
In the proof, we use the auxiliary lemmas collected in Appendix \ref{sec:aux_results}. 

\bigskip
\begin{proof}[Proof of Theorem \ref{theorem_asynormal}]
First, we are going to show 
\begin{equation}\label{eq:var_estimation}
E_N\left[\boldsymbol{\hat\zeta}\boldsymbol{\hat\zeta}'\right]^{-1}Var(\boldsymbol{\zeta})=\mathbf{I}_{p}+o_p(1).
\end{equation}
By  \eqref{assn:bandwidth} and Lemmas \ref{lemma:zeta_moments} and \ref{lemma:xi_estimates}, 
\begin{eqnarray*}
E_N[\|\boldsymbol{\hat\zeta}-\boldsymbol{\zeta}\|^2]^{1/2}
&=&
o_p\left(\sqrt{\frac{\pi_h^{3+\epsilon}}{h_N^2\pi_{h\setminus 0}}}\right),
\\
E_N[\|\boldsymbol{\zeta}\|^2]^{1/2}
&=&
O_p\left(\sqrt{\frac{\pi_h^{1-\epsilon}}{h_N^2}}\right),
\\
(E_N-E)\left[\boldsymbol{\zeta}\boldsymbol{\zeta}'\right]
&=&
o_p\left(\frac{\pi_h^2}{h_N^2\pi_{h\setminus 0}}\right),
\qquad\text{and}
\\
E[\boldsymbol{\zeta}]
&=&
o\left(1\right),
\end{eqnarray*}
where we can take a sufficiently small $\epsilon>0$. 
Therefore, 
\begin{eqnarray*}
\|E_N\left[\boldsymbol{\hat\zeta}\boldsymbol{\hat\zeta}'\right]-Var(\boldsymbol{\zeta})\|
&\leq&
E_N[\|\boldsymbol{\hat\zeta}-\boldsymbol{\zeta}\|^2]+2E_N[\|\boldsymbol{\hat\zeta}-\boldsymbol{\zeta}\|^2]^{1/2}E_N[\|\boldsymbol{\zeta}\|^2]^{1/2}+\|(E_N-E)\left[\boldsymbol{\zeta}\boldsymbol{\zeta}'\right]\|+\|E[\boldsymbol{\zeta}]\|^2,
\\
&=&
o_p\left(\frac{\pi_h^2}{h_N^2\pi_{h\setminus 0}}\right).
\end{eqnarray*}
By Lemma \ref{lemma:lowerconrate}, we have $Var(\boldsymbol{\zeta})^{-1}=O(\frac{h_N^2\pi_{h\setminus 0}}{\pi_h^2})$, and so 
$$
E_N\left[\boldsymbol{\hat\zeta}\boldsymbol{\hat\zeta}'\right]^{-1}Var(\boldsymbol{\zeta})-\mathbf{I}_{p}
=
-\left(Var(\boldsymbol{\zeta})+\left(E_N\left[\boldsymbol{\hat\zeta}\boldsymbol{\hat\zeta}'\right]-Var(\boldsymbol{\zeta})\right)\right)^{-1}\left(E_N\left[\boldsymbol{\hat\zeta}\boldsymbol{\hat\zeta}'\right]-Var(\boldsymbol{\zeta})\right)
=
o_p(1).
$$

Second, we are going to show 
\begin{equation}\label{eq:influeFunct_estimation}
\boldsymbol{\hat\theta}-\boldsymbol{\theta}-(E_N-E)[\boldsymbol{\zeta}]=o_p\left(\sqrt{\frac{\pi_h^2}{Nh_N^2\pi_{h\setminus 0}}}\right).
\end{equation}
By Lemmas \ref{lemma:influe_eror4}, \ref{lemma:Karamata's1}, \ref{lemma:Omega_property2}, \ref{lemma:Omega_property}, \ref{lemma:M_inve}, \ref{lemma:Sigma_rate}, \ref{lemma:bias_bound1}, \ref{lemma:bias_bound2}, and Theorem \ref{bias_comp_higher}, 
\begin{eqnarray*}
&&
\|\boldsymbol{\hat\theta}-\boldsymbol{\theta}-(E_N-E)[\boldsymbol{\zeta}]\|
\\
&\leq&
\|\mathbf{b}_L-\boldsymbol{\beta}+E[\boldsymbol{\zeta}]\|
\\&&
+\frac{1}{\pi_{h\setminus 0}}\left(\left\|\boldsymbol{\hat\Pi}^{-1}-\boldsymbol{\Pi}^{-1}\right\|\|\boldsymbol{\Psi}_{1:L}^{-1}\mathbf{\hat{h}}\|+\|\boldsymbol{\Pi}^{-1}\|\|\boldsymbol{\Psi}_{1:L}^{-1}\left(\mathbf{\hat{h}}-\mathbf{h}\right)\|\right)\|E_N\left[\boldsymbol{\xi}_2\right]\|
\\&&
+\frac{1}{\pi_h}(\|\mathbf{\hat{Q}}\|(\|\mathbf{\hat{V}}^{-1}\|\|(\boldsymbol{\hat\Omega}^{-1}-\boldsymbol{\Omega}^{-1})\mathbf{e}_1\|+\|\mathbf{\hat{V}}^{-1}-\mathbf{{V}}^{-1}\|\|\boldsymbol{\Omega}^{-1}\mathbf{e}_1\|)+\|\mathbf{\hat{Q}}-\mathbf{Q}\|\|\mathbf{{V}}^{-1}\|\|\boldsymbol{\Omega}^{-1}\mathbf{e}_1\|)\|E_N\left[\boldsymbol{\xi}_3\right]\|
\\&&
+\|\boldsymbol{\gamma}\|\left\|\mathbf{\hat{h}}-\mathbf{h}\right\|
\\
&=&
O\left(h_N^L\right)+O_p\left(\sqrt{\frac{\pi_h^{2}}{Nh_N^2}}+
\sqrt{\frac{\pi_h^2}{N^2h_N^4\pi_{h\setminus 0}}}\right).
\end{eqnarray*}
By \eqref{assn:bandwidth}, we have \eqref{eq:influeFunct_estimation}. 

Third, we are going to show 
\begin{equation}\label{eq:L2L4condition_new}
E[\|Var(\boldsymbol{\zeta})^{-1/2}\boldsymbol{\zeta}\|^4]=o(N).
\end{equation}
By Lemmas \ref{lemma:zeta_moments} and \ref{lemma:lowerconrate}, 
$$
E[\|Var(\boldsymbol{\zeta})^{-1/2}\boldsymbol{\zeta}\|^4]=O\left(\frac{1}{\pi_h^2}
\right).
$$ 
By \eqref{assn:bandwidth}, we obtain \eqref{eq:L2L4condition_new}. 

Now, we are going to show the statement of the theorem. 
By Lemmas \ref{lemma:zeta_moments} and \ref{lemma:lowerconrate}, and \eqref{eq:var_estimation} and \eqref{eq:influeFunct_estimation}, 
\begin{eqnarray*}
\sqrt{N}E_N\left[\boldsymbol{\hat\zeta}\boldsymbol{\hat\zeta}'\right]^{-1/2}(\boldsymbol{\hat\theta}-\boldsymbol{\theta})
&=&
\sqrt{N}Var(\boldsymbol{\zeta})^{-1/2}(E_N-E)[\boldsymbol{\zeta}]
\\&&
+\sqrt{N}E_N\left[\boldsymbol{\hat\zeta}\boldsymbol{\hat\zeta}'\right]^{-1/2}\left(\boldsymbol{\hat\theta}-\boldsymbol{\theta}-(E_N-E)[\boldsymbol{\zeta}]\right)
\\&&
+\sqrt{N}\left(E_N\left[\boldsymbol{\hat\zeta}\boldsymbol{\hat\zeta}'\right]^{-1/2}-Var(\boldsymbol{\zeta})^{-1/2}\right)(E_N-E)[\boldsymbol{\zeta}]
\\
&=&
\sqrt{N}Var(\boldsymbol{\zeta})^{-1/2}(E_N-E)[\boldsymbol{\zeta}]+o_p(1).
\end{eqnarray*}
By \eqref{eq:L2L4condition_new}, Assumption \ref{assn:iid_plus} (i), and Lyapunov's central limit theorem, $\sqrt{N}E_N\left[\boldsymbol{\hat\zeta}\boldsymbol{\hat\zeta}'\right]^{-1/2}(\boldsymbol{\hat\theta}-\boldsymbol{\theta})\rightarrow_d\mathcal{N}\left(0,\mathbf{I}_{p}\right)$.
\end{proof}

\section{Auxiliary Lemmas for the Proof of Theorem \ref{theorem_asynormal}}\label{sec:aux_results}

We introduce and prove the auxiliary results, which we use in the proof of Theorem \ref{theorem_asynormal}. 
In all the lemmas to be stated in this section, we impose the assumptions invoked in Theorem \ref{theorem_asynormal} as well as $T=p$.

\begin{lemma}\label{lemma:influe_eror4}
$\boldsymbol{\Psi}_{1:L}^{-1}\mathbf{h}=O(\frac{\pi_h}{h_N})$, $\left\|\frac{h_N}{\pi_h}\boldsymbol{\Psi}_{1:L}^{-1}\mathbf{h}\right\|$ is bounded away from zero and $\boldsymbol{\Psi}_{1:L}^{-1}(\mathbf{\hat{h}}-\mathbf{h})=O_p\left(\sqrt{\frac{\pi_h}{Nh_N^2}}\right)$.
\end{lemma}
\begin{proof}
The statement follows from 
$$
\frac{h_N}{\pi_h}\boldsymbol{\Psi}_{1:L}^{-1}\mathbf{h}
=
\frac{1}{\pi_h}E\left[\boldsymbol{1}\{|D|\leq h_N\}\left(\begin{array}{cccc}1&\cdots&\frac{D^{L-1}}{h_N^{L-1}}\end{array}\right)'\right]
=
O\left(1\right),
$$
$$
\mathbf{e}_1'\left(\frac{h_N}{\pi_h}\boldsymbol{\Psi}_{1:L}^{-1}\mathbf{h}\right)
=
\frac{1}{\pi_h}E\left[\boldsymbol{1}\{|D|\leq h_N\}\right]
=
1,
$$
and 
\begin{eqnarray*}
E\left[\left\|\boldsymbol{\Psi}_{1:L}^{-1}(\mathbf{\hat{h}}-\mathbf{h})\right\|^2\right]
&=&
E\left[\left\|\frac{1}{h_N}(E_N-E)\left[\boldsymbol{1}\{|D|\leq h_N\}\left(\begin{array}{cccc}1&\cdots&\frac{D^{L-1}}{h_N^{L-1}}\end{array}\right)'\right]\right\|^2\right]
\\
&=&
\frac{1}{Nh_N^2}E\left[\left\|\boldsymbol{1}\{|D|\leq h_N\}\left(\begin{array}{cccc}1&\cdots&\frac{D^{L-1}}{h_N^{L-1}}\end{array}\right)'\right\|^2\right]
\\
&=&
O\left(\frac{\pi_h}{Nh_N^2}\right)
\end{eqnarray*}
under Assumption \ref{assn:iid_plus} (i). 
\end{proof}

\begin{lemma}\label{lemma:Karamata's1}
(i) For every integer $l\geq 1$,  
$$
\frac{1}{h_N^l\pi_{h\setminus 0}} E[D^l\boldsymbol{1}\{|D|\leq h_N\}] =\frac{\alpha_1}{l+\alpha_1}\frac{P(0<D\leq h_N)}{\pi_{h\setminus 0}}+(-1)^l\frac{\alpha_2}{l+\alpha_2}\frac{P(-h_N\leq D<0)}{\pi_{h\setminus 0}}+o(1)
$$ 
and 
$$
\frac{1}{h_N^l\pi_{h\setminus 0}} E[|D|^l\boldsymbol{1}\{|D|\leq h_N\}] =\frac{\alpha_1}{l+\alpha_1}\frac{P(0<D\leq h_N)}{\pi_{h\setminus 0}}+\frac{\alpha_2}{l+\alpha_2}\frac{P(-h_N\leq D<0)}{\pi_{h\setminus 0}}+o(1).
$$
(ii) For every positive integer $l$,  $E\left[\left\|\mathbf{D}_{0:L}'\boldsymbol{\Psi}_{0:L}^{-1}\right\|^l\right]=O(\pi_h)$ and $E\left[\left\|\mathbf{D}_{1:L}'\boldsymbol{\Psi}_{1:L}^{-1}\right\|^l\right]=O(\pi_{h\setminus 0})$.
(iii) For every pair of positive integers $l_1$ and $l_2$, $E\left[|D|^{l_1}\left\|\mathbf{D}_{0:L}'\boldsymbol{\Psi}_{0:L}^{-1}\right\|^{l_2}\right]=O(\pi_{h\setminus 0}h_N^{l_1})$ and $E\left[|D|^{l_1}\left\|\mathbf{D}_{1:L}'\boldsymbol{\Psi}_{1:L}^{-1}\right\|^{l_2}\right]=O(\pi_{h\setminus 0}h_N^{l_1})$. 
(iv) For every random variable $Z$ in Assumption \ref{assn:iid_plus} (ii) and for every positive integer $l$, 
$$
E\left[\|Z\|\|\mathbf{D}_{0:L}'\boldsymbol{\Psi}_{0:L}^{-1}\|^l\right]=O\left(\pi_h\right)
\mbox{ and }
E\left[\|Z\|\|\mathbf{D}_{1:L}'\boldsymbol{\Psi}_{1:L}^{-1}\|^l\right]=O\left(\pi_{h\setminus 0}\right).
$$
\end{lemma}
\begin{proof}
First, we are going to show the statement in (i). 
Under Assumption \ref{assn:GP_around0} (iv), Karamata's theorem \citep[Theorem 2.1 and Exercise 2.5]{resnick2007heavy} implies 
\begin{eqnarray*}
\frac{1}{h_N^l\pi_{h\setminus 0}}
E\left[\boldsymbol{1}\{|D|\leq h_N\}D^l\right]
&=&
\frac{E\left[\boldsymbol{1}\{|D|^{-1}\geq h_N^{-1}\}(|D|^{-1})^{-l}\mid D>0\right]}{(h_N^{-1})^{-l}P(|D|^{-1}\geq h_N^{-1}\mid D>0)}\frac{P(0<D\leq h_N)}{\pi_{h\setminus 0}}
\\&&
+(-1)^l\frac{E\left[\boldsymbol{1}\{|D|^{-1}\geq h_N^{-1}\}(|D|^{-1})^{-l}\mid D<0\right]}{(h_N^{-1})^{-l}P(|D|^{-1}\geq h_N^{-1}\mid D<0)}\frac{P(-h_N\leq D<0)}{\pi_{h\setminus 0}}
\\
&=&
\frac{\alpha_1}{l+\alpha_1}\frac{P(0<D\leq h_N)}{\pi_{h\setminus 0}}+(-1)^l\frac{\alpha_2}{l+\alpha_2}\frac{P(-h_N\leq D<0)}{\pi_{h\setminus 0}}+o(1).
\end{eqnarray*}
Similarly, 
$$
\frac{1}{h_N^l\pi_{h\setminus 0}}
E\left[\boldsymbol{1}\{|D|\leq h_N\}|D|^l\right]
=
\frac{\alpha_1}{l+\alpha_1}\frac{P(0<D\leq h_N)}{\pi_{h\setminus 0}}+\frac{\alpha_2}{l+\alpha_2}\frac{P(-h_N\leq D<0)}{\pi_{h\setminus 0}}+o(1).
$$
Therefore, the statement in (i) hold. 

Second, we are going to show the statement in (ii). 
Note that 
$$
\left\|\mathbf{D}_{0:L}'\boldsymbol{\Psi}_{0:L}^{-1}\right\|^l
=
\boldsymbol{1}\{|D|\leq h_N\}\left\|\left(\begin{array}{cccc}1&\cdots&\frac{D^{L}}{h_N^{L}}\end{array}\right)'\right\|^l
\leq 
(L+1)^{l/2}\boldsymbol{1}\{|D|\leq h_N\}
$$
and 
$$
\left\|\mathbf{D}_{1:L}'\boldsymbol{\Psi}_{1:L}^{-1}\right\|^l
=
\boldsymbol{1}\{0<|D|\leq h_N\}\left\|\left(\begin{array}{cccc}\frac{D}{h_N}&\cdots&\frac{D^{L}}{h_N^{L}}\end{array}\right)'\right\|^l
\leq 
L^{l/2}\boldsymbol{1}\{0<|D|\leq h_N\}.
$$
Therefore, 
$$
E\left[\left\|\mathbf{D}_{0:L}'\boldsymbol{\Psi}_{0:L}^{-1}\right\|^l\right]
\leq 
(L+1)^{l/2}P(|D|\leq h_N)
=
O(\pi_h)
\mbox{ and }
E\left[\left\|\mathbf{D}_{1:L}'\boldsymbol{\Psi}_{1:L}^{-1}\right\|^l\right]
\leq 
L^{l/2}P(0<|D|\leq h_N)
=
O(\pi_{h\setminus 0}).
$$

Third, we are going to show the statement in (iii). 
Note that 
$$
|D|^{l_1}\left\|\mathbf{D}_{0:L}'\boldsymbol{\Psi}_{0:L}^{-1}\right\|^{l_2}
=
|D|^{l_1}\boldsymbol{1}\{0<|D|\leq h_N\}\left\|\left(\begin{array}{cccc}1&\cdots&\frac{D^{L}}{h_N^{L}}\end{array}\right)'\right\|^{l_2}
\leq
(L+1)^{l_2/2}h_N^{l_1}\boldsymbol{1}\{0<|D|\leq h_N\}
$$
and 
$$
|D|^{l_1}\left\|\mathbf{D}_{1:L}'\boldsymbol{\Psi}_{1:L}^{-1}\right\|^{l_2}
=
|D|^{l_1}\boldsymbol{1}\{0<|D|\leq h_N\}\left\|\left(\begin{array}{cccc}\frac{D}{h_N}&\cdots&\frac{D^{L}}{h_N^{L}}\end{array}\right)'\right\|^{l_2}
\leq
L^{l_2/2}h_N^{l_1}\boldsymbol{1}\{0<|D|\leq h_N\}.
$$
Therefore, 
$$
E\left[|D|^{l_1}\left\|\mathbf{D}_{0:L}'\boldsymbol{\Psi}_{0:L}^{-1}\right\|^{l_2}\right]
\leq
(L+1)^{l_2/2}h_N^{l_1}E[\boldsymbol{1}\{0<|D|\leq h_N\}]
=
O(\pi_{h\setminus 0}h_N^{l_1})
$$ 
and 
$$
E\left[|D|^{l_1}\left\|\mathbf{D}_{1:L}'\boldsymbol{\Psi}_{1:L}^{-1}\right\|^{l_2}\right]
\leq
L^{l_2/2}h_N^{l_1}E[\boldsymbol{1}\{0<|D|\leq h_N\}]
=
O(\pi_{h\setminus 0}h_N^{l_1}).
$$ 

Fourth, we are going to show the statement in (iv). 
There is a positive constant $c$ such that $\sup_{u\in[-c,c]}E[\|Z\|\mid D=u]<\infty$. 
By the statement in (ii),  
$$
E\left[\|Z\|\|\mathbf{D}_{0:L}'\boldsymbol{\Psi}_{0:L}^{-1}\|^l\right]
\leq
E\left[\|\mathbf{D}_{0:L}'\boldsymbol{\Psi}_{0:L}^{-1}\|^l\right]\sup_{u\in[-h_N,h_N]}E[\|Z\|\mid D=u]
=
O\left(\pi_h\right),\mbox{ and }
$$
$$
E\left[\|Z\|\|\mathbf{D}_{1:L}'\boldsymbol{\Psi}_{1:L}^{-1}\|^l\right]
\leq
E\left[\|\mathbf{D}_{1:L}'\boldsymbol{\Psi}_{1:L}^{-1}\|^l\right]\sup_{u\in[-h_N,h_N]}E[\|Z\|\mid D=u]
=
O\left(\pi_{h\setminus 0}\right).
$$
\end{proof}

\begin{lemma}\label{lemma:Karamata's2}
(i) For every integer $l$ with $l\geq\max\{\alpha_1,\alpha_2\}$, 
$$
\frac{h_N^l}{\pi_{h\setminus 0}}E\left[\boldsymbol{1}\{h_N<|D|\}D^{-l}\right]=\frac{\alpha_1}{l-\alpha_1}\frac{P(0<D\leq h_N)}{\pi_{h\setminus 0}}+(-1)^l\frac{\alpha_2}{l-\alpha_2}\frac{P(-h_N<D<0)}{\pi_{h\setminus 0}}+o(1).
$$
(ii) For every random variable $Z$ in Assumption \ref{assn:iid_plus} (ii) and for every positive integer $l$, 
$$
E\left[\boldsymbol{1}\{|D|>h_N\}|D|^{-1}\|Z\|\right]=O\left(\frac{\pi_{h\setminus 0}^{1-\epsilon}}{h_N}\right)\mbox{ for every number }\epsilon\in(0,1),
$$ 
$$
E\left[\boldsymbol{1}\{|D|>h_N\}D^{-2}\|Z\|\right]=O\left(\frac{\pi_{h\setminus 0}}{h_N^2}\right),\mbox { and }
$$ 
$$
E\left[\boldsymbol{1}\{|D|>h_N\}D^{-4}\|Z\|\right]=O\left(\frac{\pi_{h\setminus 0}}{h_N^4}\right). 
$$
\end{lemma}
\begin{proof}
Under Assumption \ref{assn:GP_around0} (iv), Karamata's theorem \citep[Theorem 2.1 and Exercise 2.5]{resnick2007heavy} implies that 
\begin{eqnarray*}
\frac{h_N^l}{\pi_{h\setminus 0}}E\left[\boldsymbol{1}\{|D|>h_N\}D^{-l}\right]
&=&
\frac{E\left[\boldsymbol{1}\{|D|^{-1}\leq h_N^{-1}\}(|D|^{-1})^l\mid D>0\right]}{
(h_N^{-1})^lP(|D|^{-1}\geq h_N^{-1}\mid D>0)}\frac{P(0<D\leq h_N)}{\pi_{h\setminus 0}}
\\&&
+(-1)^l\frac{E\left[\boldsymbol{1}\{|D|^{-1}\leq h_N^{-1}\}(|D|^{-1})^l\mid D<0\right]}{
(h_N^{-1})^lP(|D|^{-1}\geq h_N^{-1}\mid D<0)}\frac{P(-h_N<D<0)}{\pi_{h\setminus 0}}
\\
&=&
\frac{\alpha_1}{l-\alpha_1}\frac{P(0<D\leq h_N)}{\pi_{h\setminus 0}}+(-1)^l\frac{\alpha_2}{l-\alpha_2}\frac{P(-h_N<D<0)}{\pi_{h\setminus 0}}+o(1),
\end{eqnarray*}
which proves the statement in (i). 
We are going to show the statement in (ii). 
There is a positive constant $c$ such that $\sup_{u\in[-c,c]}E[\|Z\|\mid D=u]<\infty$. 
By the triangle inequality and H\"older's inequality, 
\begin{eqnarray*}
\left\|E\left[\boldsymbol{1}\{|D|>h_N\}|D|^{-1}\|Z\|\right]\right\|
&\leq&
\left\|E\left[\boldsymbol{1}\{|D|>c\}|D|^{-1}\|Z\|\right]\right\|+\left\|E\left[\boldsymbol{1}\{h_N<|D|\leq c\}|D|^{-1}\|Z\|\right]\right\|
\\
&\leq&
c^{-1}E\left[\|Z\|\right]+E\left[\boldsymbol{1}\{h_N<|D|\}|D|^{-1}\right]\sup_{u\in[-c,c]}E[\|Z\|\mid D=u]\\
\\
&\leq&
c^{-1}E\left[\|Z\|\right]+E\left[\boldsymbol{1}\{h_N<|D|\}|D|^{-1/1-\epsilon}\right]^{1-\epsilon}\sup_{u\in[-c,c]}E[\|Z\|\mid D=u],
\end{eqnarray*}
\begin{eqnarray*}
\left\|E\left[\boldsymbol{1}\{|D|>h_N\}|D|^{-2}\|Z\|\right]\right\|
&\leq&
\left\|E\left[\boldsymbol{1}\{|D|>c\}|D|^{-2}\|Z\|\right]\right\|+\left\|E\left[\boldsymbol{1}\{h_N<|D|\}|D|^{-2}\|Z\|\right]\right\|
\\
&\leq&
c^{-2}E\left[\|Z\|\right]+E\left[\boldsymbol{1}\{h_N<|D|\}|D|^{-2}\right]\sup_{u\in[-c,c]}E[\|Z\|\mid D=u],
\end{eqnarray*}
and 
\begin{eqnarray*}
\left\|E\left[\boldsymbol{1}\{|D|>h_N\}|D|^{-4}\|Z\|\right]\right\|
&\leq&
\left\|E\left[\boldsymbol{1}\{|D|>c\}|D|^{-4}\|Z\|\right]\right\|+\left\|E\left[\boldsymbol{1}\{h_N<|D|\}|D|^{-4}\|Z\|\right]\right\|
\\
&\leq&
c^{-4}E\left[\|Z\|\right]+E\left[\boldsymbol{1}\{h_N<|D|\}|D|^{-4}\right]\sup_{u\in[-c,c]}E[\|Z\|\mid D=u].
\end{eqnarray*}
By the statement in (i), the statement in (ii) holds. 
\end{proof}

\begin{lemma}\label{lemma:Omega_property2}
The minimum eigenvalue of $\boldsymbol{\Pi}$ is bounded away from zero, the maximum eigenvalue of $\boldsymbol{\Pi}$ is bounded, and  $\boldsymbol{\hat\Pi}^{-1}-\boldsymbol{\Pi}^{-1}=O_p\left(\frac{1}{\sqrt{N\pi_{h\setminus 0}}}\right)$.
\end{lemma}
\begin{proof}
First, we are going to show the first two statements.
Let $\mathbf{c}$ be any $L$-dimensional column vector with $\|\mathbf{c}\|=1$. 
Define $\mathbf{\tilde{c}}=\mathrm{diag}(1,-1,1,-1,\ldots)\mathbf{c}$. Note that $\|\mathbf{\tilde{c}}\|=1$. 
Define 
\begin{align*}
\mathbf{C}_1
=
\left(\begin{array}{ccc}
\frac{\alpha_1}{\alpha_1+2}&\cdots&\frac{\alpha_1}{\alpha_1+L+1}\\
\vdots&\ddots&\vdots\\
\frac{\alpha_1}{\alpha_1+L+1}&\cdots&\frac{\alpha_1}{\alpha_1+2L}\\
\end{array}\right)
\qquad\mbox{ and }\qquad
\mathbf{C}_2
=
\left(\begin{array}{ccc}
\frac{\alpha_2}{\alpha_2+2}&\cdots&\frac{\alpha_2}{\alpha_2+L+1}\\
\vdots&\ddots&\vdots\\
\frac{\alpha_2}{\alpha_2+L+1}&\cdots&\frac{\alpha_2}{\alpha_2+2L}\\
\end{array}\right).
\end{align*}
Note that $\mathbf{C}_1$ and $\mathbf{C}_2$ are Cauchy matrices and positive definite.
By Lemma \ref{lemma:Karamata's1},
$$
\mathbf{c}'\boldsymbol{\Pi}\mathbf{c}=\frac{P(0<D\leq h_N)}{\pi_{h\setminus 0}}\mathbf{c}'\mathbf{C}_1\mathbf{c}+\frac{P(-h_N\leq D<0)}{\pi_{h\setminus 0}}\mathbf{\tilde{c}}'\mathbf{C}_2\mathbf{\tilde{c}}+o(1),
$$
which implies that $\mathbf{c}'\boldsymbol{\Pi}\mathbf{c}$ is bounded away from zero. 
Therefore, the minimum eigenvalue of $\boldsymbol{\Pi}$ is bounded away from zero and the maximum eigenvalue of $\boldsymbol{\Pi}$ is bounded as $h_N \rightarrow 0$. 

Second, we are going to show $\boldsymbol{\hat\Pi}-\boldsymbol{\Pi}=O_p\left(\frac{1}{\sqrt{N\pi_{h\setminus 0}}}\right)$. 
Since $\boldsymbol{\hat\Pi}-\boldsymbol{\Pi}=\frac{1}{\pi_{h\setminus 0}}(E_N-E)\left[\boldsymbol{\Psi}_{1:L}^{-1}\mathbf{D}_{1:L}\mathbf{D}_{1:L}'\boldsymbol{\Psi}_{1:L}^{-1}\right]$, 
by Assumption \ref{assn:iid_plus} (i) and Lemma \ref{lemma:Karamata's1}, 
$$
E\left[\left\|\boldsymbol{\hat\Pi}-\boldsymbol{\Pi}\right\|^2\right]^{1/2}=\frac{1}{\sqrt{N\pi_{h\setminus 0}^2}}E\left[\left\|\mathbf{D}_{1:L}'\boldsymbol{\Psi}_{1:L}^{-1}\right\|^4\right]^{1/2}=O\left(\frac{1}{\sqrt{N\pi_{h\setminus 0}}}\right).
$$

Third, we are going to show $\boldsymbol{\hat\Pi}^{-1}-\boldsymbol{\Pi}^{-1}=O_p\left(\frac{1}{\sqrt{N\pi_{h\setminus 0}}}\right)$.
Note that $\boldsymbol{\hat\Pi}^{-1}-\boldsymbol{\Pi}^{-1}=-(\boldsymbol{\Pi}+(\boldsymbol{\hat\Pi}-\boldsymbol{\Pi}))^{-1}(\boldsymbol{\hat\Pi}-\boldsymbol{\Pi})\boldsymbol{\Pi}^{-1}$.
Since the minimum eigenvalue of $\boldsymbol{\Pi}$ is bounded away from zero and $\boldsymbol{\hat\Pi}-\boldsymbol{\Pi}=o_p(1)$, the matrix $(\boldsymbol{\Pi}+(\boldsymbol{\hat\Pi}-\boldsymbol{\Pi}))^{-1}$ is $O_p(1)$.  
Since $\boldsymbol{\hat\Pi}-\boldsymbol{\Pi}=\frac{1}{\pi_{h\setminus 0}}(E_N-E)\left[\boldsymbol{\Psi}_{1:L}^{-1}\mathbf{D}_{1:L}\mathbf{D}_{1:L}'\boldsymbol{\Psi}_{1:L}^{-1}\right]=O_p\left(\frac{1}{\sqrt{N\pi_{h\setminus 0}}}\right)$ by Lemma \ref{lemma:Karamata's1}, we have $\boldsymbol{\hat\Pi}^{-1}-\boldsymbol{\Pi}^{-1}=O_p\left(\frac{1}{\sqrt{N\pi_{h\setminus 0}}}\right)$.
\end{proof}

\begin{lemma}\label{lemma:Omega_property}
$\boldsymbol{\Omega}^{-1}\mathbf{e}_1=O(1)$ and $\left(\boldsymbol{\hat\Omega}^{-1}-\boldsymbol{\Omega}^{-1}\right)\mathbf{e}_1=O_p\left(\sqrt{\frac{\pi_h}{N\pi_{h\setminus 0}^2}}\right)$.
\end{lemma}
\begin{proof}
First, the minimum eigenvalue of $\mathbf{\bar{C}}_N$ is positive and bounded away from zero for sufficiently large $N$, where 
$$
\mathbf{\bar{C}}_N
=
\left[\begin{array}{cc}
1&\frac{1}{\pi_{h\setminus 0}}E[\mathbf{D}_{1:L}'\boldsymbol{\Psi}_{1:L}^{-1}]\\
\frac{1}{\pi_{h\setminus 0}}E[\mathbf{D}_{1:L}'\boldsymbol{\Psi}_{1:L}^{-1}]'&\boldsymbol{\Pi}
\end{array}\right].
$$
Let $\mathbf{c}$ be any $(L+1)$-dimensional column vector with $\|\mathbf{c}\|=1$. 
Define $\mathbf{\tilde{c}}=\mathrm{diag}(1,-1,1,-1,\ldots)\mathbf{c}$. 
Note that $\|\mathbf{\tilde{c}}\|=1$. 
Define 
$$
\mathbf{\tilde{C}}_1=\left(\begin{array}{cccc}
\frac{\alpha_1}{\alpha_1}&\cdots&\frac{\alpha_1}{\alpha_1+L}\\
\vdots&\ddots&\vdots\\
\frac{\alpha_1}{\alpha_1+L}&\cdots&\frac{\alpha_1}{\alpha_1+2L}\\
\end{array}\right)
\mbox{ and }
\mathbf{\tilde{C}}_2=\left(\begin{array}{ccc}
\frac{\alpha_2}{\alpha_2}&\cdots&\frac{\alpha_2}{\alpha_2+L}\\
\vdots&\ddots&\vdots\\
\frac{\alpha_2}{\alpha_2+L}&\cdots&\frac{\alpha_2}{\alpha_2+2L}\\
\end{array}\right).
$$
Note that $\mathbf{\tilde{C}}_1$ and $\mathbf{\tilde{C}}_2$ are Cauchy matrices and positive definite.
By Lemma \ref{lemma:Karamata's1},
$$
\mathbf{c}'\mathbf{\bar{C}}_N\mathbf{c}=\frac{P(0<D\leq h_N)}{\pi_{h\setminus 0}}\mathbf{c}'\mathbf{\tilde{C}}_1\mathbf{c}+\frac{P(-h_N\leq D<0)}{\pi_{h\setminus 0}}\mathbf{\tilde{c}}'\mathbf{\tilde{C}}_2\mathbf{\tilde{c}}+o(1),
$$
which implies that $\mathbf{c}'\mathbf{\bar{C}}_N\mathbf{c}$ is bounded away from zero. 
Therefore, the minimum eigenvalue of $\mathbf{\bar{C}}_N$ is bounded away from zero. 

Second, we are going to show $\boldsymbol{\Omega}^{-1}\mathbf{e}_1=O(1)$.
Since 
$$
\boldsymbol{\Omega}=
\left[\begin{array}{cc}
1&\frac{1}{\pi_h}E[\mathbf{D}_{1:L}'\boldsymbol{\Psi}_{1:L}^{-1}]\\
\frac{1}{\pi_h}E[\mathbf{D}_{1:L}'\boldsymbol{\Psi}_{1:L}^{-1}]'&\frac{\pi_{h\setminus 0}}{\pi_h}\boldsymbol{\Pi}
\end{array}\right],
$$
the inverse matrix formula for block matrices implies 
\begin{align}
\boldsymbol{\Omega}^{-1}\mathbf{e}_1
=&
\left[\begin{array}{cc}
\left(1-\frac{1}{\pi_h\pi_{h\setminus 0}}E[\mathbf{D}_{1:L}'\boldsymbol{\Psi}_{1:L}^{-1}]\boldsymbol{\Pi}^{-1}E[\mathbf{D}_{1:L}'\boldsymbol{\Psi}_{1:L}^{-1}]'\right)^{-1}\\
-\frac{1}{\pi_{h\setminus 0}}\boldsymbol{\Pi}^{-1}E[\mathbf{D}_{1:L}'\boldsymbol{\Psi}_{1:L}^{-1}]'\left(1-\frac{1}{\pi_h\pi_{h\setminus 0}}E[\mathbf{D}_{1:L}'\boldsymbol{\Psi}_{1:L}^{-1}]\boldsymbol{\Pi}^{-1}E[\mathbf{D}_{1:L}'\boldsymbol{\Psi}_{1:L}^{-1}]'\right)^{-1}
\end{array}\right]\nonumber\\
=&
\left[\begin{array}{cc}1\\-\frac{1}{\pi_{h\setminus 0}}\boldsymbol{\Pi}^{-1}E[\mathbf{D}_{1:L}'\boldsymbol{\Psi}_{1:L}^{-1}]\end{array}\right]\left(1-\frac{1}{\pi_h\pi_{h\setminus 0}}E[\mathbf{D}_{1:L}'\boldsymbol{\Psi}_{1:L}^{-1}]\boldsymbol{\Pi}^{-1}E[\mathbf{D}_{1:L}'\boldsymbol{\Psi}_{1:L}^{-1}]'\right)^{-1}
\label{eq:omega_inverse}
\end{align}
as long as $\boldsymbol{\Pi}$ is invertible and $1-\frac{1}{\pi_h\pi_{h\setminus 0}}E[\mathbf{D}_{1:L}'\boldsymbol{\Psi}_{1:L}^{-1}]\boldsymbol{\Pi}^{-1}E[\mathbf{D}_{1:L}'\boldsymbol{\Psi}_{1:L}^{-1}]'\ne 0$. 
It suffices to show that $1-\frac{1}{\pi_h\pi_{h\setminus 0}}E[\mathbf{D}_{1:L}'\boldsymbol{\Psi}_{1:L}^{-1}]\boldsymbol{\Pi}^{-1}E[\mathbf{D}_{1:L}'\boldsymbol{\Psi}_{1:L}^{-1}]'$ is bounded away from zero for sufficiently large $N$. 
By the determinant formula for block matrices, we have
$$
1-\frac{1}{\pi_{h\setminus 0}^2}E[\mathbf{D}_{1:L}'\boldsymbol{\Psi}_{1:L}^{-1}]\boldsymbol{\Pi}^{-1}E[\mathbf{D}_{1:L}'\boldsymbol{\Psi}_{1:L}^{-1}]'=\frac{\mathrm{det}\left(\mathbf{\bar{C}}_N\right)}{\mathrm{det}(\boldsymbol{\Pi})}.
$$ 
Since the minimum eigenvalue of $\mathbf{\bar{C}}_N$ is bounded away from zero, $1-\frac{1}{\pi_{h\setminus 0}^2}E[\mathbf{D}_{1:L}'\boldsymbol{\Psi}_{1:L}^{-1}]\boldsymbol{\Pi}^{-1}E[\mathbf{D}_{1:L}'\boldsymbol{\Psi}_{1:L}^{-1}]'$ is positive and bounded away from zero. Since 
$$
1-\frac{1}{\pi_h\pi_{h\setminus 0}}E[\mathbf{D}_{1:L}'\boldsymbol{\Psi}_{1:L}^{-1}]\boldsymbol{\Pi}^{-1}E[\mathbf{D}_{1:L}'\boldsymbol{\Psi}_{1:L}^{-1}]'\geq 1-\frac{1}{\pi_{h\setminus 0}^2}E[\mathbf{D}_{1:L}'\boldsymbol{\Psi}_{1:L}^{-1}]\boldsymbol{\Pi}^{-1}E[\mathbf{D}_{1:L}'\boldsymbol{\Psi}_{1:L}^{-1}]',
 $$
 it follows that $1-\frac{1}{\pi_h\pi_{h\setminus 0}}E[\mathbf{D}_{1:L}'\boldsymbol{\Psi}_{1:L}^{-1}]\boldsymbol{\Pi}^{-1}E[\mathbf{D}_{1:L}'\boldsymbol{\Psi}_{1:L}^{-1}]'$ is bounded away from zero.

Third, we are going to show that the minimum eigenvalue of $\frac{\pi_h}{\pi_{h\setminus 0}}\boldsymbol{\Omega}$  is bounded away from zero. 
Note that 
$$
\frac{\pi_h}{\pi_{h\setminus 0}}\boldsymbol{\Omega}=\mathbf{\bar{C}}_N+\frac{\pi_h-\pi_{h\setminus 0}}{\pi_{h\setminus 0}}\mathbf{e}_1\mathbf{e}_1'.
$$
The minimum eigenvalue of $\mathbf{\bar{C}}_N$
is bounded away from zero, and $\frac{\pi_h-\pi_{h\setminus 0}}{\pi_{h\setminus 0}}\mathbf{e}_1\mathbf{e}_1'$ is positive semidefinite. 
Therefore, the minimum eigenvalue of $\frac{\pi_h}{\pi_{h\setminus 0}}\boldsymbol{\Omega}$ is bounded away from zero. 

Fourth, we are going to show that $\frac{\pi_h}{\pi_{h\setminus 0}}(\boldsymbol{\hat\Omega}-\boldsymbol{\Omega})=O_p\left(\sqrt{\frac{\pi_h}{N\pi_{h\setminus 0}^2}}\right)
$. Since $\frac{\pi_h}{\pi_{h\setminus 0}}(\boldsymbol{\hat\Omega}-\boldsymbol{\Omega})=\frac{1}{\pi_{h\setminus 0}}(E_N-E)\left[\boldsymbol{\Psi}_{0:L}^{-1}\mathbf{D}_{0:L}\mathbf{D}_{0:L}'\boldsymbol{\Psi}_{0:L}^{-1}\right]$,
by Assumption \ref{assn:iid_plus} (i) and Lemma \ref{lemma:Karamata's1}, 
$$
E\left[\left\|\frac{\pi_h}{\pi_{h\setminus 0}}(\boldsymbol{\hat\Omega}-\boldsymbol{\Omega})\right\|^2\right]^{1/2}=\frac{1}{\sqrt{N\pi_{h\setminus 0}^2}}E\left[\left\|\mathbf{D}_{0:L}'\boldsymbol{\Psi}_{0:L}^{-1}\right\|^4\right]^{1/2}=O\left(\sqrt{\frac{\pi_h}{N\pi_{h\setminus 0}^2}}\right).
$$

Fifth, we are going to show $\left(\boldsymbol{\hat\Omega}^{-1}-\boldsymbol{\Omega}^{-1}\right)\mathbf{e}_1=O_p\left(\sqrt{\frac{\pi_h}{N\pi_{h\setminus 0}^2}}\right)$.
Note that 
$$
\left(\boldsymbol{\hat\Omega}^{-1}-\boldsymbol{\Omega}^{-1}\right)\mathbf{e}_1=-\left(\frac{\pi_h}{\pi_{h\setminus 0}}\boldsymbol{\Omega}+\frac{\pi_h}{\pi_{h\setminus 0}}(\boldsymbol{\hat\Omega}-\boldsymbol{\Omega})\right)^{-1}\frac{\pi_h}{\pi_{h\setminus 0}}\left(\boldsymbol{\hat\Omega}-\boldsymbol{\Omega}\right)\boldsymbol{\Omega}^{-1}\mathbf{e}_1.
$$
Since the minimum eigenvalue of $\frac{\pi_h}{\pi_{h\setminus 0}}\boldsymbol{\Omega}$ is bounded away from zero and $\frac{\pi_h}{\pi_{h\setminus 0}}(\boldsymbol{\hat\Omega}-\boldsymbol{\Omega})=o_p(1)$, the matrix $\left(\frac{\pi_h}{\pi_{h\setminus 0}}\boldsymbol{\Omega}+\frac{\pi_h}{\pi_{h\setminus 0}}(\boldsymbol{\hat\Omega}-\boldsymbol{\Omega})\right)^{-1}$ is $O_p(1)$.  
Since $\frac{\pi_h}{\pi_{h\setminus 0}}\left(\boldsymbol{\hat\Omega}-\boldsymbol{\Omega}\right)=O_p\left(\sqrt{\frac{\pi_h}{N\pi_{h\setminus 0}^2}}\right)$ and $\boldsymbol{\Omega}^{-1}\mathbf{e}_1=O(1)$, we have $\left(\boldsymbol{\hat\Omega}^{-1}-\boldsymbol{\Omega}^{-1}\right)\mathbf{e}_1=O_p\left(\sqrt{\frac{\pi_h}{N\pi_{h\setminus 0}^2}}\right)$.
\end{proof}

\begin{lemma}\label{lemma:M_inve}
$\mathbf{V}^{-1}=O(1)$ and $\mathbf{\hat{V}}^{-1}-\mathbf{V}^{-1}=O_p\left(\sqrt{\frac{\pi_h}{N\pi_{h\setminus 0}^2}}\right)$. 
\end{lemma}
\begin{proof}
First, we are going to show $\mathbf{V}^{-1}=O(1)$. It suffices to investigate the minimum eigenvalue of $\mathbf{V}$. 
Let $\mathbf{c}$ be any $(T-1)p$-dimensional column vector with $\|c\|=1$. 
Define $m_1(u)=E\left[\|\mathbf{X}^\ast\mathbf{W}\mathbf{c}\|^2\mid D=u\right]$. 
Assumption \ref{assn:more_primitiv_higher} (ii) implies that there is a constant $C$ such that $\left|m_1(u)-m_1(0)\right|\leq C|u|$ for every $u\in[-h_N,h_N]$ when $h_N$ is sufficiently small. 
Since 
$\mathbf{c}'\mathbf{V}\mathbf{c}-m_1(0)=\mathbf{e}_1'\boldsymbol{\Omega}^{-1}\frac{1}{\pi_h}E\left[\boldsymbol{\Psi}_{0:L}^{-1}\mathbf{D}_{0:L}\left(m_1(D)-m_1(0)\right)\right]$, we have 
$$
\left|\mathbf{c}'\mathbf{V}\mathbf{c}-m_1(0)\right|\leq \|\boldsymbol{\Omega}^{-1}\mathbf{e}_1\|\frac{1}{\pi_h}E\left[|D|\|\mathbf{D}_{0:L}'\boldsymbol{\Psi}_{0:L}^{-1}\|\right]=o(1)
$$ by Lemmas \ref{lemma:Karamata's1} and \ref{lemma:Omega_property}.
Since $m_1(0)>0$ by Assumption \ref{assn:iid_plus} (iii), it follows that $\mathbf{c}'\mathbf{V}\mathbf{c}$ is bounded away from zero. 

Next, we are going to show $\mathbf{\hat{V}}^{-1}-\mathbf{V}^{-1}=O_p\left(\sqrt{\frac{\pi_h}{N\pi_{h\setminus 0}^2}}\right)$. 
Since $\mathbf{V}^{-1}=O(1)$, it suffices to show $\mathbf{\hat{V}}-\mathbf{V}=O_p\left(\sqrt{\frac{\pi_h}{N\pi_{h\setminus 0}^2}}\right)$.
Let $\mathbf{c}$ be any $(T-1)p$-dimensional column vector with $\|c\|=1$. 
Then 
$$
(\mathbf{\hat{V}}-\mathbf{V})\mathbf{c}=\frac{1}{\pi_h}E\left[(\mathbf{X}^\ast\mathbf{W})'\mathbf{X}^\ast\mathbf{W}\mathbf{c}\mathbf{D}_{0:L}'\boldsymbol{\Psi}_{0:L}^{-1}\right]\left(\boldsymbol{\hat\Omega}^{-1}-\boldsymbol{\Omega}^{-1}\right)\mathbf{e}_1+\frac{1}{\pi_h}(E_N-E)\left[(\mathbf{X}^\ast\mathbf{W})'\mathbf{X}^\ast\mathbf{W}\mathbf{c}\mathbf{D}_{0:L}'\boldsymbol{\Psi}_{0:L}^{-1}\right]\boldsymbol{\hat\Omega}^{-1}\mathbf{e}_1.
$$
Since Assumption \ref{assn:iid_plus} and Lemma \ref{lemma:Karamata's1} imply 
$$
\|E\left[(\mathbf{X}^\ast\mathbf{W})'\mathbf{X}^\ast\mathbf{W}\mathbf{c}\mathbf{D}_{0:L}'\boldsymbol{\Psi}_{0:L}^{-1}\right]\|
\leq
E\left[\|\mathbf{X}^\ast\mathbf{W}\|^2\|\mathbf{D}_{0:L}'\boldsymbol{\Psi}_{0:L}^{-1}\|\right]
=
O(\pi_h)
$$
and 
$$
E\left[\left\|(E_N-E)\left[(\mathbf{X}^\ast\mathbf{W})'\mathbf{X}^\ast\mathbf{W}\mathbf{c}\mathbf{D}_{0:L}'\boldsymbol{\Psi}_{0:L}^{-1}\right]\right\|^2\right]^{1/2}
\leq
\frac{1}{\sqrt{N}}E\left[\left\|\mathbf{X}^\ast\mathbf{W}\right\|^4\left\|\mathbf{D}_{0:L}'\boldsymbol{\Psi}_{0:L}^{-1}\right\|^2\right]^{1/2}
=
O\left(\sqrt{\frac{\pi_h}{N}}\right),
$$
together with Lemma \ref{lemma:Omega_property}, we have $(\mathbf{\hat{V}}-\mathbf{V})\mathbf{c}=O_p\left(\sqrt{\frac{\pi_h}{N\pi_{h\setminus 0}^2}}\right)$.
\end{proof}

\begin{lemma}\label{lemma:Sigma_rate}
$\mathbf{Q}=O(\frac{\pi_h^{1-\epsilon}}{h_N})$ for every positive number $\epsilon\in(0,1)$, and $\mathbf{\hat{Q}}-\mathbf{Q}=O_p\left(\sqrt{\frac{\pi_h^2}{Nh_N^2\pi_{h\setminus 0}}}\right)$.
\end{lemma}
\begin{proof}
Let $\mathbf{c}$ be any $(T-1)p$-dimensional column vector with $\|\mathbf{c}\|=1$. 
The first statement follows from Lemmas \ref{lemma:influe_eror4} and \ref{lemma:Karamata's2} and 
$$
\mathbf{Q}\mathbf{c}=\mathbf{R}\mathbf{c}-E\left[\boldsymbol{1}\{|D|>h_N\}D^{-1}\mathbf{X}^\ast\mathbf{W}\right]\mathbf{c}-\frac{1}{\pi_{h\setminus 0}}E\left[\mathbf{X}^\ast\mathbf{W}\mathbf{c}\mathbf{D}_{1:L}'\boldsymbol{\Psi}_{1:L}^{-1}\right]\boldsymbol{\Pi}^{-1}\boldsymbol{\Psi}_{1:L}^{-1}\mathbf{h}.
$$
We are going to show the second statement in the rest of the proof. 
Assumption \ref{assn:iid_plus} and Lemma \ref{lemma:Karamata's1} imply 
$$
\|E\left[\mathbf{X}^\ast\mathbf{W}\mathbf{c}\mathbf{D}_{1:L}'\boldsymbol{\Psi}_{1:L}^{-1}\right]\|
\leq
E\left[\|\mathbf{X}^\ast\mathbf{W}\mathbf{c}\|\|\mathbf{D}_{1:L}'\boldsymbol{\Psi}_{1:L}^{-1}\|\right]
=
O\left(\pi_{h\setminus 0}\right),
$$
$$
E\left[\left\|(E_N-E)\left[\mathbf{X}^\ast\mathbf{W}\mathbf{c}\mathbf{D}_{1:L}'\boldsymbol{\Psi}_{1:L}^{-1}\right]\right\|^2\right]^{1/2}
\leq
\frac{1}{\sqrt{N}}E\left[\left\|\mathbf{X}^\ast\mathbf{W}\mathbf{c}\right\|^2\left\|\mathbf{D}_{1:L}'\boldsymbol{\Psi}_{1:L}^{-1}\right\|^2\right]^{1/2}
=
O\left(\sqrt{\frac{\pi_{h\setminus 0}}{N}}\right),
$$
and 
$$
E\left[\left\|(E_N-E)\left[\mathbf{X}^\ast\mathbf{W}\boldsymbol{1}\{|D|>h_N\}D^{-1}\right]\right\|^2\right]^{1/2}
\leq
\frac{1}{\sqrt{N}}E\left[\left\|\mathbf{X}^\ast\mathbf{W}\right\|^2\boldsymbol{1}\{|D|>h_N\}D^{-2}\right]^{1/2}
=
O\left(\sqrt{\frac{\pi_{h\setminus 0}}{Nh_N^2}}\right).
$$
Since 
\begin{eqnarray*}
(\mathbf{\hat{Q}}-\mathbf{Q})\mathbf{c}
&=&
\frac{1}{\pi_{h\setminus 0}}E\left[\mathbf{X}^\ast\mathbf{W}\mathbf{c}\mathbf{D}_{1:L}'\boldsymbol{\Psi}_{1:L}^{-1}\right](\boldsymbol{\hat\Pi}^{-1}-\boldsymbol{\Pi}^{-1})\boldsymbol{\Psi}_{1:L}^{-1}\mathbf{h}
\\&&
+\frac{1}{\pi_{h\setminus 0}}(E_N-E)\left[\mathbf{X}^\ast\mathbf{W}\mathbf{c}\mathbf{D}_{1:L}'\boldsymbol{\Psi}_{1:L}^{-1}\right]\boldsymbol{\hat\Pi}^{-1}\boldsymbol{\Psi}_{1:L}^{-1}\mathbf{h}
\\&&
+(E_N-E)\left[\mathbf{X}^\ast\mathbf{W}\boldsymbol{1}\{|D|>h_N\}D^{-1}\right]\mathbf{c},
\end{eqnarray*}
together with Lemmas \ref{lemma:influe_eror4} and \ref{lemma:Omega_property2}, we have $(\mathbf{\hat{Q}}-\mathbf{Q})\mathbf{c}=O_p\left(\sqrt{\frac{\pi_h^2}{Nh_N^2\pi_{h\setminus 0}}}\right)$.
\end{proof}

\begin{lemma}\label{lemma:fourth1}
(i) $E[\boldsymbol{\xi}_1]=0$.
(ii) $E[\|\boldsymbol{\xi}_1\|^2]=O\left(\frac{\pi_{h\setminus 0}}{h_N^2}\right)$ and $E[\|\boldsymbol{\xi}_1\|^4]=O\left(\frac{\pi_{h\setminus 0}}{h_N^4}\right)$.
(iii) $(E_N-E)[\boldsymbol{\xi}_1]=O_p\left(\sqrt{\frac{\pi_{h\setminus 0}}{Nh_N^2}}\right)$.
\end{lemma}
\begin{proof}
The statement in (i) follows from the definition of $\boldsymbol{\xi}_1$.
The statement in (ii) follows from Assumption \ref{assn:iid_plus} (ii) and Lemma \ref{lemma:Karamata's2}. 
The statement in (iii) follows from Assumption \ref{assn:iid_plus} (i) and the statement in (ii).
\end{proof}

\begin{lemma}\label{lemma:bias_bound1}
(i) $E[\boldsymbol{\xi}_2]=O(h_N^{L+1}\pi_{h\setminus 0})$.
(ii) $E[\|\boldsymbol{\xi}_2\|^2]=O\left(\pi_{h\setminus 0}\right)$ and $E[\|\boldsymbol{\xi}_2\|^4]=O\left(\pi_{h\setminus 0}\right)$.
(iii) $(E_N-E)[\boldsymbol{\xi}_2]=O\left(\sqrt{\frac{\pi_{h\setminus 0}}{N}}\right)$.
\end{lemma}
\begin{proof}
By Assumption \ref{assn:more_primitiv_higher} (i) and $\mathbf{m}(0)=0$, there is a constant ${C}$ such that $\left\|\mathbf{m}(u)-\sum_{l=1}^{L}\mathbf{m}^{(l)}(0)\frac{u^{l}}{l!}\right\|\leq {C}|u|^{L+1}$ for every $u\in[-h_N,h_N]$ when $h_N$ is sufficiently small. 
Then 
\begin{eqnarray*}
\left\|E\left[\boldsymbol{\xi}_2\right]\right\|
&=&
\left\|E\left[\left(\mathbf{m}(D)-\sum_{l=1}^{L}\mathbf{m}^{(l)}(0)\frac{D^{l}}{l!}\right)\mathbf{D}_{1:L}'\boldsymbol{\Psi}_{1:L}^{-1}\right]\right\|
\\
&\leq&
E\left[\left\|\mathbf{m}(D)-\sum_{l=1}^{L}\mathbf{m}^{(l)}(0)\frac{D^{l}}{l!}\right\|\left\|\mathbf{D}_{1:L}'\boldsymbol{\Psi}_{1:L}^{-1}\right\|\right]
\\
&\leq&
{C}E\left[D^{L+1}\left\|\mathbf{D}_{1:L}'\boldsymbol{\Psi}_{1:L}^{-1}\right\|\right]
\\
&\leq&
O(\pi_{h\setminus 0}h_N^{L+1}).
\end{eqnarray*}
The statement in (ii) follows from Assumption \ref{assn:iid_plus} (ii) and Lemma \ref{lemma:Karamata's1}. 
The statement in (iii) follows from Assumption \ref{assn:iid_plus} (i) and the statement in (ii).
\end{proof}

\begin{lemma}\label{lemma:bias_bound2}
(i) $E\left[\boldsymbol{\xi}_3\right]=O\left(h_N\pi_{h\setminus 0}\right)$ and $E\left[\boldsymbol{\xi}_3\boldsymbol{\Omega}^{-1}\mathbf{e}_1\right]=O\left(h_N^{L+1}\pi_{h\setminus 0}\right)$.
(ii) $E[\|\boldsymbol{\xi}_3\|^2]=O\left(\pi_h\right)$ and $E[\|\boldsymbol{\xi}_3\|^4]=O\left(\pi_h\right)$.
(iii) $(E_N-E)\left[\boldsymbol{\xi}_3\right]=O_p\left(\sqrt{\frac{\pi_h}{N}}\right)$.
\end{lemma}
\begin{proof}
Define $m_2(u)=E\left[(\mathbf{X}^\ast\mathbf{W})'\mathbf{X}^\ast(\mathbf{Y}-\mathbf{W}\boldsymbol{\delta})\mid D=u\right]$. 
Note that $m_2(0)=E\left[(\mathbf{X}^\ast\mathbf{W})'D\boldsymbol{\beta}(\mathbf{X})\mid D=0\right]=0$ by \eqref{eq:base_ID}, and Assumption \ref{assn:more_primitiv_higher}(i) implies that there is a constant $C$ such that $\left\|m_2(u)-\sum_{l=0}^{L}m_2^{(l)}(0)\frac{u^{l}}{l!}\right\|\leq C|u|^{L+1}$ for every $u\in[-h_N,h_N]$ when $h_N$ is sufficiently small. 
Then,
\begin{align*}
\left\|E\left[\boldsymbol{\xi}_3\right]\right\|
=&
\left\|E\left[(m_2(D)-m_2(0))\mathbf{D}_{0:L}'\boldsymbol{\Psi}_{0:L}^{-1}\right]\right\|\\
\leq&
\left\|E\left[\left(\sum_{l=1}^{L}m_2^{(l)}(0)\frac{D^{l}}{l!}\right)\mathbf{D}_{0:L}'\boldsymbol{\Psi}_{0:L}^{-1}\right]\right\|+C\left\|E\left[|D|^{L+1}\mathbf{D}_{0:L}'\boldsymbol{\Psi}_{0:L}^{-1}\right]\right\|
\end{align*}
and 
$$
\left\|E\left[\boldsymbol{\xi}_3\right]\boldsymbol{\Omega}^{-1}\mathbf{e}_1\right\|
\le
\left\|E\left[\left(m_2(D)-\sum_{l=0}^{L}m_2^{(l)}(0)\frac{D^{l}}{l!}\right)\mathbf{D}_{0:L}'\boldsymbol{\Psi}_{0:L}^{-1}\right]\boldsymbol{\Omega}^{-1}\mathbf{e}_1\right\|
+
C\left\|E\left[|D|^{L+1}\mathbf{D}_{0:L}'\boldsymbol{\Psi}_{0:L}^{-1}\right]\boldsymbol{\Omega}^{-1}\mathbf{e}_1\right\|.
$$
By Lemmas \ref{lemma:Karamata's1} and \ref{lemma:Omega_property}, the statement in (i) holds.
The statement in (ii) follows from Assumption \ref{assn:iid_plus} (ii) and Lemma \ref{lemma:Karamata's1}. 
The statement in (iii) follows from Assumption \ref{assn:iid_plus} (i) and the statement in (ii).
\end{proof}

\begin{lemma}\label{lemma:zeta_moments}
The following statements hold for every positive number $\epsilon\in(0,1)$.
(i) $E[\boldsymbol{\zeta}]=O\left(h_N^L\right)$.
(ii) $E[\|\boldsymbol{\zeta}\|^2]=O\left(\frac{\pi_h^{1-\epsilon}}{h_N^2}\right)$ and $E[\|\boldsymbol{\zeta}\|^4]=O\left(\frac{\pi_h^{1-\epsilon}}{h_N^4}\right)$.
(iii) $(E_N-E)\left[\boldsymbol{\zeta}\right]=O_p\left(\sqrt{\frac{\pi_h^{1-\epsilon}}{Nh_N^2}}\right)$ and $(E_N-E)\left[\boldsymbol{\zeta}\boldsymbol{\zeta}'\right]=O_p\left(\sqrt{\frac{\pi_h^{1-\epsilon}}{Nh_N^4}}\right)$.
\end{lemma}
\begin{proof}
Note that 
$$
E[\boldsymbol{\zeta}]
=
\frac{1}{\pi_{h\setminus 0}}E[\boldsymbol{\xi}_2]\boldsymbol{\Pi}^{-1}\boldsymbol{\Psi}_{1:L}^{-1}\mathbf{h}+\frac{1}{\pi_h}\mathbf{Q}\mathbf{V}^{-1}E[\boldsymbol{\xi}_3]\boldsymbol{\Omega}^{-1}\mathbf{e}_1.
$$
By Lemmas \ref{lemma:influe_eror4}, \ref{lemma:Omega_property2}, \ref{lemma:M_inve}, and \ref{lemma:Sigma_rate}, $E[\boldsymbol{\zeta}]=O\left(h_N^L\right)$, which shows the statement in (i).
Since 
$$
E[\|\boldsymbol{\zeta}\|^l]^{1/l}\leq E[\|\boldsymbol{\xi}_1\|^l]^{1/l}+\frac{1}{\pi_{h\setminus 0}}E[\|\boldsymbol{\xi}_2\|^l]^{1/l}\|\boldsymbol{\Pi}^{-1}\|\|\boldsymbol{\Psi}_{1:L}^{-1}\mathbf{h}\|+\frac{1}{\pi_h}\|\mathbf{Q}\|\|\mathbf{V}^{-1}\|E[\|\boldsymbol{\xi}_3\|^l]^{1/l}\|\boldsymbol{\Omega}^{-1}\mathbf{e}_1\|
$$ 
for $l=2,4$,
by Lemmas \ref{lemma:Sigma_rate}, \ref{lemma:fourth1}, \ref{lemma:bias_bound2}, and \ref{lemma:bias_bound1}, we have 
$$
E[\|\boldsymbol{\zeta}\|^l]=O\left(\frac{\pi_h^{1-\epsilon}}{h_N^{l}}\right)
$$
for $l=2,4$, which shows the statement in (ii).
The statement in (iii) follows from Assumption \ref{assn:iid_plus} (i) and the statement in (ii).
\end{proof}

\begin{lemma}\label{lemma:lowerconrate}
The minimum eigenvalue of $\frac{h_N^2\pi_{h\setminus 0}}{\pi_h^2}Var(\boldsymbol{\zeta})$ is bounded away from zero. 
\end{lemma}
\begin{proof}
Let $\mathbf{c}$ be any $p$-dimensional column vector with $\|\mathbf{c}\|=1$. 
By Assumptions \ref{assn:more_primitiv_higher} (ii) and \ref{assn:iid_plus} (iii), there are positive constants $c_1$ and $C$ such that $\mathbf{c}'Var\left(\mathbf{X}^\ast(\mathbf{Y}-\mathbf{W}\boldsymbol{\delta})\mid D=u\right)\mathbf{c}\geq C$ for every $u\in[-c_1,c_1]$.
Without loss of generality, assume $h_N\leq c_1$ for the rest of this proof.
Note that $\boldsymbol{\zeta}=\boldsymbol{\xi}_1+(\frac{1}{\pi_{h\setminus 0}}\boldsymbol{\xi}_2\boldsymbol{\Pi}^{-1}\boldsymbol{\Psi}_{1:L}^{-1}\mathbf{h}+\frac{1}{\pi_h}\mathbf{Q}\mathbf{V}^{-1}\boldsymbol{\xi}_3\boldsymbol{\Omega}^{-1}\mathbf{e}_1)$ and that $\boldsymbol{\xi}_1$ and $(\frac{1}{\pi_{h\setminus 0}}\boldsymbol{\xi}_2\boldsymbol{\Pi}^{-1}\boldsymbol{\Psi}_{1:L}^{-1}\mathbf{h}+\frac{1}{\pi_h}\mathbf{Q}\mathbf{V}^{-1}\boldsymbol{\xi}_3\boldsymbol{\Omega}^{-1}\mathbf{e}_1)$ are independent because they use disjoint sets of observations. 
Therefore, 
$$
Var(\mathbf{c}'\boldsymbol{\zeta})\geq \max\left\{Var(\mathbf{c}'\boldsymbol{\xi}_1),Var\left(\mathbf{c}'\left(\frac{1}{\pi_{h\setminus 0}}\boldsymbol{\xi}_2\boldsymbol{\Pi}^{-1}\boldsymbol{\Psi}_{1:L}^{-1}\mathbf{h}+\frac{1}{\pi_h}\mathbf{Q}\mathbf{V}^{-1}\boldsymbol{\xi}_3\boldsymbol{\Omega}^{-1}\mathbf{e}_1\right)\right)\right\}.
$$

We want to show that $\frac{h_N^2\pi_{h\setminus 0}}{\pi_h^2}Var(\mathbf{c}'\boldsymbol{\xi}_1)$ is bounded away from zero. 
We divide the proof into two cases: $P(D=0)=0$ and $P(D=0)>0$.
Consider the case when $P(D=0)=0$. 
We are going to show that $\frac{h_N^2}{\pi_{h\setminus 0}}Var(\mathbf{c}'\boldsymbol{\xi}_1)$ is bounded away from zero.
By the conditional variance identity, 
$$
Var(\mathbf{c}'\boldsymbol{\xi}_1)\geq E\left[Var\left(\mathbf{c}'\boldsymbol{\xi}_1\mid D\right)\right].
$$
Note that 
\begin{eqnarray*}
E\left[Var\left(\mathbf{c}'\boldsymbol{\xi}_1\mid D\right)\right]
&=&
E\left[\boldsymbol{1}\{|D|>h_N\}D^{-2}\mathbf{c}'Var\left(\mathbf{X}^\ast(\mathbf{Y}-\mathbf{W}\boldsymbol{\delta})\mid D\right)\mathbf{c}\right]
\\
&\geq&
E\left[\boldsymbol{1}\{h_N<|D|\leq c_1\}D^{-2}\mathbf{c}'Var\left(\mathbf{X}^\ast(\mathbf{Y}-\mathbf{W}\boldsymbol{\delta})\mid D\right)\mathbf{c}\right]
\\
&\geq&
CE\left[\boldsymbol{1}\{h_N<|D|\leq c_1\}D^{-2}\right].
\end{eqnarray*}
Lemma \ref{lemma:Karamata's2} provides the limit of $\frac{h_N^2}{\pi_{h\setminus 0}}E\left[\boldsymbol{1}\{h_N<|D|\leq c_1\}D^{-2}\right]$, so 
$$
\frac{h_N^2}{\pi_{h\setminus 0}}Var(\mathbf{c}'\boldsymbol{\xi}_1)\geq C\left(\frac{\alpha_1}{2-\alpha_1}\frac{P(0<D\leq h_N)}{\pi_{h\setminus 0}}+\frac{\alpha_2}{2-\alpha_2}\frac{P(-h_N<D<0)}{\pi_{h\setminus 0}}\right)+o(1),
$$
which implies that $\frac{h_N^2}{\pi_{h\setminus 0}}Var(\mathbf{c}'\boldsymbol{\xi}_1)$ is bounded away from zero. 

Now consider the case when $P(D=0)>0$.
Since  $\boldsymbol{\xi}_1$ and $(\frac{1}{\pi_{h\setminus 0}}\boldsymbol{\xi}_2\boldsymbol{\Pi}^{-1}\boldsymbol{\Psi}_{1:L}^{-1}\mathbf{h}+\frac{1}{\pi_h}\mathbf{Q}\mathbf{V}^{-1}\boldsymbol{\xi}_3\boldsymbol{\Omega}^{-1}\mathbf{e}_1)$ are independent, we need to show that ${h_N^2\pi_{h\setminus 0}}Var\left(\mathbf{c}'\left(\frac{1}{\pi_{h\setminus 0}}\boldsymbol{\xi}_2\boldsymbol{\Pi}^{-1}\boldsymbol{\Psi}_{1:L}^{-1}\mathbf{h}+\frac{1}{\pi_h}\mathbf{Q}\mathbf{V}^{-1}\boldsymbol{\xi}_3\boldsymbol{\Omega}^{-1}\mathbf{e}_1\right)\right)$ is bounded away from zero.
Note that, by Lemmas \ref{lemma:Sigma_rate}, \ref{lemma:bias_bound1}, and \ref{lemma:bias_bound2},
\begin{eqnarray*}
&&
\left\|{h_N^2\pi_{h\setminus 0}}Var\left(\mathbf{c}'\left(\frac{1}{\pi_{h\setminus 0}}\boldsymbol{\xi}_2\boldsymbol{\Pi}^{-1}\boldsymbol{\Psi}_{1:L}^{-1}\mathbf{h}+\frac{1}{\pi_h}\mathbf{Q}\mathbf{V}^{-1}\boldsymbol{\xi}_3\boldsymbol{\Omega}^{-1}\mathbf{e}_1\right)\right)
-\frac{h_N^2}{\pi_{h\setminus 0}}Var\left(\mathbf{c}'\mathbf{h}'\boldsymbol{\Psi}_{1:L}^{-1}\boldsymbol{\Pi}^{-1}\boldsymbol{\xi}_2\right)\right\|
\\
&\leq&
2\frac{h_N^2}{\pi_h}\left\|\mathbf{Q}\right\|E\left[\|\boldsymbol{\xi}_2\|^2\right]^{1/2}E\left[\|\boldsymbol{\xi}_3\|^2\right]^{1/2}\|\boldsymbol{\Pi}^{-1}\|\|\boldsymbol{\Psi}_{1:L}^{-1}\mathbf{h}\|\|\mathbf{V}^{-1}\|\|\boldsymbol{\Omega}^{-1}\mathbf{e}_1\|
\\&&
+\frac{h_N^2\pi_{h\setminus 0}}{\pi_h^2}\left\|\mathbf{Q}\right\|^2E\left[\|\boldsymbol{\xi}_3\|^2\right]\|\mathbf{V}^{-1}\|^2\|\boldsymbol{\Omega}^{-1}\mathbf{e}_1\|^2
\\
&=&
O(\pi_{h\setminus 0}^{1/2}),
\end{eqnarray*}
where we use the fact that $1/\pi_h=O(1)$ when $P(D=0)>0$.
Therefore, it suffices to show that $\frac{h_N^2}{\pi_{h\setminus 0}}Var\left(\mathbf{c}'\mathbf{h}'\boldsymbol{\Psi}_{1:L}^{-1}\boldsymbol{\Pi}^{-1}\boldsymbol{\xi}_2\right)$ is bounded away from zero. 
By the conditional variance identity, 
$$
Var\left(\mathbf{c}'\mathbf{h}'\boldsymbol{\Psi}_{1:L}^{-1}\boldsymbol{\Pi}^{-1}\boldsymbol{\xi}_2\right)
\geq
E[Var((\boldsymbol{\xi}_2\boldsymbol{\Pi}^{-1}\boldsymbol{\Psi}_{1:L}^{-1}\mathbf{h})'\mathbf{c}\mid D)].
$$
Since 
\begin{eqnarray*}
Var((\boldsymbol{\xi}_2\boldsymbol{\Pi}^{-1}\boldsymbol{\Psi}_{1:L}^{-1}\mathbf{h})'\mathbf{c}\mid D)
&=&
\left(\boldsymbol{\Psi}_{1:L}^{-1}\mathbf{h}\right)'\boldsymbol{\Pi}^{-1}\boldsymbol{\Psi}_{1:L}^{-1}\mathbf{D}_{1:L}\mathbf{c}'Var(\mathbf{X}^\ast(\mathbf{Y}-\mathbf{W}\boldsymbol{\delta})\mid D)\mathbf{c}\mathbf{D}_{1:L}'\boldsymbol{\Psi}_{1:L}^{-1}\boldsymbol{\Pi}^{-1}\boldsymbol{\Psi}_{1:L}^{-1}\mathbf{h}
\\
&=&
(\mathbf{c}'Var(\mathbf{X}^\ast(\mathbf{Y}-\mathbf{W}\boldsymbol{\delta})\mid D)\mathbf{c})\left(\boldsymbol{\Psi}_{1:L}^{-1}\mathbf{h}\right)'\boldsymbol{\Pi}^{-1}\boldsymbol{\Psi}_{1:L}^{-1}\mathbf{D}_{1:L}\mathbf{D}_{1:L}'\boldsymbol{\Psi}_{1:L}^{-1}\boldsymbol{\Pi}^{-1}\boldsymbol{\Psi}_{1:L}^{-1}\mathbf{h}
\\
&\geq&
C\left(\boldsymbol{\Psi}_{1:L}^{-1}\mathbf{h}\right)'\boldsymbol{\Pi}^{-1}\boldsymbol{\Psi}_{1:L}^{-1}\mathbf{D}_{1:L}\mathbf{D}_{1:L}'\boldsymbol{\Psi}_{1:L}^{-1}\boldsymbol{\Pi}^{-1}\boldsymbol{\Psi}_{1:L}^{-1}\mathbf{h},
\end{eqnarray*}
we have
\begin{eqnarray*}
Var\left(\mathbf{c}'\mathbf{h}'\boldsymbol{\Psi}_{1:L}^{-1}\boldsymbol{\Pi}^{-1}\boldsymbol{\xi}_2\right)
&\geq&
E[Var((\boldsymbol{\xi}_2\boldsymbol{\Pi}^{-1}\boldsymbol{\Psi}_{1:L}^{-1}\mathbf{h})'\mathbf{c}\mid D)]
\\
&\geq&
C\left(\boldsymbol{\Psi}_{1:L}^{-1}\mathbf{h}\right)'\boldsymbol{\Pi}^{-1}E\left[\boldsymbol{\Psi}_{1:L}^{-1}\mathbf{D}_{1:L}\mathbf{D}_{1:L}'\boldsymbol{\Psi}_{1:L}^{-1}\right]\boldsymbol{\Pi}^{-1}\boldsymbol{\Psi}_{1:L}^{-1}\mathbf{h}
\\
&=&
\pi_{h\setminus 0}C\left(\boldsymbol{\Psi}_{1:L}^{-1}\mathbf{h}\right)'\boldsymbol{\Pi}^{-1}\boldsymbol{\Psi}_{1:L}^{-1}\mathbf{h}.
\end{eqnarray*}
By Lemmas \ref{lemma:influe_eror4} and \ref{lemma:Omega_property2}, $\frac{h_N^2}{\pi_{h\setminus 0}}Var\left(\mathbf{c}'\mathbf{h}'\boldsymbol{\Psi}_{1:L}^{-1}\boldsymbol{\Pi}^{-1}\boldsymbol{\xi}_2\right)$ is bounded away from zero. 
\end{proof}

\begin{lemma}\label{lemma:xi_estimates}
$E_N[\|\boldsymbol{\hat\zeta}-\boldsymbol{\zeta}\|^2]^{1/2}=o\left(1\right)+O_p\left(\sqrt{\frac{\pi_h^{2(1-\epsilon)}}{Nh_N^2\pi_{h\setminus 0}^2}}\right)$ for every positive number $\epsilon\in(0,1)$.
\end{lemma}
\begin{proof}
By Lemmas  \ref{lemma:influe_eror4}, \ref{lemma:Karamata's1}, \ref{lemma:Karamata's2}, \ref{lemma:Omega_property2}, \ref{lemma:Omega_property}, \ref{lemma:M_inve},  \ref{lemma:Sigma_rate},  \ref{lemma:bias_bound1}, \ref{lemma:bias_bound2}, and \ref{lemma:zeta_moments}, we have 
\begin{eqnarray*}
\|\boldsymbol{\hat\delta}-\boldsymbol{\delta}\|
&\leq&
\frac{1}{\pi_h}\|\mathbf{\hat{V}}^{-1}\|\|(E_N-E)\left[\boldsymbol{\xi}_3\right]\|\|\boldsymbol{\hat\Omega}^{-1}\mathbf{e}_1\|
\\&&
+\frac{1}{\pi_h}\|\mathbf{\hat{V}}^{-1}\|\|E\left[\boldsymbol{\xi}_3\right]\|\|(\boldsymbol{\hat\Omega}^{-1}-\boldsymbol{\Omega}^{-1})\mathbf{e}_1\|
\\&&
+\frac{1}{\pi_h}\|\mathbf{\hat{V}}^{-1}\|\|E\left[\boldsymbol{\xi}_3\right]\boldsymbol{\Omega}^{-1}\mathbf{e}_1\|
\\
&=&
O\left(h_N^{L+1}\right)+O_p\left(\sqrt{\frac{1}{N\pi_h}}\right),
\\
\|(\boldsymbol{\hat\gamma}-\boldsymbol{\gamma})\boldsymbol{\Psi}_{1:L}\|
&\leq&
\frac{1}{\pi_{h\setminus 0}}\|E\left[\boldsymbol{\xi}_2\right]\|\|\boldsymbol{\hat\Pi}^{-1}\|
\\&&
+\frac{1}{\pi_{h\setminus 0}}\|(E_N-E)\left[\boldsymbol{\xi}_2\right]\|\|\boldsymbol{\hat\Pi}^{-1}\|
\\&&
+\frac{1}{\pi_{h\setminus 0}}E_N\left[\|(\mathbf{X}^\ast\mathbf{W}\|\|\mathbf{D}_{1:L}'\boldsymbol{\Psi}_{1:L}^{-1}\|\right]\|\boldsymbol{\hat\delta}-\boldsymbol{\delta}\|\|\boldsymbol{\hat\Pi}^{-1}\|
\\
&=&
O\left(h_N^{L+1}\right)+O_p\left(\sqrt{\frac{1}{N\pi_{h\setminus 0}}}\right),
\\
E_N[\|\boldsymbol{\hat\xi}_1-\boldsymbol{\xi}_1\|^2]^{1/2}
&\leq&
E_N[\boldsymbol{1}\{|D|>h_N\}D^{-2}\|\mathbf{X}^\ast\mathbf{W}\|^2]^{1/2}\|\boldsymbol{\hat\delta}-\boldsymbol{\delta}\|+\|(E_N-E)[\boldsymbol{\xi}_1]\|
\\
&=&
O\left(\sqrt{h_N^{2L}\pi_{h\setminus 0}}\right)+O_p\left(\sqrt{\frac{\pi_{h\setminus 0}}{Nh_N^2\pi_h}}\right),
\\
E_N[\|\boldsymbol{\hat\xi}_2-\boldsymbol{\xi}_2\|^2]^{1/2}
&\leq&
E_N[\|\mathbf{X}^\ast\mathbf{W}\|^2\|\mathbf{D}_{1:L}'\boldsymbol{\Psi}_{1:L}^{-1}\|^2]^{1/2}\|\boldsymbol{\hat\delta}-\boldsymbol{\delta}\|+\|(\boldsymbol{\hat\gamma}-\boldsymbol{\gamma})\boldsymbol{\Psi}_{1:L}\|E_N[\|\mathbf{D}_{1:L}'\boldsymbol{\Psi}_{1:L}^{-1}\|^4]^{1/2}
\\
&=&
O\left(\sqrt{h_N^{2(L+1)}\pi_{h\setminus 0}}\right)+O_p\left(\sqrt{\frac{1}{N}}\right),
\\
E_N[\|\boldsymbol{\hat\xi}_3-\boldsymbol{\xi}_3\|^2]^{1/2}
&\leq&
E_N[\|(\mathbf{X}^\ast\mathbf{W})'\mathbf{X}^\ast\mathbf{W}\|^2\|\|\mathbf{D}_{0:L}'\boldsymbol{\Psi}_{0:L}^{-1}\|^2]^{1/2}\|\boldsymbol{\hat\delta}-\boldsymbol{\delta}\|
\\
&=&
O\left(\sqrt{h_N^{2(L+1)}\pi_h}\right)+O_p\left(\sqrt{\frac{1}{N}}\right),
\mbox{ and }
\\
E_N[\|\boldsymbol{\hat\zeta}-\boldsymbol{\zeta}\|^2]^{1/2}
&=&
E_N[\|\boldsymbol{\hat\xi}_1-\boldsymbol{\xi}_1\|^2]^{1/2}+\frac{1}{\pi_{h\setminus 0}}E_N[\|\boldsymbol{\hat\xi}_2-\boldsymbol{\xi}_2\|^2]^{1/2}\|\boldsymbol{\hat\Pi}^{-1}\|\|\boldsymbol{\Psi}_{1:L}^{-1}\mathbf{\hat{h}}\|
\\&&
+\frac{1}{\pi_h}\|\mathbf{\hat{Q}}\|\|\mathbf{\hat{V}}^{-1}\|E_N[\|\boldsymbol{\hat\xi}_3-\boldsymbol{\xi}_3\|^2]^{1/2}\|\boldsymbol{\hat\Omega}^{-1}\mathbf{e}_1\|
\\&&
+\frac{1}{\pi_{h\setminus 0}}E_N[\|\boldsymbol{\xi}_2\|^2]^{1/2}\left(\left\|\boldsymbol{\hat\Pi}^{-1}-\boldsymbol{\Pi}^{-1}\right\|\|\boldsymbol{\Psi}_{1:L}^{-1}\mathbf{\hat{h}}\|+\|\boldsymbol{\Pi}^{-1}\|\|\boldsymbol{\Psi}_{1:L}^{-1}\left(\mathbf{\hat{h}}-\mathbf{h}\right)\|\right)
\\&&
+
\frac{1}{\pi_h}\|\mathbf{\hat{Q}}\|\left(\|\mathbf{\hat{V}}^{-1}\|\|(\boldsymbol{\hat\Omega}^{-1}-\boldsymbol{\Omega}^{-1})\mathbf{e}_1\|+\left\|\mathbf{\hat{V}}^{-1}-\mathbf{V}^{-1}\right\|\|\boldsymbol{\Omega}^{-1}\mathbf{e}_1\|\right)
E_N[\|\boldsymbol{\xi}_3\|^2]^{1/2}
\\&&
+
\frac{1}{\pi_h}\left\|\mathbf{\hat{Q}}-\mathbf{Q}\right\|\|\mathbf{{V}}^{-1}\|
\|\boldsymbol{\Omega}^{-1}\mathbf{e}_1\|
E_N[\|\boldsymbol{\xi}_3\|^2]^{1/2}
\\
&=&
o\left(1\right)+O_p\left(\sqrt{\frac{\pi_h^{2(1-\epsilon)}}{Nh_N^2\pi_{h\setminus 0}^2}}\right).
\end{eqnarray*}
Therefore, the statement of the lemma follows.
\end{proof}

\section{Details of Section \ref{section:T>p}}\label{supp:sec:main}

\subsection{Asymptotic Properties of $\boldsymbol{\hat\theta}$ with $T>p$}

In this section, we present theoretical results in support of the method of estimation and inference for the case of $T>p$ introduced in Section \ref{section:T>p}. 
We impose the following assumptions. 

\begin{assumption}\label{supp:assn:GP_around0}
(i) $P(D\ne 0)>0$. 
(ii) The conditional distribution of $D$ given $D\ne 0$ has a density $\phi$ with respect to the Lebesgue measure in a neighborhood of zero. 
(iii) There is a positive number $\alpha$ such that $P(0<D\leq u)/u^{\alpha}$ converges to a positive finite constant as $u\rightarrow 0$. 
(iv) If $P(D=0)>0$, then $\alpha\leq 1$.
\end{assumption}

\begin{assumption}\label{supp:assn:more_primitiv_higher}
(i) The mappings $u\mapsto E[(\mathbf{X}'\mathbf{X})^\ast\mathbf{X}'(\mathbf{Y}-\mathbf{W}\boldsymbol{\delta})\mid D=u]$ is $(L+1)$-times differentiable at $0$, and its $(L+1)$-th derivative is  bounded in a neighborhood of zero. 
(ii) The mapping $u\mapsto Var\left((\mathbf{X}'\mathbf{X})^\ast\mathbf{X}'(\mathbf{Y}-\mathbf{W}\boldsymbol{\delta})\mid D=u\right)$ is differentiable at $0$, and its derivative is  bounded in a neighborhood of zero. 
\end{assumption} 

\begin{assumption}\label{supp:assn:iid_plus}
(i) $\{(\mathbf{Y}_i,\mathbf{X}_i)\}_{i=1}^N$ is independent and identically distributed. 
(ii) $E[Z]<\infty$ and $E[Z\mid D=u]$ is bounded  in a neighborhood of zero for 
$Z=
\|(\mathbf{X}'\mathbf{X})^\ast\mathbf{X}'\mathbf{Y}\|^4$,
$\|(\mathbf{X}'\mathbf{X})^\ast\mathbf{X}'\mathbf{W}\|^4$.
(iii) 
$E[\|\mathbf{W}\|^4]<\infty$,
$E[\|\mathbf{W}\|^4\|\mathbf{Y}\|^4]<\infty$, and 
$E[\|\mathbf{W}\|^4\|\mathbf{W}\|^4]<\infty$.
(iv) The minimum eigenvalues of $E[\boldsymbol{1}\{D>0\}\mathbf{W}'\mathbf{M}_{\mathbf{X}}\mathbf{W}]$, $Var((\mathbf{X}'\mathbf{X})^\ast\mathbf{X}'(\mathbf{Y}-\mathbf{W}\boldsymbol{\delta})\mid \mathrm{det}(\mathbf{X}'\mathbf{X})=0)$, and  $Var(\boldsymbol{\eta})$ are bounded away from zero, where
$$
\boldsymbol{\eta}=\boldsymbol{1}\{D>h_N\}\left(D^{-1}(\mathbf{X}'\mathbf{X})^\ast\mathbf{X}'+\left(\mathbf{R}-E\left[\boldsymbol{1}\{D>h_N\}D^{-1}(\mathbf{X}'\mathbf{X})^\ast\mathbf{X}'\mathbf{W}\right]\right)\mathbf{{V}}^{-1}\mathbf{W}\mathbf{M}_{\mathbf{X}}\right)(\mathbf{Y}-\mathbf{W}\boldsymbol{\delta})
$$
with $\mathbf{{V}}=E[\boldsymbol{1}\{D>h_N\}\mathbf{W}'\mathbf{M}_{\mathbf{X}}\mathbf{W}]$.
\end{assumption}

In Assumption \ref{supp:assn:GP_around0}, we assume the tail to be approximately Pareto. Also, we require the density $\phi$ is bounded away from zero when $D$ has a point mass at $0$. 
Assumption \ref{supp:assn:iid_plus} (iv) requires that the influence function $\boldsymbol{\eta}$ for  $\boldsymbol{\hat\beta}^M+\mathbf{R}\boldsymbol{\hat\delta}$ has an invertible variance  with the trimmed mean $\boldsymbol{\hat\beta}^M=E_N[\boldsymbol{1}\{D>h_N\}D^{-1}(\mathbf{X}'\mathbf{X})^\ast\mathbf{X}'(\mathbf{Y}-\mathbf{W}\boldsymbol{\hat\delta})]$.

Now we are going to state the theoretical results for the case of $T>p$. 
The first result is about the bias for $\boldsymbol{\hat\beta}_L$. 
We define the population analog of $\boldsymbol{\hat\beta}_L$ by
$$
\mathbf{b}_L=E[\boldsymbol{1}\{D>h_N\}D^{-1}(\mathbf{X}'\mathbf{X})^\ast\mathbf{X}'(\mathbf{Y}-\mathbf{W}\boldsymbol{\delta})]+\boldsymbol{\gamma}\mathbf{h},
$$
where $\mathbf{h}=E[\boldsymbol{1}\{D\leq h_N\}(\begin{array}{cccc}1&\cdots&D^{L-1}\end{array})']$, $\mathbf{m}(u)=E[(\mathbf{X}'\mathbf{X})^\ast\mathbf{X}'(\mathbf{Y}-\mathbf{W}\boldsymbol{\delta})\mid D=u]$, and $\boldsymbol{\gamma}=(\begin{array}{cccc}\mathbf{m}^{(1)}(0)/1!&\cdots&\mathbf{m}^{(L)}(0)/L!\end{array})$. 
The proof is the same as that of Theorem \ref{bias_comp_higher}.

\begin{theorem}\label{supp:bias_comp_higher}
Under $T>p$ and Assumptions \ref{assn:GP1.1}, \ref{supp:assn:GP_around0}, and \ref{supp:assn:more_primitiv_higher}, $\mathbf{b}_L=\boldsymbol{\beta}+O\left(E[D^L\boldsymbol{1}\{D\leq h_N\}]\right)$. 
\end{theorem}

\noindent
The second result is the asymptotic normality of $\boldsymbol{\hat\theta}$ as an estimator for $\boldsymbol{\theta}$. 

\begin{theorem}\label{supp:theorem_asynormal}
Suppose there is some $c>0$ such that 
\begin{equation}\label{supp:assn:bandwidth}
Nh_N^{2+c}\rightarrow\infty\mbox{ and }Nh_N^{2(L+1)}\rightarrow 0
\end{equation}
as $h_N \rightarrow 0$ and $N \rightarrow \infty$.
Under $T>p$ and Assumptions \ref{assn:GP1.1}, \ref{supp:assn:GP_around0},  \ref{supp:assn:more_primitiv_higher}, and \ref{supp:assn:iid_plus}, 
$$
\sqrt{N}E_N\left[\boldsymbol{\hat\zeta}\boldsymbol{\hat\zeta}'\right]^{-1/2}(\boldsymbol{\hat\theta}-\boldsymbol{\theta})\rightarrow_d\mathcal{N}\left(0,\mathbf{I}_{p}\right)
$$ 
as $h_N \rightarrow 0$ and $N \rightarrow \infty$.
\end{theorem}

\noindent
A proof of this theorem is presented in the subsequent subsection.

\subsection{Proof of Theorem \ref{supp:theorem_asynormal}} 

We use the following notations in the proof of Theorem \ref{supp:theorem_asynormal}. 
We use $\pi_h=P(D\leq h_N)$ and $\pi_{h\setminus 0}=P(0<D\leq h_N)$. 
In the following proof, we characterize the influence function as $\boldsymbol{\zeta}=\boldsymbol{\xi}_1+\frac{1}{\pi_{h\setminus 0}}\boldsymbol{\xi}_2\boldsymbol{\Pi}^{-1}\boldsymbol{\Psi}_{1:L}^{-1}\mathbf{h}+\mathbf{Q}\mathbf{{V}}^{-1}\boldsymbol{\xi}_3$, 
where $\boldsymbol{\Psi}_{1:L}=\mathrm{diag}\left(\begin{array}{cccc}h_N&\cdots&h_N^L\end{array}\right)$, 
\begin{align*}
\boldsymbol{\Pi}
=&
\frac{1}{\pi_{h\setminus 0}}\boldsymbol{\Psi}_{1:L}^{-1}E\left[\mathbf{D}_{1:L}\mathbf{D}_{1:L}'\right]\boldsymbol{\Psi}_{1:L}^{-1},
\\
\mathbf{Q}
=&
\mathbf{R}-E\left[\left(\boldsymbol{1}\{D>h_N\}D^{-1}+\mathbf{D}_{1:L}'E\left[\mathbf{D}_{1:L}\mathbf{D}_{1:L}'\right]^{-1}\mathbf{h}\right)(\mathbf{X}'\mathbf{X})^\ast\mathbf{X}'\mathbf{W}\right],
\\
\boldsymbol{\xi}_1
=&
\boldsymbol{1}\{D>h_N\}D^{-1}(\mathbf{X}'\mathbf{X})^\ast\mathbf{X}'(\mathbf{Y}-\mathbf{W}\boldsymbol{\delta})-E[\boldsymbol{1}\{D>h_N\}D^{-1}(\mathbf{X}'\mathbf{X})^\ast\mathbf{X}'(\mathbf{Y}-\mathbf{W}\boldsymbol{\delta})],
\\
\boldsymbol{\xi}_2
=&
((\mathbf{X}'\mathbf{X})^\ast\mathbf{X}'(\mathbf{Y}-\mathbf{W}\boldsymbol{\delta})-\boldsymbol{\gamma}\mathbf{D}_{1:L})\mathbf{D}_{1:L}'\boldsymbol{\Psi}_{1:L}^{-1},\mbox{ and }
\\
\boldsymbol{\xi}_3
=&
\boldsymbol{1}\{D>h_N\}\mathbf{W}\mathbf{M}_{\mathbf{X}}(\mathbf{Y}-\mathbf{W}\boldsymbol{\delta}).
\end{align*}
Their sample counterparts are 
\begin{align*}
\boldsymbol{\hat\Pi}
=&
\frac{1}{\pi_{h\setminus 0}}\boldsymbol{\Psi}_{1:L}^{-1}E_N\left[\mathbf{D}_{1:L}\mathbf{D}_{1:L}'\right]\boldsymbol{\Psi}_{1:L}^{-1},
\\
\boldsymbol{\hat\xi}_1
=&
\boldsymbol{1}\{D>h_N\}D^{-1}(\mathbf{X}'\mathbf{X})^\ast\mathbf{X}'(\mathbf{Y}-\mathbf{W}\boldsymbol{\hat\delta})-E_N[\boldsymbol{1}\{D>h_N\}D^{-1}(\mathbf{X}'\mathbf{X})^\ast\mathbf{X}'(\mathbf{Y}-\mathbf{W}\boldsymbol{\hat\delta})],
\\
\boldsymbol{\hat\xi}_2
=&
((\mathbf{X}'\mathbf{X})^\ast\mathbf{X}'(\mathbf{Y}-\mathbf{W}\boldsymbol{\hat\delta})-\boldsymbol{\hat\gamma}\mathbf{D}_{1:L})\mathbf{D}_{1:L}'\boldsymbol{\Psi}_{1:L}^{-1},
\mbox{ and }
\\
\boldsymbol{\hat\xi}_3
=&
\boldsymbol{1}\{D>h_N\}\mathbf{W}\mathbf{M}_{\mathbf{X}}(\mathbf{Y}-\mathbf{W}\boldsymbol{\hat\delta}).
\end{align*}
Except for $\boldsymbol{\Psi}_{1:L}$, they are all different from the corresponding ones in the proof of Theorem \ref{theorem_asynormal}. 

Unlike Section \ref{sec:main}, we do not assume that  $\phi$ is bounded away from zero here. 
Thus, the asymptotic properties for $(\pi_{h\setminus 0},\pi_h)$ depends on the tail index $\alpha$. 
By Assumption \ref{supp:assn:GP_around0} (iv), 
$$
\alpha>1\implies \pi_{h\setminus 0}=\pi_h.
$$ 
By Assumption \ref{supp:assn:GP_around0} (iii), the term $\pi_{h\setminus 0}/h_N^{\alpha}$ is bounded away from zero and infinity.

Using the auxiliary results in Section \ref{supp:sec:aux_results}, we are now ready to prove Theorem \ref{supp:theorem_asynormal}.

\begin{proof}[Proof of Theorem \ref{supp:theorem_asynormal}]
First, we are going to show 
\begin{equation}\label{supp:eq:var_estimation}
E_N\left[\boldsymbol{\hat\zeta}\boldsymbol{\hat\zeta}'\right]^{-1}Var(\boldsymbol{\zeta})=\mathbf{I}_{p}+o_p(1).
\end{equation}
By \eqref{supp:assn:bandwidth} and Lemmas \ref{supp:lemma:zeta_moments} and \ref{supp:lemma:xi_estimates}, 
\begin{align*}
E_N[\|\boldsymbol{\hat\zeta}-\boldsymbol{\zeta}\|^2]^{1/2}
=&
o_p\left(1\{\alpha\leq 2\}{\frac{\pi_h}{\pi_{h\setminus 0}}}+1\{\alpha>2\}\right),
\\
E_N[\|\boldsymbol{\zeta}\|^2]^{1/2}
=&
O\left(1\{\alpha\leq 2\}{\frac{\pi_h}{h_N^2}}+1\{\alpha>2\}\right),
\\
(E_N-E)\left[\boldsymbol{\zeta}\boldsymbol{\zeta}'\right]
=&
o_p\left(1\{\alpha\leq 2\}{\frac{\pi_h^2}{h_N^2\pi_{h\setminus 0}}}+1\{\alpha>2\}\right), \qquad\text{and}\\
E[\boldsymbol{\zeta}]
=&
o(1).
\end{align*}
Therefore, 
\begin{align*}
\|E_N\left[\boldsymbol{\hat\zeta}\boldsymbol{\hat\zeta}'\right]-Var(\boldsymbol{\zeta})\|
\leq&
E_N[\|\boldsymbol{\hat\zeta}-\boldsymbol{\zeta}\|^2]+2E_N[\|\boldsymbol{\hat\zeta}-\boldsymbol{\zeta}\|^2]^{1/2}E_N[\|\boldsymbol{\zeta}\|^2]^{1/2}+\|(E_N-E)\left[\boldsymbol{\zeta}\boldsymbol{\zeta}'\right]\|+\|E[\boldsymbol{\zeta}]\|^2
\\
=&
o_p\left(1\{\alpha\leq 2\}{\frac{\pi_h^2}{h_N^2\pi_{h\setminus 0}}}+1\{\alpha>2\}\right).
\end{align*}
By Lemma \ref{supp:lemma:lowerconrate},
$$
E_N\left[\boldsymbol{\hat\zeta}\boldsymbol{\hat\zeta}'\right]^{-1}Var(\boldsymbol{\zeta})-\mathbf{I}_{p}
=
-\left(Var(\boldsymbol{\zeta})+\left(E_N\left[\boldsymbol{\hat\zeta}\boldsymbol{\hat\zeta}'\right]-Var(\boldsymbol{\zeta})\right)\right)^{-1}\left(E_N\left[\boldsymbol{\hat\zeta}\boldsymbol{\hat\zeta}'\right]-Var(\boldsymbol{\zeta})\right)
=
o_p(1).
$$

Second, we are going to show 
\begin{equation}\label{supp:eq:influeFunct_estimation}
\sqrt{N}(\boldsymbol{\hat\theta}-\boldsymbol{\theta}-(E_N-E)[\boldsymbol{\zeta}])=o_p\left(1\{\alpha\leq 2\}\sqrt{\frac{\pi_h^2}{h_N^2\pi_{h\setminus 0}}}+1\{\alpha>2\}\right).
\end{equation}
Note that 
\begin{align*}
\|\boldsymbol{\hat\theta}-\boldsymbol{\theta}-(E_N-E)[\boldsymbol{\zeta}]\|
\leq&
\|\mathbf{b}_L-\boldsymbol{\beta}+E[\boldsymbol{\zeta}]\|
\\&
+\frac{1}{\pi_{h\setminus 0}}
\left(\left\|\boldsymbol{\hat\Pi}^{-1}-\boldsymbol{\Pi}^{-1}\right\|\|\boldsymbol{\Psi}_{1:L}^{-1}\mathbf{\hat{h}}\|+\|\boldsymbol{\Pi}^{-1}\|\|\boldsymbol{\Psi}_{1:L}^{-1}\left(\mathbf{\hat{h}}-\mathbf{h}\right)\|\right)\|E_N\left[\boldsymbol{\xi}_2\right]\|
\\&
+\left(\|\mathbf{Q}\|\|\mathbf{\hat{V}}^{-1}-\mathbf{{V}}^{-1}\|+\|\mathbf{\hat{Q}}-\mathbf{Q}\|\|\mathbf{\hat{V}}^{-1}\|\right)\|E_N[\boldsymbol{\xi}_3]\|
\\&
+\|\boldsymbol{\gamma}\|\left\|\mathbf{\hat{h}}-\mathbf{h}\right\|.
\end{align*}
Lemmas \ref{supp:lemma:influe_eror4}, \ref{supp:lemma:Omega_property2}, \ref{supp:lemma:Sigma_rate}, \ref{supp:lemma:M_inve}, \ref{supp:lemma:bias_bound2},  \ref{supp:lemma:bias_bound1}, and \ref{supp:lemma:zeta_moments} and Theorem \ref{supp:bias_comp_higher} bound each term on the right hand side of the above equation, which leads to 
$$
\boldsymbol{\hat\theta}-\boldsymbol{\theta}-(E_N-E)[\boldsymbol{\zeta}]
=
O\left(h_N^{L}\pi_h\right)
+
O_p\left(
\sqrt{\frac{\pi_h}{N}}+
1\{\alpha\leq 2\}
\sqrt{\frac{\pi_h^2}{N^2h_N^4\pi_{h\setminus 0}}}
+1\{\alpha>2\}\frac{1}{Nh_N}
\right),
$$
where we can take a sufficiently small $\epsilon>0$. 
By  \eqref{supp:assn:bandwidth}, we have \eqref{supp:eq:influeFunct_estimation}.

Third, we are going to show 
\begin{equation}\label{supp:eq:L2L4_estimation}
E[\|Var(\boldsymbol{\zeta})^{-1/2}\boldsymbol{\zeta}\|^4]=o(N).
\end{equation}
By Lemmas  \ref{supp:lemma:zeta_moments} and \ref{supp:lemma:lowerconrate}, 
$$
E[\|Var(\boldsymbol{\zeta})^{-1/2}\boldsymbol{\zeta}\|^4]=O\left(1\{\alpha\leq 2\}\frac{1}{\pi_{h\setminus 0}^{1+\epsilon}}+1\{2<\alpha\leq 4\}\frac{\pi_{h\setminus 0}^{1-\epsilon}}{h_N^4}+1\{\alpha>4\}\right),
$$
where we can take a sufficiently small $\epsilon>0$. 
By \eqref{supp:assn:bandwidth}, we have \eqref{supp:eq:L2L4_estimation}.

Now we are going to show the statement of this theorem. 
By Lemma  \ref{supp:lemma:lowerconrate},  \eqref{supp:eq:var_estimation} and \eqref{supp:eq:influeFunct_estimation}, 
$$
\sqrt{N}E_N\left[\boldsymbol{\hat\zeta}\boldsymbol{\hat\zeta}'\right]^{-1/2}(\boldsymbol{\hat\theta}-\boldsymbol{\theta})=\sqrt{N}Var(\boldsymbol{\zeta})^{-1/2}(E_N-E)[\boldsymbol{\zeta}]+o_p(1).
$$
By \eqref{supp:eq:L2L4_estimation}, Assumption \ref{supp:assn:iid_plus} (i), and  Lyapunov's central limit theorem, 
$
\sqrt{N}E_N\left[\boldsymbol{\hat\zeta}\boldsymbol{\hat\zeta}'\right]^{-1/2}(\boldsymbol{\hat\theta}-\boldsymbol{\theta})\rightarrow_d\mathcal{N}\left(0,\mathbf{I}_{p}\right).
$
\end{proof}

\section{Auxiliary Results for Theorem \ref{supp:theorem_asynormal}}\label{supp:sec:aux_results}

We introduce and prove auxiliary lemmas, which we use for the proof of Theorem \ref{supp:theorem_asynormal}. 
In all the lemma in this section, we impose the assumptions in the statement of Theorem \ref{supp:theorem_asynormal} as well as $T>p$.

\begin{lemma}\label{supp:lemma:influe_eror4}
$\boldsymbol{\Psi}_{1:L}^{-1}\mathbf{h}=O(\frac{\pi_h}{h_N})$, $\left\|\frac{h_N}{\pi_h}\boldsymbol{\Psi}_{1:L}^{-1}\mathbf{h}\right\|$ is bounded away from zero,  and $\boldsymbol{\Psi}_{1:L}^{-1}(\mathbf{\hat{h}}-\mathbf{h})=O_p\left(\sqrt{\frac{\pi_h}{Nh_N^2}}\right)$.
\end{lemma}
\begin{proof}
While the definition of $D$ is different, the proof is the same as that of Lemma \ref{lemma:influe_eror4}.
\end{proof}

\begin{lemma}\label{supp:lemma:Karamata's1}
(i) For every number $l\geq 1$,  
$$
\frac{1}{h_N^l\pi_{h\setminus 0}} E[D^l\boldsymbol{1}\{D\leq h_N\}] =\frac{\alpha}{l+\alpha}+o(1)
$$ 
as $h_N \rightarrow 0$.
(ii) For every pair of positive numbers $l_1,l_2$, $E\left[D^{l_1}\left\|\mathbf{D}_{1:L}'\boldsymbol{\Psi}_{1:L}^{-1}\right\|^{l_2}\right]=O(\pi_{h\setminus 0}h_N^{l_1})$. 
(iii) $E\left[\|Z\|\|\mathbf{D}_{1:L}'\boldsymbol{\Psi}_{1:L}^{-1}\|^l\right]=O\left(\pi_{h\setminus 0}\right)$ for every random variable $Z$ in Assumption \ref{supp:assn:iid_plus} (iii) and for every positive integer $l$. 
\end{lemma}
\begin{proof}
First, we are going to show the statement in (i). 
Under Assumption \ref{supp:assn:GP_around0}, Karamata's theorem \citep[Theorem 2.1 and Exercise 2.5]{resnick2007heavy} implies 
$$
\frac{1}{h_N^l\pi_{h\setminus 0}}
E\left[\boldsymbol{1}\{D\leq h_N\}D^l\right]
=
\frac{E\left[\boldsymbol{1}\{D^{-1}\geq h_N^{-1}\}(D^{-1})^{-l}\mid D\ne 0\right]}{(h_N^{-1})^{-l}P(D^{-1}\geq h_N^{-1}\mid D\ne 0)}
=
\frac{\alpha}{l+\alpha}+o(1).
$$
Therefore, the statement in (i) holds. 

Next, we are going to show the statement in (ii). 
Note that 
$$
D^{l_1}\left\|\mathbf{D}_{1:L}'\boldsymbol{\Psi}_{1:L}^{-1}\right\|^{l_2}
=
D^{l_1}\boldsymbol{1}\{0<D\leq h_N\}\left\|\left(\begin{array}{cccc}\frac{D}{h_N}&\cdots&\frac{D^{L}}{h_N^{L}}\end{array}\right)'\right\|^{l_2}
\leq
L^{l_2/2}h_N^{l_1}\boldsymbol{1}\{0<D\leq h_N\}.
$$
Therefore, 
$$
E\left[D^{l_1}\left\|\mathbf{D}_{1:L}'\boldsymbol{\Psi}_{1:L}^{-1}\right\|^{l_2}\right]
\leq
L^{l_2/2}h_N^{l_1}E[\boldsymbol{1}\{0<D\leq h_N\}]
=
O(\pi_{h\setminus 0}h_N^{l_1}).
$$ 

Third, we are going to show the statement in (iii). 
There is a positive constant $c$ such that $\sup_{u\in[-c,c]}E[\|Z\|\mid D=u]<\infty$. 
By the statement in (ii),  
$$
E\left[\|Z\|\|\mathbf{D}_{1:L}'\boldsymbol{\Psi}_{1:L}^{-1}\|^l\right]
\leq
E\left[\|\mathbf{D}_{1:L}'\boldsymbol{\Psi}_{1:L}^{-1}\|^l\right]\sup_{u\in[-h_N,h_N]}E[\|Z\|\mid D=u]
=
O\left(\pi_{h\setminus 0}\right).
$$
\end{proof}

\begin{lemma}\label{supp:lemma:Karamata's2}
(i) For every positive number $l$,
$$
E\left[\boldsymbol{1}\{h_N<D\}D^{-l}\right]=O(1)\mbox{ if }l<\alpha
$$
and 
$$
\frac{h_N^l}{\pi_{h\setminus 0}}E\left[\boldsymbol{1}\{h_N<D\}D^{-l}\right]=\frac{\alpha}{l-\alpha}+o(1)\mbox{ if }l>\alpha.
$$
(ii)$E\left[\boldsymbol{1}\{D>h_N\}D^{-l}\|Z\|\right]=O\left(1\{\alpha\leq l\}\frac{\pi_{h\setminus 0}^{1-\epsilon}}{h_N^l}+1\{\alpha>l\}\right)$ for every positive number $\epsilon\in(0,1)$ for every random variable $Z$ in Assumption \ref{supp:assn:iid_plus} (ii) and for every positive integer $l$. 
\end{lemma}
\begin{proof}
First, we are going to show the statement in (i). 
Under Assumption \ref{supp:assn:GP_around0}, Karamata's theorem \citep[Theorem 2.1 and Exercise 2.5]{resnick2007heavy} implies that 
$$
\frac{h_N^l}{\pi_{h\setminus 0}}E\left[\boldsymbol{1}\{D>h_N\}D^{-l}\right]
=
\frac{E\left[\boldsymbol{1}\{D^{-1}\leq h_N^{-1}\}(D^{-1})^l\mid D\ne 0\right]}{
(h_N^{-1})^lP(D^{-1}\geq h_N^{-1}\mid D\ne 0)}
=
\frac{\alpha}{l-\alpha}+o(1).
$$

Next, we are going to show the statement in (ii). 
There is a positive constant $c$ such that $\sup_{u\in[-c,c]}E[\|Z\|\mid D=u]<\infty$. 
Note that
\begin{align*}
\left\|E\left[\boldsymbol{1}\{|D|>h_N\}|D|^{-l}\|Z\|\right]\right\|
\leq&
\left\|E\left[\boldsymbol{1}\{|D|>c\}|D|^{-l}\|Z\|\right]\right\|+\left\|E\left[\boldsymbol{1}\{h_N<|D|\leq c\}|D|^{-l}\|Z\|\right]\right\|
\\
\leq&
c^{-l}E\left[\|Z\|\right]+E\left[\boldsymbol{1}\{h_N<|D|\}|D|^{-l}\right]\sup_{u\in[-c,c]}E[\|Z\|\mid D=u].
\end{align*}
By H\"older's inequality and the statement in (i), 
$$
E\left[\boldsymbol{1}\{|D|>h_N\}|D|^{-l}\right]
=
\begin{cases}
O(1)&\mbox{ if }\alpha>l
\\
O\left(\frac{\pi_{h\setminus 0}}{h_N^l}\right)&\mbox{ if }\alpha<l
\\
O(E\left[\boldsymbol{1}\{|D|>h_N\}|D|^{-l/{1-\epsilon}}\right]^{1-\epsilon})
=
O\left(\frac{\pi_{h\setminus 0}}{h_N^{l/{1-\epsilon}}}\right)^{1-\epsilon}
=
O\left(\frac{\pi_{h\setminus 0}^{1-\epsilon}}{h_N^l}\right)&\mbox{ if }\alpha=l,
\end{cases}
$$
which implies the statement in (ii). 
\end{proof}

\begin{lemma}\label{supp:lemma:Omega_property2}
The minimum eigenvalue of $\boldsymbol{\Pi}$ is bounded away from zero,
the maximum eigenvalue of $\boldsymbol{\Pi}$ is bounded,
 and  $\boldsymbol{\hat\Pi}^{-1}-\boldsymbol{\Pi}^{-1}=O_p\left(\frac{1}{\sqrt{N\pi_{h\setminus 0}}}\right)$.
\end{lemma}
\begin{proof}
While the definition of $D$ is different, the proof is the same as that of Lemma \ref{lemma:Omega_property2}.
\end{proof}

\begin{lemma}\label{supp:lemma:M_inve}
$\mathbf{V}^{-1}=O(1)$, and $\mathbf{\hat{V}}^{-1}-\mathbf{V}^{-1}=O_p\left(\sqrt{\frac{1}{N}}\right)$. 
\end{lemma}
\begin{proof}
By Assumption \ref{supp:assn:iid_plus} (iv), the minimum eigenvalue of $\mathbf{V}$ is bounded away from zero, which implies $\mathbf{V}^{-1}=O(1)$. 
Since $\mathbf{V}^{-1}=O(1)$, it suffices to show $\mathbf{\hat{V}}-\mathbf{V}=O_p\left(\sqrt{\frac{1}{N}}\right)$ to show $\mathbf{\hat{V}}^{-1}-\mathbf{V}^{-1}=O_p\left(\sqrt{\frac{1}{N}}\right)$. 
Let $\mathbf{c}$ be any $(T-1)p$-dimensional column vector with $\|c\|=1$. 
Together with Assumption \ref{supp:assn:iid_plus} (i), we have
$$
E\left[\left\|(\mathbf{\hat{V}}-\mathbf{V})\mathbf{c}\right\|^2\right]^{1/2}
=
E\left[\left\|(E_N-E)\left[\boldsymbol{1}\{D>h_N\}\mathbf{W}'\mathbf{M}_{\mathbf{X}}\mathbf{W}\mathbf{c}\right]\right\|^2\right]^{1/2}
\leq
\frac{1}{\sqrt{N}}E\left[\|\mathbf{W}\|^4\|\mathbf{M}_{\mathbf{X}}\|^2\right]^{1/2}.
$$
Since $\|\mathbf{M}_{\mathbf{X}}\|\leq 1$ and $E\left[\|\mathbf{W}\|^4\right]<\infty$, we have $\mathbf{\hat{V}}-\mathbf{V}=O_p\left(\sqrt{\frac{1}{N}}\right)$.
\end{proof}

\begin{lemma}\label{supp:lemma:Sigma_rate}
For every $\epsilon\in(0,1)$, $\mathbf{Q}=O(1\{\alpha\leq 1\}\frac{\pi_h^{1-\epsilon}}{h_N}+1\{\alpha>1\})$, and $\mathbf{\hat{Q}}-\mathbf{Q}=O_p\Big(1\{\alpha\leq 2\}\sqrt{\frac{\pi_h^{2-\epsilon}}{Nh_N^2\pi_{h\setminus 0}}}+1\{\alpha>2\}\sqrt{\frac{1}{N}}\Big)$.
\end{lemma}
\begin{proof}
Let $\mathbf{c}$ be any $(T-1)p$-dimensional column vector with $\|\mathbf{c}\|=1$. 
The the first statement follows from Lemmas \ref{supp:lemma:influe_eror4}, \ref{supp:lemma:Karamata's1} and \ref{supp:lemma:Karamata's2} and 
\begin{eqnarray*}
\mathbf{Q}\mathbf{c}
&=&
\mathbf{R}\mathbf{c}+E\left[\boldsymbol{1}\{D>h_N\}D^{-1}(\mathbf{X}'\mathbf{X})^\ast\mathbf{X}'\mathbf{W}\mathbf{c}\right]+\frac{1}{\pi_{h\setminus 0}}E\left[(\mathbf{X}'\mathbf{X})^\ast\mathbf{X}'\mathbf{W}\mathbf{c}\mathbf{D}_{1:L}'\boldsymbol{\Psi}_{1:L}^{-1}\right]\boldsymbol{\Pi}^{-1}\boldsymbol{\Psi}_{1:L}^{-1}\mathbf{h}
\\
&=&
O\left(1\{\alpha\leq 1\}\frac{\pi_h^{1-\epsilon}}{h_N}+1\{\alpha>1\}\right).
\end{eqnarray*}
Assumption \ref{supp:assn:iid_plus} (i) and Lemmas \ref{supp:lemma:Karamata's1} and \ref{supp:lemma:Karamata's2} imply
$$
\|E\left[(\mathbf{X}'\mathbf{X})^\ast\mathbf{X}'\mathbf{W}\mathbf{c}\mathbf{D}_{1:L}'\boldsymbol{\Psi}_{1:L}^{-1}\right]\|
\leq
E\left[\|(\mathbf{X}'\mathbf{X})^\ast\mathbf{X}'\mathbf{W}\mathbf{c}\|\|\mathbf{D}_{1:L}'\boldsymbol{\Psi}_{1:L}^{-1}\|\right]
=
O\left(\pi_{h\setminus 0}\right),
$$
$$
E\left[\left\|(E_N-E)\left[(\mathbf{X}'\mathbf{X})^\ast\mathbf{X}'\mathbf{W}\mathbf{c}\mathbf{D}_{1:L}'\boldsymbol{\Psi}_{1:L}^{-1}\right]\right\|^2\right]^{1/2}
\leq
\frac{1}{\sqrt{N}}E\left[\left\|(\mathbf{X}'\mathbf{X})^\ast\mathbf{X}'\mathbf{W}\mathbf{c}\right\|^2\left\|\mathbf{D}_{1:L}'\boldsymbol{\Psi}_{1:L}^{-1}\right\|^2\right]^{1/2}
=
O\left(\sqrt{\frac{\pi_{h\setminus 0}}{N}}\right),
$$
and 
\begin{align*}
E\left[\left\|(E_N-E)\left[(\mathbf{X}'\mathbf{X})^\ast\mathbf{X}'\mathbf{W}\boldsymbol{1}\{D>h_N\}D^{-1}\right]\right\|^2\right]^{1/2}
\leq&
\frac{1}{\sqrt{N}}E\left[\left\|(\mathbf{X}'\mathbf{X})^\ast\mathbf{X}'\mathbf{W}\right\|^2\boldsymbol{1}\{D>h_N\}D^{-2}\right]^{1/2}
\\
=&
O\left(1\{\alpha\leq 2\}\sqrt{\frac{\pi_{h\setminus 0}^{1-\epsilon}}{Nh_N^2}}+1\{\alpha>2\}\sqrt{\frac{1}{N}}\right).
\end{align*}
Since 
\begin{align*}
(\mathbf{\hat{Q}}-\mathbf{Q})\mathbf{c}
=&
\frac{1}{\pi_{h\setminus 0}}E\left[(\mathbf{X}'\mathbf{X})^\ast\mathbf{X}'\mathbf{W}\mathbf{c}\mathbf{D}_{1:L}'\boldsymbol{\Psi}_{1:L}^{-1}\right](\boldsymbol{\hat\Pi}^{-1}-\boldsymbol{\Pi}^{-1})\boldsymbol{\Psi}_{1:L}^{-1}\mathbf{h}
\\&
+\frac{1}{\pi_{h\setminus 0}}(E_N-E)\left[(\mathbf{X}'\mathbf{X})^\ast\mathbf{X}'\mathbf{W}\mathbf{c}\mathbf{D}_{1:L}'\boldsymbol{\Psi}_{1:L}^{-1}\right]\boldsymbol{\hat\Pi}^{-1}\boldsymbol{\Psi}_{1:L}^{-1}\mathbf{h}
\\&
+(E_N-E)\left[(\mathbf{X}'\mathbf{X})^\ast\mathbf{X}'\mathbf{W}\boldsymbol{1}\{D>h_N\}D^{-1}\right]\mathbf{c},
\end{align*}
together with Lemmas \ref{supp:lemma:influe_eror4} and \ref{supp:lemma:Omega_property2}, we have
$$
(\mathbf{\hat{Q}}-\mathbf{Q})\mathbf{c}=O_p\left(1\{\alpha\leq 2\}\sqrt{\frac{\pi_h^{2-\epsilon}}{Nh_N^2\pi_{h\setminus 0}}}+1\{\alpha>2\}\sqrt{\frac{1}{N}}\right).
$$
\end{proof}

\begin{lemma}\label{supp:lemma:fourth1}
For every $\epsilon\in(0,1)$, the following statements hold: 
(i) $E[\boldsymbol{\xi}_1]=0$. 
(ii) $E[\|\boldsymbol{\xi}_1\|^2]=O(1\{\alpha\leq 2\}\frac{\pi_{h\setminus 0}^{1-\epsilon}}{h_N^2}+1\{\alpha>2\})$ and $E[\|\boldsymbol{\xi}_1\|^4]=O(1\{\alpha\leq 4\}\frac{\pi_{h\setminus 0}^{1-\epsilon}}{h_N^4}+1\{\alpha>4\})$.
(iii) $(E_N-E)[\boldsymbol{\xi}_1]=O_p(1\{\alpha\leq 2\}\sqrt{\frac{\pi_{h\setminus 0}^{1-\epsilon}}{Nh_N^2}}+1\{\alpha>2\}\sqrt{\frac{1}{N}})$.
\end{lemma}
\begin{proof}
The statement in (i) follows from the definition of $\boldsymbol{\xi}_1$. 
The statement in (ii) follows from Lemma \ref{supp:lemma:Karamata's2}.
The statement in (iii) follows from Assumption \ref{supp:assn:iid_plus} (i) and the statement in (ii).
\end{proof}

\begin{lemma}\label{supp:lemma:bias_bound1}
(i) $E[\boldsymbol{\xi}_2]=O(h_N^{L+1}\pi_{h\setminus 0})$.
(ii) $E[\|\boldsymbol{\xi}_2\|^2]=O\left(\pi_{h\setminus 0}\right)$ and $
E[\|\boldsymbol{\xi}_2\|^4]=O\left(\pi_{h\setminus 0}\right)$. 
(iii) $(E_N-E)[\boldsymbol{\xi}_2]=O(\sqrt{\frac{\pi_{h\setminus 0}}{N}})$.
\end{lemma}
\begin{proof}
By Assumption \ref{supp:assn:more_primitiv_higher} (i) and $\mathbf{m}(0)=0$, there is a constant ${C}$ such that $\left\|\mathbf{m}(u)-\sum_{l=1}^{L}\mathbf{m}^{(l)}(0)\frac{u^{l}}{l!}\right\|\leq {C}|u|^{L+1}$ for every $u\in[-h_N,h_N]$ when $h_N$ is sufficiently small. 
Therefore, for every $u\in[-h_N,h_N]$, we have
\begin{align*}
\left\|E\left[\boldsymbol{\xi}_2\right]\right\|
=&
\left\|E\left[\left(\mathbf{m}(D)-\sum_{l=1}^{L}\mathbf{m}^{(l)}(0)\frac{D^{l}}{l!}\right)\mathbf{D}_{1:L}'\boldsymbol{\Psi}_{1:L}^{-1}\right]\right\|
\\
\leq&
E\left[\left\|\mathbf{m}(D)-\sum_{l=1}^{L}\mathbf{m}^{(l)}(0)\frac{D^{l}}{l!}\right\|\left\|\mathbf{D}_{1:L}'\boldsymbol{\Psi}_{1:L}^{-1}\right\|\right]
\\
\leq&
{C}E\left[D^{L+1}\left\|\mathbf{D}_{1:L}'\boldsymbol{\Psi}_{1:L}^{-1}\right\|\right]
\\
=&
O(\pi_{h\setminus 0}h_N^{L+1}),
\end{align*}
where the last equality follows from Lemma \ref{supp:lemma:Karamata's1}.
This shows the statement in (i).
The statement in (ii) follows from Lemma \ref{supp:lemma:Karamata's1}. 
The statement in (iii) follows from Assumption \ref{supp:assn:iid_plus} (i) and the statement in (ii).
\end{proof}

\begin{lemma}\label{supp:lemma:bias_bound2}
For every $\epsilon\in(0,1)$, the following statements hold: 
(i) $E\left[\boldsymbol{\xi}_3\right]=0$.
(ii) $E[\|\boldsymbol{\xi}_3\|^2]=O\left(1\right)$ and $E[\|\boldsymbol{\xi}_3\|^4]=O\left(1\right)$.
(iii) $(E_N-E)\left[\boldsymbol{\xi}_3\right]=O_p\left(\sqrt{\frac{1}{N}}\right)$.
\end{lemma}
\begin{proof}
The first statement follows from \eqref{supp:eq:base_ID1}.
Since $\|\boldsymbol{\xi}_3\|\leq\|\mathbf{W}\|\|\mathbf{M}_{\mathbf{X}}\|\|\mathbf{Y}-\mathbf{W}\boldsymbol{\delta}\|$ and $\|\mathbf{M}_{\mathbf{X}}\|\leq 1$,
we have 
$$
E[\|\boldsymbol{\xi}_3\|^l]^{1/l}\leq E[\|\mathbf{W}\|^l\|\mathbf{Y}-\mathbf{W}\boldsymbol{\delta}\|^l]^{1/l}
$$
for $l=2,4$, which implies the statement in (ii).
The statement in (iii) follows from Assumption \ref{supp:assn:iid_plus} (i) and the statement in (i). 
\end{proof}

\begin{lemma}\label{supp:lemma:zeta_moments}
For every $\epsilon\in(0,1)$, the following statements hold: 
(i) $E[\boldsymbol{\zeta}]=O(h_N^{L}\pi_h)$.
(ii) $E[\|\boldsymbol{\zeta}\|^2]=O(1\{\alpha\leq 1\}\frac{\pi_h^{2-\epsilon}}{h_N^2\pi_{h\setminus 0}}+1\{1<\alpha\leq 2\}\frac{\pi_{h\setminus 0}^{1-\epsilon}}{h_N^2}+1\{\alpha>2\})$ and 
$
E[\|\boldsymbol{\zeta}\|^4]=O(1\{\alpha\leq 1\}\frac{\pi_h^{4-\epsilon}}{h_N^4\pi_{h\setminus 0}^3}+1\{1<\alpha\leq 4\}\frac{\pi_{h\setminus 0}^{1-\epsilon}}{h_N^4}+1\{\alpha>4\})$.
(iii) 
$(E_N-E)[\boldsymbol{\zeta}\boldsymbol{\zeta}']=O_p\left(1\{\alpha\leq 1\}\sqrt{\frac{\pi_h^{4-\epsilon}}{Nh_N^4\pi_{h\setminus 0}^3}}+1\{1<\alpha\leq 4\}\sqrt{\frac{\pi_{h\setminus 0}^{1-\epsilon}}{Nh_N^4}}+1\{\alpha>4\}\sqrt{\frac{1}{N}}\right)$.
\end{lemma}
\begin{proof}
By  Lemmas \ref{supp:lemma:influe_eror4}, \ref{supp:lemma:Omega_property2}, and \ref{supp:lemma:bias_bound1}, we have
$$
E[\boldsymbol{\zeta}]
=
\frac{1}{\pi_{h\setminus 0}}E[\boldsymbol{\xi}_2]\boldsymbol{\Pi}^{-1}\boldsymbol{\Psi}_{1:L}^{-1}\mathbf{h}
=
O(h_N^{L}\pi_h),
$$
which proves the statement in (i). 
Note that 
$$
E[\|\boldsymbol{\zeta}\|^l]^{1/l}\leq E[\|\boldsymbol{\xi}_1\|^l]^{1/l}+\frac{1}{\pi_{h\setminus 0}}E[\|\boldsymbol{\xi}_2\|^l]^{1/l}\|\boldsymbol{\Pi}^{-1}\|\|\boldsymbol{\Psi}_{1:L}^{-1}\mathbf{h}\|+\|\mathbf{Q}\|\|\mathbf{{V}}^{-1}\|E[\|\boldsymbol{\xi}_3\|^l]^{1/l}
$$
for $l=2,4$. 
By Lemmas \ref{supp:lemma:Sigma_rate}, \ref{supp:lemma:fourth1}, \ref{supp:lemma:bias_bound1}, and \ref{supp:lemma:bias_bound2}, we have 
$$
E[\|\boldsymbol{\zeta}\|^l]=O\left(1\{\alpha\leq 1\}\frac{\pi_h^{l-\epsilon}}{h_N^l\pi_{h\setminus 0}^{l-1}}+1\{1<\alpha\leq l\}\frac{\pi_{h\setminus 0}^{1-\epsilon}}{h_N^l}+1\{\alpha>l\}\right),
$$
which leads to the statement in (ii). 
The statement in (iii) follows from Assumption \ref{supp:assn:iid_plus} (i) and the statement in (ii). 
\end{proof}

\begin{lemma}\label{supp:lemma:lowerconrate}
The minimum eigenvalue of $\left(1\{\alpha\leq 2\}\frac{h_N^2\pi_{h\setminus 0}}{\pi_h^2}+1\{\alpha>2\}\right)Var(\boldsymbol{\zeta})$ is bounded away from zero.
\end{lemma}
\begin{proof}
Let $\mathbf{c}$ be any $p$-dimensional column vector with $\|\mathbf{c}\|=1$. 
In $\boldsymbol{\zeta}$, the terms $\left(\boldsymbol{\xi}_1+\mathbf{Q}\mathbf{{V}}^{-1}\boldsymbol{\xi}_3\right)$ and $\frac{1}{\pi_{h\setminus 0}}\boldsymbol{\xi}_2\boldsymbol{\Pi}^{-1}\boldsymbol{\Psi}_{1:L}^{-1}\mathbf{h}$ are independent. 
Therefore, 
$$
Var(\mathbf{c}'\boldsymbol{\zeta})\geq\max\left\{Var(\mathbf{c}'\left(\boldsymbol{\xi}_1+\mathbf{Q}\mathbf{{V}}^{-1}\boldsymbol{\xi}_3\right)),\frac{1}{\pi_{h\setminus 0}^2}Var(\mathbf{c}'\boldsymbol{\xi}_2\boldsymbol{\Pi}^{-1}\boldsymbol{\Psi}_{1:L}^{-1}\mathbf{h})\right\}.
$$

First, we are going to show that $\frac{h_N^2}{\pi_h^2\pi_{h\setminus 0}}Var(\mathbf{c}'\boldsymbol{\xi}_2\boldsymbol{\Pi}^{-1}\boldsymbol{\Psi}_{1:L}^{-1}\mathbf{h})$ is bounded away from zero when $\alpha\leq 2$. 
By Assumptions \ref{supp:assn:more_primitiv_higher} (ii) and \ref{supp:assn:iid_plus} (iv),
there is a positive constant $C$ such that $\mathbf{c}'Var\left((\mathbf{X}'\mathbf{X})^\ast\mathbf{X}'(\mathbf{Y}-\mathbf{W}\boldsymbol{\delta})\mid D=u\right)\mathbf{c}\geq C$ for every $u\in[-h_N,h_N]$ for sufficiently small $h_N$.
By the conditional variance identity, 
$$
Var(\mathbf{c}'\boldsymbol{\xi}_2\boldsymbol{\Pi}^{-1}\boldsymbol{\Psi}_{1:L}^{-1}\mathbf{h})\geq E[Var((\boldsymbol{\xi}_2\boldsymbol{\Pi}^{-1}\boldsymbol{\Psi}_{1:L}^{-1}\mathbf{h})'\mathbf{c}\mid D)].
$$
Since 
\begin{align*}
Var((\boldsymbol{\xi}_2\boldsymbol{\Pi}^{-1}\boldsymbol{\Psi}_{1:L}^{-1}\mathbf{h})'\mathbf{c}\mid D)
=&
\left(\boldsymbol{\Psi}_{1:L}^{-1}\mathbf{h}\right)'\boldsymbol{\Pi}^{-1}(\mathbf{c}'Var((\mathbf{X}'\mathbf{X})^\ast\mathbf{X}'(\mathbf{Y}-\mathbf{W}\boldsymbol{\delta})\mid D)\mathbf{c})\boldsymbol{\Psi}_{1:L}^{-1}\mathbf{D}_{1:L}\mathbf{D}_{1:L}'\boldsymbol{\Psi}_{1:L}^{-1}\boldsymbol{\Pi}^{-1}\boldsymbol{\Psi}_{1:L}^{-1}\mathbf{h}
\\
\geq&
C\left(\boldsymbol{\Psi}_{1:L}^{-1}\mathbf{h}\right)'\boldsymbol{\Pi}^{-1}\boldsymbol{\Psi}_{1:L}^{-1}\mathbf{D}_{1:L}\mathbf{D}_{1:L}'\boldsymbol{\Psi}_{1:L}^{-1}\boldsymbol{\Pi}^{-1}\boldsymbol{\Psi}_{1:L}^{-1}\mathbf{h},
\end{align*}
we have 
$$
Var(\mathbf{c}'\boldsymbol{\xi}_2\boldsymbol{\Pi}^{-1}\boldsymbol{\Psi}_{1:L}^{-1}\mathbf{h})\geq \pi_{h\setminus 0}C\left(\boldsymbol{\Psi}_{1:L}^{-1}\mathbf{h}\right)'\boldsymbol{\Pi}^{-1}\boldsymbol{\Psi}_{1:L}^{-1}\mathbf{h}.
$$
By Lemmas \ref{supp:lemma:influe_eror4} and \ref{supp:lemma:Omega_property2}, $\frac{h_N^2}{\pi_h^2\pi_{h\setminus 0}}Var(\mathbf{c}'\boldsymbol{\xi}_2\boldsymbol{\Pi}^{-1}\boldsymbol{\Psi}_{1:L}^{-1}\mathbf{h})$ is bounded away from zero. 

Second, we are going to show that $Var(\mathbf{c}'(\boldsymbol{\xi}_1+\mathbf{Q}\mathbf{{V}}^{-1}\boldsymbol{\xi}_3))$ is bounded away from zero when $\alpha>2$. 
By Assumption \ref{supp:assn:iid_plus} (iv), $Var(\mathbf{c}'\boldsymbol{\eta})$ is bounded away from zero. 
Note that 
\begin{align*}
E\left[\left\|\mathbf{c}'\left(\boldsymbol{\eta}-(\boldsymbol{\xi}_1+\mathbf{Q}\mathbf{{V}}^{-1}\boldsymbol{\xi}_3)\right)\right\|^2\right]^{1/2}
=&
E\left[\left\|\frac{1}{\pi_{h\setminus 0}}\boldsymbol{\xi}_3'\mathbf{{V}}^{-1}E\left[\mathbf{W}'\mathbf{X}(\mathbf{X}'\mathbf{X})^\ast\mathbf{c}\mathbf{D}_{1:L}'\boldsymbol{\Psi}_{1:L}^{-1}\right]\boldsymbol{\Pi}^{-1}\boldsymbol{\Psi}_{1:L}^{-1}\mathbf{h}\right\|^2\right]^{1/2}
\\
\leq&
\frac{1}{\pi_{h\setminus 0}}E\left[\left\|\boldsymbol{\xi}_3\right\|^2\right]^{1/2}
\|\mathbf{{V}}^{-1}\|E\left[\|\mathbf{W}'\mathbf{X}(\mathbf{X}'\mathbf{X})^\ast\|\|\mathbf{D}_{1:L}'\boldsymbol{\Psi}_{1:L}^{-1}\|\right]\|\boldsymbol{\Pi}^{-1}\|\|\boldsymbol{\Psi}_{1:L}^{-1}\mathbf{h}\|
\\
=&
O\left(\frac{\pi_h}{h_N}\right)
\\
=&
o(1),
\end{align*}
where the last equality uses $\alpha>1$.
Therefore, $Var(\mathbf{c}'\left(\boldsymbol{\xi}_1+\mathbf{Q}\mathbf{{V}}^{-1}\boldsymbol{\xi}_3\right))$ is bounded away from zero. 
\end{proof}

\begin{lemma}\label{supp:lemma:xi_estimates}
$E_N[\|\boldsymbol{\hat\zeta}-\boldsymbol{\zeta}\|^2]^{1/2}=o(1)+O_p\left(\sqrt{\frac{\pi_h^2}{Nh_N^2\pi_{h\setminus 0}^2}}\right)$.
\end{lemma}
\begin{proof}
By Lemmas  \ref{supp:lemma:influe_eror4}, \ref{supp:lemma:Karamata's1}, \ref{supp:lemma:Karamata's2}, \ref{supp:lemma:Omega_property2}, \ref{supp:lemma:M_inve},  \ref{supp:lemma:Sigma_rate},  \ref{supp:lemma:bias_bound1}, \ref{supp:lemma:bias_bound2},  \ref{supp:lemma:zeta_moments}, and \ref{supp:lemma:xi_estimates}, we have 
\begin{align*}
\|\boldsymbol{\hat\delta}-\boldsymbol{\delta}\|
\leq&
\|\mathbf{\hat{V}}^{-1}\|\|E[\boldsymbol{\xi}_3]\|+\|\mathbf{\hat{V}}^{-1}\|\|(E_N-E)[\boldsymbol{\xi}_3]\|
=
O_p\left(\sqrt{\frac{1}{N}}\right)
\\
\|(\boldsymbol{\hat\gamma}-\boldsymbol{\gamma})\boldsymbol{\Psi}_{1:L}\|
\leq&
\frac{1}{\pi_{h\setminus 0}}\|E\left[\boldsymbol{\xi}_2\right]\|\|\boldsymbol{\hat\Pi}^{-1}\|+\frac{1}{\pi_{h\setminus 0}}\|(E_N-E)\left[\boldsymbol{\xi}_2\right]\|\|\boldsymbol{\hat\Pi}^{-1}\|
\\&
+\frac{1}{\pi_{h\setminus 0}}\|\boldsymbol{\hat\delta}-\boldsymbol{\delta}\|E_N\left[\|(\mathbf{X}'\mathbf{X})^\ast\mathbf{X}'\mathbf{W}\|\|\mathbf{D}_{1:L}'\boldsymbol{\Psi}_{1:L}^{-1}\|\right]\|\boldsymbol{\hat\Pi}^{-1}\|
\\
=&
O(h_N^{L+1})+O_p\left(\sqrt{\frac{1}{N\pi_{h\setminus 0}}}\right)
\\
E_N[\|\boldsymbol{\hat\xi}_1-\boldsymbol{\xi}_1\|^2]^{1/2}
\leq&
\|E_N[\boldsymbol{\xi}_1]\|+E_N[\boldsymbol{1}\{D>h_N\}D^{-2}\|(\mathbf{X}'\mathbf{X})^\ast\mathbf{X}'\mathbf{W}\|^2]^{1/2}\|\boldsymbol{\hat\delta}-\boldsymbol{\delta}\|
=
O_p\left(\sqrt{\frac{\pi_h^2}{Nh_N^2\pi_{h\setminus 0}^2}}\right),
\\
E_N[\|\boldsymbol{\hat\xi}_2-\boldsymbol{\xi}_2\|^2]^{1/2}
\leq&
E_N[\|(\mathbf{X}'\mathbf{X})^\ast\mathbf{X}'\mathbf{W}\|^2\|\mathbf{D}_{1:L}'\boldsymbol{\Psi}_{1:L}^{-1}\|^2]^{1/2}\|\boldsymbol{\hat\delta}-\boldsymbol{\delta}\|+\|(\boldsymbol{\hat\gamma}-\boldsymbol{\gamma})\boldsymbol{\Psi}_{1:L}\|E_N[\|\mathbf{D}_{1:L}'\boldsymbol{\Psi}_{1:L}^{-1}\|^4]^{1/2}
\\
=&
o(h_N^{L+1})+O_p\left(\frac{1}{\sqrt{N}}\right),
\\
E_N[\|\boldsymbol{\hat\xi}_3-\boldsymbol{\xi}_3\|^2]^{1/2}
\leq&
E_N[\|\mathbf{W}\|^4\|\mathbf{M}_{\mathbf{X}}\|^2]^{1/2}\|\boldsymbol{\hat\delta}-\boldsymbol{\delta}\|
=
O_p\left(\frac{1}{\sqrt{N}}\right),\mbox{ and } 
\\
E_N[\|\boldsymbol{\hat\zeta}-\boldsymbol{\zeta}\|^2]^{1/2}
\leq&
E_N[\|\boldsymbol{\hat\xi}_1-\boldsymbol{\xi}_1\|^2]^{1/2}+\frac{1}{\pi_{h\setminus 0}}E_N[\|(\boldsymbol{\hat\xi}_2-\boldsymbol{\xi}_2)\|^2]^{1/2}\|\boldsymbol{\hat\Pi}^{-1}\|\|\boldsymbol{\Psi}_{1:L}^{-1}\mathbf{\hat{h}}\|
\\&
+\left\|\mathbf{\hat{Q}}\right\|\|\mathbf{\hat{V}}^{-1}\|E_N[\|\boldsymbol{\hat\xi}_3-\boldsymbol{\xi}_3\|^2]^{1/2}
\\&
+\frac{1}{\pi_{h\setminus 0}}\left(\|\boldsymbol{\hat\Pi}^{-1}-\boldsymbol{\Pi}^{-1}\|\|\boldsymbol{\Psi}_{1:L}^{-1}\mathbf{\hat{h}}\|+\|\boldsymbol{\Pi}^{-1}\|\|\boldsymbol{\Psi}_{1:L}^{-1}(\mathbf{\hat{h}}-\mathbf{h})\|\right)E_N[\|\boldsymbol{\xi}_2\|^2]^{1/2}
\\&
+\left(\left\|\mathbf{\hat{Q}}\right\|\|\mathbf{\hat{V}}^{-1}-\mathbf{{V}}^{-1}\|+\left\|\mathbf{\hat{Q}}-\mathbf{Q}\right\|\|\mathbf{{V}}^{-1}\|\right)E_N[\|\boldsymbol{\xi}_3\|^2]^{1/2}
\\
=&
o(1)+O_p\left(\sqrt{\frac{\pi_h^2}{Nh_N^2\pi_{h\setminus 0}^2}}\right).
\end{align*}
Therefore, the statement of the lemma follows.
\end{proof}

\end{document}